\documentclass[aps,prx,reprint,twocolumn,floatfix,superscriptaddress,notitlepage,usenames,dvipsnames,svgnames,table,nofootinbib,longbibliography]{revtex4-1}
\usepackage[utf8]{inputenc}
\usepackage[T1]{fontenc}
\usepackage{graphicx}
\usepackage{grffile}
\usepackage{wrapfig}
\usepackage{rotating}
\usepackage[normalem]{ulem}
\usepackage{amsmath}
\usepackage{textcomp}
\usepackage{amssymb}
\usepackage{capt-of}

\usepackage{parskip}
\usepackage{mathrsfs}
\usepackage[margin= 0.75in]{geometry}
\usepackage[braket, qm]{qcircuit}

\usepackage[english]{babel}
\usepackage{letltxmacro}
\LetLtxMacro{\ORIGselectlanguage}{\selectlanguage}
\makeatletter
\DeclareRobustCommand{\selectlanguage}[1]{%
  \@ifundefined{alias@\string#1}
    {\ORIGselectlanguage{#1}}
    {\begingroup\edef\x{\endgroup
      \noexpand\ORIGselectlanguage{\@nameuse{alias@#1}}}\x}%
}
\newcommand{\definelanguagealias}[2]{%
  \@namedef{alias@#1}{#2}%
}
\makeatother
\definelanguagealias{en}{english}

\usepackage{times}
\usepackage{bbm}
\usepackage{etoolbox}
\usepackage{tcolorbox} 
\usepackage{fleqn} 
\usepackage{tikz}
\usepackage{amsthm}
\usepackage{amssymb}
\usetikzlibrary{arrows.meta}
\usepackage{mathtools}

\newcommand{\prlsection}[1]{{\em {#1}.---~}}

\usepackage{bm}
\usepackage{dsfont}
\usepackage[font=small,labelfont=bf,justification=justified,format=plain]{caption}
\usepackage{subcaption}

\usepackage{thmtools}
\usepackage{thm-restate}

\definecolor{blue-violet}{rgb}{0.54, 0.17, 0.89}
\usepackage{hyperref}
\hypersetup{
 pdfauthor={The quantum coherent people},
 pdftitle={Quantum coherence as a signature of chaos},
 pdflang={English},colorlinks=true,linkcolor=RubineRed,citecolor=blue-violet,urlcolor=Cerulean} 
\usepackage[capitalise]{cleveref}

\begin{document}
\title{Quantum coherence as a signature of chaos}
\author{Namit Anand}
\email [e-mail: ]{namitana@usc.edu}
\affiliation{Department of Physics and Astronomy, and Center for Quantum Information Science and Technology, University of Southern California, Los Angeles, California 90089-0484, USA}

\author{Georgios Styliaris}
\email [e-mail: ]{georgios.styliaris@mpq.mpg.de}

\affiliation{Department of Physics and Astronomy, and Center for Quantum Information Science and Technology, University of Southern California, Los Angeles, California 90089-0484, USA}

\affiliation{Max-Planck-Institut f\"ur Quantenoptik, Hans-Kopfermann-Str. 1, 85748 Garching, Germany}

\affiliation{Munich Center for Quantum Science and Technology, Schellingstraße 4, 80799 M\"unchen, Germany}

\author{Meenu Kumari}
\email [e-mail: ]{mkumari@perimeterinstitute.ca}

\affiliation{Perimeter Institute for Theoretical Physics, Waterloo, ON N2L 2Y5, Canada}

\author{Paolo Zanardi}
\email [e-mail: ]{zanardi@usc.edu}

\affiliation{Department of Physics and Astronomy, and Center for Quantum Information Science and Technology, University of Southern California, Los Angeles, California 90089-0484, USA}

\date{\today}

\begin{abstract}
We establish a rigorous connection between quantum coherence and quantum chaos by employing coherence measures originating from the resource theory framework as a \emph{diagnostic} tool for quantum chaos. We quantify this connection at two different levels: quantum \emph{states} and quantum \emph{channels}. At the level of states, we show how several well-studied quantifiers of chaos are, in fact, quantum coherence measures in disguise (or closely related to them). We further this connection for \textit{all} quantum coherence measures by using tools from majorization theory. Then, we numerically study the coherence of chaotic-vs-integrable eigenstates and find excellent agreement with random matrix theory in the bulk of the spectrum. At the level of channels, we show that the coherence-generating power (CGP) --- a measure of how much coherence a dynamical process generates on average --- emerges as a \emph{subpart} of the out-of-time-ordered correlator (OTOC), a measure of information scrambling in many-body systems. Via numerical simulations of the (nonintegrable) transverse-field Ising model, we show that the OTOC and CGP capture quantum recurrences in quantitatively the same way. Moreover, using random matrix theory, we analytically characterize the OTOC-CGP connection for the Haar and Gaussian ensembles. In closing, we remark on how our coherence-based signatures of chaos relate to other diagnostics, namely the Loschmidt echo, OTOC, and the Spectral Form Factor.
\end{abstract}
\maketitle

\section{Introduction}
\label{sec:introduction}
Quantum coherence and quantum entanglement are arguably the two cardinal attributes of quantum theory, originating from the superposition principle and the tensor product structure (TPS), respectively~\cite{nielsen_quantum_2010, streltsovColloquiumQuantumCoherence2017, horodecki_quantum_2009}. While entanglement as a signature of quantum chaos has been well-studied in both the few- and many-body case~\cite{wang_entanglement_2004, RigolEE2017,kumari_untangling_2019, chaudhury2009nature, neillErgodicDynamicsThermalization2016}, a rigorous connection between quantum coherence and quantum chaos still remains elusive. Here, we clarify in a quantitative way the role that quantum coherence plays in the study of chaotic quantum systems. Apart from the foundational role that the superposition principle plays in ``everything quantum,'' there are (at least) two distinct ways in which quantum coherence enters the study of quantum chaotic systems. The first, and perhaps the more conceptual one, is the Eigenstate Thermalization Hypothesis (ETH)~\cite{srednickiChaosQuantumThermalization1994, deutsch_quantum_1991, rigol_thermalization_2008} and the \emph{diagonal ensemble} associated with it. The notion of quantum coherence is a \emph{basis-dependent} one and the diagonal ensemble reveals the Hamiltonian eigenbasis as the relevant physical basis, especially when studying thermalization, ergodicity, and other temporal characteristics. Moreover, an initial state's overlap with sufficiently many energy-levels --- which is related to coherence in the energy-eigenbasis --- is a sufficient condition for equilibration (under some additional assumptions)~\cite{reimann_foundation_2008, linden_quantum_2009, shortEquilibrationQuantumSystems2011}. Second, the out-of-time-ordered correlator (OTOC)~\cite{larkin_quasiclassical_1969, kitaev_simple_2015} a quantifier of quantum chaos~\footnote{The precise role of the OTOC in characterizing chaoticity is nuanced and we refer the reader to \cref{sec:otoc-intro} and Refs.~\cite{PhysRevB.98.134303,PhysRevLett.123.160401,luitz2017information,PhysRevE.101.010202,PhysRevLett.124.140602,hashimoto2020exponential,wang2020quantum} for a detailed discussion.} and information scrambling, is usually studied via the input of two \emph{local} unitaries and grows when they start noncommuting as one of them spreads under the Heisenberg time evolution. The locality of the observables in the OTOC ``probes'' the entanglement structure and its growth~\cite{lashkari_towards_2013, kitaev_simple_2015}. At the same time, it is natural to ask, what does the \emph{strength} of the noncommutativity probe (without reference to any TPS)? We argue that this is precisely a measure of quantum coherence (more specifically, the incompatibility of the bases associated to the unitaries~\cite{styliaris_quantifying_2019-1}). For example, given two (non-degenerate) observables \(A,B\) and the associated eigenbases \(\mathbb{B}_{A}, \mathbb{B}_{B}\), we can ask, how coherent are the eigenstates of \(A\) when expressed in the (eigen)basis \(\mathbb{B}_{B}\). Clearly, if \(\left[ A,B \right] =0\) then the eigenstates of \(A\) are incoherent in \(\mathbb{B}_{B}\). On the other hand, if \(\mathbb{B}_{A}\) and \(\mathbb{B}_{B}\) are \emph{mutually unbiased}, then the eigenstates of \(A\) are \emph{maximally coherent} in \(\mathbb{B}_{B}\), and various measures of incompatibility are maximized~\cite{styliaris_quantifying_2019-1}. Following this intuition, we will show that the OTOC is intimately related to a measure of incompatibility called the coherence-generating power (CGP), as exemplified by our \autoref{thm:otoc-cgp-connection}.  

\emph{Quantifying chaos}.--- Signatures of quantum chaos can be
broadly classified into three categories: (i) spectral properties,
such as level-spacing
distribution~\cite{bohigas_characterization_1984, haake_quantum_2010},
level number variance~\cite{guhr_random-matrix_1998}, etc., (ii)
eigenstate structure, such as eigenstate entanglement (defined as the
average entanglement entropy over \emph{all} eigenstates) and the
associated area and volume laws~\cite{eisertColloquiumAreaLaws2010},
and (iii) dynamical quantities such as Loschmidt echo~\cite{peres1984stability,PhysRevLett.86.2490,Goussev:2012,Gorin2006LEreview}, entangling power~\cite{zanardiEntanglingPowerQuantum2000,zanardi_entanglement_2001,wang_entanglement_2004,scott_entangling_2003,PhysRevE.64.036207}, quantum
discord~\cite{madhok_signatures_2015}, OTOCs, etc. (see also
Ref.~\cite{haake_quantum_2010} for other examples), which, in general are a property of both the eigenvalues and eigenvectors of the Hamiltonian. In this paper, we connect quantum chaos and quantum coherence in the sense of (ii) and (iii), by examining the coherence structure of chaotic-vs-integrable eigenstates, and by studying the coherence-generating power of chaotic dynamics.

\prlsection{Outline} This paper is organized as follows. In \cref{sec:Preliminaries}, we review the resource theory of quantum coherence and the coherence measures that will be used throughout this paper. In \cref{subsec:why-coherence}, we discuss connections between coherence measures and delocalization measures, first via examples, and then via the mathematical formalism of majorization. We also discuss the connection between coherence and entanglement and how their interplay affects coherence measures' ability to diagnose quantum chaos. In \cref{sec:xxz-numerics}, we numerically examine the coherence structure of integrable-vs-chaotic eigenstates and introduce new tools inspired from majorization theory to study quantum chaos. In \cref{sec:operations}, we establish the connection between OTOC and CGP and, in particular, show how the CGP emerges as a \textit{subpart} of the OTOC. Then, in \cref{sec:rmt-and-short-time}, using tools from random matrix theory, we analytically perform averages over the CGP of random Hamiltonians and unveil a connection between CGP and the spectral form factor. We also study the short-time growth of the CGP and remark on its connection with quantum fluctuations and the resource theory of incompatibility. Furthermore, in \cref{sec:recurrences}, we numerically vindicate our OTOC-CGP connection by studying the integrable and chaotic regimes in a transverse-field Ising model. Finally, in \cref{sec:discussion}, we make some closing remarks and discuss our results. Our main results are stated as Theorems and all proofs can be found in the \cref{sec:app:proofs}.  

\section{Preliminaries.}
\label{sec:Preliminaries}
\emph{Resource theory of quantum coherence.}--- Despite the
fundamental role that quantum coherence plays in quantum theory, a
rigorous quantification of coherence as a \emph{physical resource} was
only initiated in recent
years~\cite{baumgratzQuantifyingCoherence2014,
  abergQuantifyingSuperposition2006,
  streltsovColloquiumQuantumCoherence2017}. We briefly review the
resource theory of coherence and the quantification tools it
provides. Let \(\mathcal{H} \cong \mathbb{C}^{d}\) be the Hilbert
space associated to a \(d\)-dimensional quantum system and
\(S(\mathcal{H})\) the set of all quantum states. Quantum coherence of
states is quantified
with respect to a \emph{preferred} orthonormal basis for the Hilbert
space, \(\mathbb{B} = \{ | j \rangle \}_{j=1}^d\). All states that are
diagonal in the basis \(\mathbb{B}\) are deemed \emph{incoherent}
(that is, devoid of any resource) while others \emph{coherent}. That is,
incoherent states have the form, \(\rho = \sum\nolimits_{j=1}^{d}
p_{j} \Pi_{j}\), where \(\Pi_{j} \equiv | j \rangle \langle  j |\) is
the rank-\(1\) projector associated to the basis state \(| j \rangle\)
and \(p_{j} \geq 0, \sum\nolimits_{j=1}^{d} p_{j} = 1\) is a
probability distribution. The collection of all incoherent states
forms a convex set, \(\mathcal{I}_{\mathbb{B}}\) (usually called the
``free states'' of the resource theory)~\footnote{We remark that to
  quantify coherence, indeed a weaker notion than that of a basis is
  required, which takes into account the freedom in choosing arbitrary
  global phases and orderings for the basis elements.}. A common
quantifier of the amount of resource in a state \(\sigma\) is to
measure its (minimum) distance from the set
\(\mathcal{I}_{\mathbb{B}}\), using appropriately chosen distance
measures, say \(\mathcal{R}_{\mathbf{d}}(\sigma) \coloneqq
\min_{\delta \in \mathcal{I}_{\mathbb{B}}} \mathbf{d}(\sigma,
\delta)\). where \(\mathbf{d}(\cdot, \cdot)\) is a distance measure on
the state space and \(\mathcal{R}_{\mathtt{d}}\) its associated
resource quantifier (usually called the ``resource measures'' of the
resource theory). The coherence quantifiers that we will be working with
in this paper are the \(l_2\)-norm of coherence~\footnote{Note that
  although the $2$-coherence is a monotone for all
  unital channels (which includes unitary evolution), it is not monotonic under the full set of incoherent
  operations IO (introduced later)~\cite{baumgratzQuantifyingCoherence2014}. However, this is not a problem
  since we are only concerned with unitary evolutions in this work.} (hereafter \(2\)-coherence) and the relative entropy of coherence, defined as~\cite{baumgratzQuantifyingCoherence2014},
\begin{align}
\label{eq:p-norm-and-rel-entropy-coherence-defn}
&\mathtt{c}^{(\text{2})}_{\mathbb{B}}(\rho) \coloneqq \min_{\sigma \in \mathcal{I}_{\mathbb{B}}} \left\Vert \rho - \sigma \right\Vert_{l_{2}}^{2} = \left\Vert \rho - \mathcal{D}_{\mathbb{B}}(\rho)\right\Vert_{l_{2}}^{2}, \\ \label{eq:rel-entropy-coherence-defn}
&\mathtt{c}^{(\text{rel})}_{\mathbb{B}}(\rho) \coloneqq \min_{\sigma \in \mathcal{I}_{\mathbb{B}}} S^{(\text{rel})}(\rho || \sigma)  = S(\mathcal{D}_{\mathbb{B}}(\rho)) - S(\rho), 
\end{align}
where, \(\mathcal{D}_{\mathbb{B}}(X) \coloneqq \sum_{j=1}^{d} \Pi_j X
\Pi_j\) is the dephasing superoperator, \(S^{(\text{rel})}(\rho ||
\sigma)\) is the quantum relative entropy, and \(S(\rho)\) is the von
Neumann entropy~\cite{baumgratzQuantifyingCoherence2014}. The \(2\)-coherence~\footnote{For the purposes of computing the $2$-coherence, recall that the $l_2$-norm of a matrix is equal to its Hilbert-Schmidt norm.} has
been identified as the escape probability, a key figure of merit for
few- and many-body localization~\cite{styliaris_quantum_2019-1}, while
the relative entropy of coherence has several operational
interpretations, prominent amongst which are its role as the
distillable coherence~\cite{winterOperationalResourceTheory2016} and as a measure of deviations from thermal equilibrium~\cite{2013arXiv13081245R}.

A final but key ingredient of quantum resource theories are the
so-called ``free operations,'' transformations that do not generate
any resource, but may consume it. For the resource theory of
coherence, we will focus on the class of \emph{incoherent
  operations} (IO): completely-positive (CP) maps such that there
exists at least one Kraus representation which satisfies \(K_{j} \rho
K_{j}^{\dagger}/ \operatorname{Tr}\left( K_{j} \rho K_{j}^{\dagger}
\right) \in \mathcal{I}_{\mathbb{B}} ~~\forall \rho \in
\mathcal{I}_{\mathbb{B}},~~\forall j\)~\footnote{One can also think of
  them as generalized measurements instead, since that requires a
  specific Kraus representation~\cite{nielsen_quantum_2010}}. Resource measures that are \textit{non-increasing} under the action of free operations are called resource \textit{monotones}.

\section{At the level of states}
\label{sec:level-of-states}

\subsection{Why study quantum coherence?}
\label{subsec:why-coherence}
The sudden delocalization of chaotic systems following a quench has been well-studied for both classical and quantum systems, see Refs. \cite{dalessio_quantum_2016,Borgonovi2016} and the references therein. Various quantifiers of this delocalization have been introduced in the quantum chaos literature to characterize integrable and chaotic quantum systems. Here, we argue that many of these delocalization measures are nothing but quantum coherence measures in disguise. We argue this in two ways: first, we consider some paradigmatic measures of delocalization such as Shannon entropy, participation ratio, etc.,~\cite{kota_embedded_2014} and connect them with measures of quantum coherence studied in the resource theories framework. Moreover, this also reveals that the notion of delocalization in the available phase space, energy space, etc., is precisely the notion of quantum coherence in an appropriate basis. Second, we show that the notion of when one state is more \textit{delocalized} than the other (and measures to quantify them) is captured in a very general way by the mathematical formalism of majorization. This further allows us to make a precise connection to the resource theory of coherence since state transformation under incoherent operations is completely characterized in terms of majorization. Finally, using the majorization result from the resource theoretic framework of coherence, we argue that quantum coherence measures capture precisely what delocalization measures set out to quantify: how ``localized'' or ``uniformly spread'' a quantum state is across a basis. Along the way we also remark on coherence measures' ability to probe entanglement measures, which have long been used as quantifiers of chaos.

\prlsection{Connection with delocalization measures} Let us start with a simple example: Given a state \(|
\psi \rangle\) expressed in some basis \(\mathbb{B} = \{ | j \rangle
\}\), \(| \psi \rangle = \sum\nolimits_{j=1}^{d} c_{j} |
j \rangle\), one can consider various ways to quantify how \textit{uniformly spread} the probability distribution generated from \(\{ \left| c_{j} \right|^{2}
\}\) is. For instance, an incoherent state \(| j \rangle\) corresponds
to the (extremely) nonuniform probability distribution \(\mathbf{p}_{|j\rangle} = \{ 1, 0, \cdots, 0 \}\), that is, it is the most ``localized'' state; while a highly coherent state~\footnote{In fact, this family of states are maximally coherent in the resource theory of coherence with incoherent operations; analogous to how Bell states are maximally entangled in the resource theory of pure bipartite entanglement.} of the form \(| \psi \rangle =
\frac{1}{\sqrt{d}} \sum\limits_{j=1}^{d} e^{-i \theta_{j}} | j
\rangle\) corresponds to the uniform probability distribution \(\mathbf{p}_{|\psi\rangle} = \{ \frac{1}{d},
\frac{1}{d}, \cdots, \frac{1}{d} \}\), that is, it is maximally ``delocalized''. Therefore, if we quantify the uniformity of the associated probability distributions by
evaluating, for example, their Shannon entropy, we see that the
incoherent state corresponds to the minimum entropy \(S(\{ \left| c_{\alpha} \right|^{2}
\}) = 0\), while the highly coherent state maximizes the Shannon
entropy, \(S(\{ \left| c_{\alpha} \right|^{2}
\}) = \log \left( d \right)\). This uniformity is precisely what coherence measures and delocalization measures quantify. 

We now discuss some examples where there is a precise connection between them. We consider the same notation as above, a pure state $| \psi \rangle$, a basis \(\mathbb{B} = \{ | j \rangle \}\), and \(\{ p_{j} \}_{j=1}^{d}\), where \(p_{j} \equiv \left| \left\langle k| \psi \right\rangle \right|^{2}\) is the associated probability distribution.

1. The Shannon entropy (also known as the informational entropy in the quantum chaos literature) of the probability distribution \(\{ p_{j} \}_{j=1}^{d}\) has been used as a measure of delocalization~\cite{kota_embedded_2014,dalessio_quantum_2016,Borgonovi2016}. We note that for pure states, this is equal to the relative entropy of coherence. That is,
\begin{align}
    \mathtt{c}^{(\text{rel})}_{\mathbb{B}}(\rho) = S(\{ p_{j} \})
\end{align}
This follows from the definition in~\cref{eq:p-norm-and-rel-entropy-coherence-defn} and the fact that the Shannon entropy of pure states is zero, that is, \(S(| \psi \rangle \langle  \psi | ) = 0\). It is worth noting that the Shannon entropy is the first R\'enyi entropy~\cite{renyi1961measures}, a family of entropies which provide powerful connections with majorization theory and state transformation in resource theories~\cite{chitambar_quantum_2019}.

2. The second participation ratio (also known as the number of principal components)~\cite{kota_embedded_2014,dalessio_quantum_2016,Borgonovi2016}, defined as 
\begin{align}
  \mathrm{PR}_{2,\mathbb{B}}(| \psi \rangle) \coloneqq \sum_{j}|\langle j |
  \psi\rangle|^{4}.
\end{align}
Note that for pure states and any given basis $\mathbb{B}$, the $\mathrm{PR}_{2,\mathbb{B}}$ is equal to one minus the $2$-coherence, that is~\footnote{A proof of this follows immediately by expanding the formula for $2$-coherence of pure states, ${\mathtt{c}}_{\mathbb{B}}^{(2)}(\rho)=1-\left\langle\rho, \mathcal{D}_{\mathbb{B}}(\rho)\right\rangle$.},
\begin{align}
\label{eq:pr2-c2b}
  \mathrm{PR}_{2,\mathbb{B}}(| \psi \rangle) = 1-
  \mathtt{c}^{(2)}_{\mathbb{B}}(| \psi \rangle \langle  \psi | ).
\end{align}

Moreover, the negative logarithm of \(\mathrm{PR}_{2}\) is equal to the second R\'enyi entropy~\cite{renyi1961measures} of the probability distribution \(\{ p_{j} \}\). And both the first and second R\'enyi entropies are measures of quantum coherence~\cite{streltsovColloquiumQuantumCoherence2017}.

3. We now review three quantities, the Loschmidt echo, the escape probability and the effective dimension, which find a multitude of applications in quantum chaos, thermalization, and localization. The Loschmidt echo is defined as the overlap between the initial state $\ket{\psi}$ and the state after time $t$~\cite{peres1984stability,PhysRevLett.86.2490,Gorin2006LEreview},
\begin{align}
\mathcal{L}_{t} (|\psi\rangle) \coloneqq \left| \langle \psi | e^{-iHt} |  \psi \rangle \right|^{2}.
\end{align}
The effective dimension of a quantum state is defined
as its inverse purity~\cite{reimann_foundation_2008, linden_quantum_2009},
\begin{align}
d^{\mathrm{eff}}(\rho) = \frac{1}{\mathrm{Tr}[\rho^2]},
\end{align}
which intuitively corresponds to the number of pure states that
contribute to the (in general) mixed state $\rho$. In
Refs.~\cite{reimann_foundation_2008, linden_quantum_2009},
\(d^{\mathrm{eff}}(\rho)\) was used to provide a sufficient condition for equilibration
in closed quantum systems. And finally, we recall that the infinite-time average of a quantity \(A\) is defined
as 
\begin{align}
    \overline{A}
\coloneqq \lim_{T \rightarrow \infty} \frac{1}{T} \int\limits_{0}^{T} A(t) dt.
\end{align}

Infinite-time averaging connects these various quantities as follows (with $\rho = |\psi \rangle \langle \psi|$)
\begin{align}
\overline{\mathcal{L}_{t}(|\psi\rangle)}= \textnormal{PR}_{2,\mathbb{B}_{\mathrm{H}}}(\rho)= \frac{1}{d^{\mathrm{eff}}(\overline{\rho})} = 1- \mathcal{P}_{\psi},
\end{align}
where $\mathbb{B}_{\mathrm{H}}$ is the Hamiltonian eigenbasis and $\mathcal{P}_{\psi}:=1-\overline{\left|\left\langle\psi\left|e^{-i H t}\right| \psi\right\rangle\right|^{2}}$ is the escape probability of the state $|\psi\rangle$; which using Proposition 4 of Ref.~\cite{styliaris_quantifying_2019-1} is also equal to the $2$-coherence in the Hamiltonian eigenbasis.

Note that, the proof of
Proposition 4 in Ref.~\cite{styliaris_quantifying_2019-1} can
potentially reveal many more connections since there it was observed
that the infinite time-average of the time evolution operator (for a non-degenerate Hamiltonian) $\mathcal{U}_{t} \coloneqq U (\cdot) U^\dagger$ is equivalent to dephasing in the
Hamiltonian eigenbasis, that is,
\(\overline{\mathcal{U}_{t}}=\mathcal{D}_{\mathbb{B}_{\mathrm{H}}}\). The action of
$\mathcal{D}_{\mathbb{B}_{\mathrm{H}}}$ reveals the ``diagonal ensemble,''
fundamental to the study of thermalization in closed quantum
systems~\cite{rigol_thermalization_2008}.

\prlsection{Arbitrary coherence measures and majorization} Given two vectors $\vec{v},\vec{w} \in \mathbb{R}^{n}$, we say that ``$\vec{v}$ is majorized by $\vec{w}$,'' (equivalently $\vec{w}$ majorizes $\vec{v}$) written as $\vec{v} \prec \vec{w}$, if~\cite{marshall_inequalities_2011} 
\begin{align}
\begin{aligned}
    \sum_{j=1}^{k} v_{[j]} &\leq \sum_{j=1}^{k} w_{[j]}, \forall k = 1, \cdots, n-1 \\
    \sum_{j=1}^{n} v_{[j]} &= \sum_{j=1}^{n} w_{[j]},
\end{aligned}    
\end{align}

where $v_{[j]}$ is the $j$th element of $\vec{v}$ when sorted in a nonincreasing order. Majorization induces a preorder~\footnote{A preorder is a binary relation that is reflexive and transitive but not necessarily antisymmetric.} on the vectors in $\mathbb{R}^{n}$ and it is natural to ask what functions preserve this preorder? All functions $f: \mathbb{R}^n \rightarrow \mathbb{R}$ such that $\vec{v} \prec \vec{w} \implies f(\vec{v}) \leq f(\vec{w})$ are called Schur-convex (equivalently, Schur-concave if $\vec{v} \prec \vec{w} \implies f(\vec{v}) \geq f(\vec{w})$). Many functionals employed in the study of quantum chaos like Shannon entropy, the family of R\'enyi entropies, and others, are an example of Schur-concave functions that preserve the ordering imposed from majorization. Using a theorem of Hardy-Littlewood-Polya~\cite{marshall_inequalities_2011}, we have the following
\begin{align}
    \vec{v} \prec \vec{w} \iff \sum_{j=1}^{n} g\left(v_{j}\right) \leq \sum_{j=1}^{n} g\left(w_{j}\right),
\end{align}
for all continuous convex functions $g: \mathbb{R} \rightarrow \mathbb{R}$.

That is, studying majorization is equivalent to studying the ordering induced from \textit{all} continuous convex functions obeying an ordering. It is in this specific sense that majorization allows us to go beyond any specific quantum coherence measure and allows us to discuss the behavior of \textit{all} coherence measures.

To make the connection to quantum coherence, we note that given two states $\rho, \sigma$ and a coherence measure \(\mathtt{c}_{\mathbb{B}}(\cdot)\), if \(\mathtt{c}_{\mathbb{B}}(\rho) > \mathtt{c}_{\mathbb{B}}(\sigma)\) then \(\sigma\) cannot be transformed into \(\rho\) via incoherent operations (since IO can only \textit{nonincrease} the amount of coherence in a state). On the other hand, \(\mathtt{c}_{\mathbb{B}}(\rho) \leq \mathtt{c}_{\mathbb{B}}(\sigma)\) provides a necessary (but not sufficient) condition on the state transformation \(\sigma \mapsto \rho\) using IO. A necessary and sufficient condition was obtained in Ref.~\cite{du_conditions_2015} in terms of majorization (the theorem has been rephrased for simplicity). In the following, $\mathcal{D}_{\mathbb{B}}$ is the dephasing superoperator in the basis $\mathbb{B}$; and the notion of matrix majorization has been used, with $A \prec B \iff \mathrm{spec}(A) \prec \mathrm{spec}(B)$, where $\mathrm{spec}(A)$ is the vector of eigenvalues of $A$.\\

\begin{restatable}[\cite{du_conditions_2015}]{thm}{coherencemajorization}
\label{thm:coherence-transformation-majorization}
A quantum state \( | \psi \rangle \) can be
transformed to another state \(| \phi \rangle\) via
incoherent operations if and only if \(\mathcal{D}_{\mathbb{B}} \left( | \psi \rangle \langle  \psi |  \right) \prec \mathcal{D}_{\mathbb{B}} \left( | \phi \rangle \langle  \phi |  \right)\).
\end{restatable}

\emph{Remark:} First, note that $\mathcal{D}_{\mathbb{B}} \left( |
  \psi \rangle \langle  \psi |  \right) \equiv
\mathbf{p}_{\psi}$ is isomorphic to the probability vector
obtained from the state $|\psi\rangle$ expressed in the basis
$\mathbb{B}$. Therefore, the condition \(\mathcal{D}_{\mathbb{B}}
\left( | \psi \rangle \langle  \psi |  \right) \prec
\mathcal{D}_{\mathbb{B}} \left( | \phi \rangle \langle  \phi |
\right) = \mathbf{p}_{\psi} \prec
\mathbf{p}_{\phi}\), that is, it is equivalent to the state
$|\psi\rangle$ being more uniformly spread in the basis $\mathbb{B}$
than the state $|\phi\rangle$, in the sense of majorization. Now,
since the majorization condition is equivalent~\footnote{The condition is only sufficient but becomes necessary for the generic case of full-rank pure states (which can be obtained by an arbitrarily small perturbation) and holds true for physically relevant scenarios.} to transforming \(|
\psi \rangle \mapsto | \phi \rangle\) via an incoherent operation, the
amount of coherence in \(| \psi \rangle\) is greater than or equal to the amount
of coherence in \(| \phi \rangle\), for \textit{every} quantum
coherence measure. Formally, \( \mathcal{R}_{c} \left(  | \psi \rangle
  \langle  \psi |  \right) \geq \mathcal{R}_{c} \left(  | \phi \rangle
  \langle  \phi |  \right) \), for every coherence monotone \(
\mathcal{R}_{c}: \mathcal{S}(\mathcal{H}) \rightarrow \mathbb{R}_0^+
\). Therefore, quantum coherence measures capture in a precise sense
what traditional delocalization measures set out to quantify: how
uniformly spread is a quantum state with respect to a basis
\(\mathbb{B}\); in fact, the above theorem quantitatively shows that
these two notions are equivalent.

Having established a web of connections between several key quantities used in the study of quantum chaos and equilibration, we now discuss how quantum coherence measures can inherit their ability to diagnose quantum chaos from their interplay with entanglement measures.

\begin{figure*}[!ht]
   \raggedright
\begin{subfigure}{.37\textwidth}
  \includegraphics[width=1.3\linewidth]{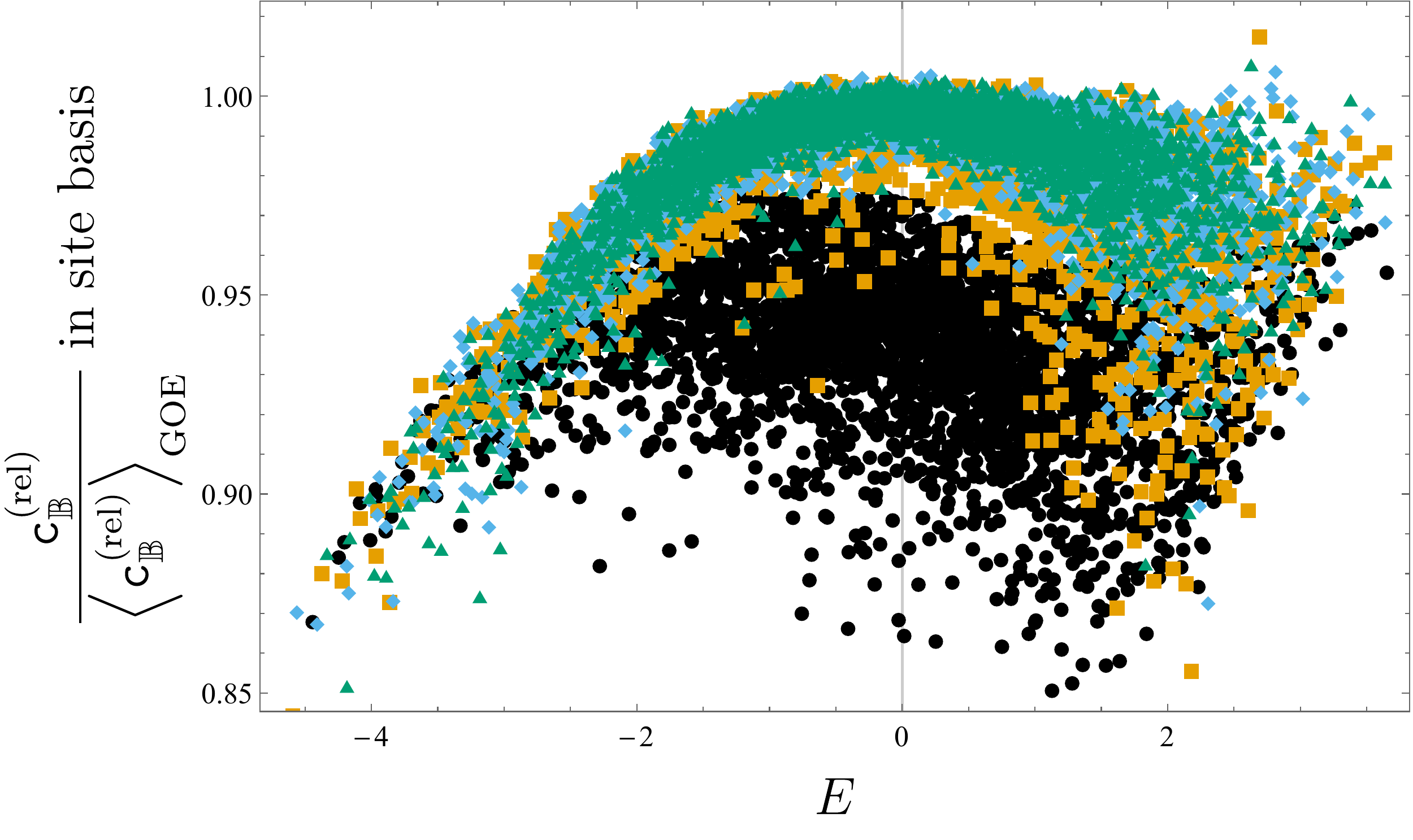}
  \caption{}
\end{subfigure}\hspace{65pt}
\begin{subfigure}{.39\textwidth}
  \includegraphics[width=1.3\linewidth]{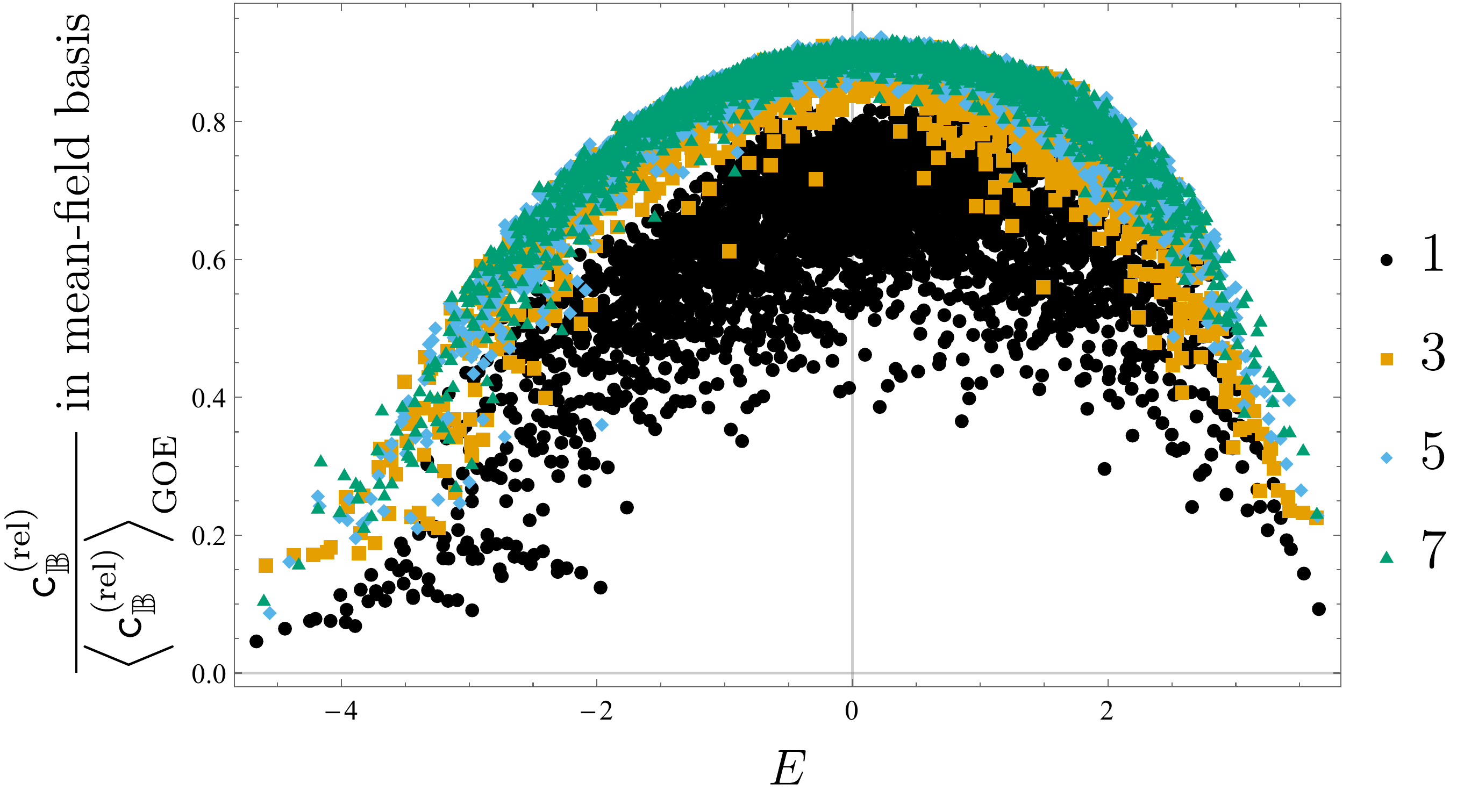}
    \caption{}
\end{subfigure}%
\caption{Relative entropy of coherence for eigenstates of the Hamiltonian defined in \cref{eq:xxz-defect-defn} as a function of their energy, normalized with the GOE prediction using \cref{eq:rel-entropy-goe-formula}, $\left\langle \mathtt{c}^{(\text{rel})}_{\mathbb{B}} \right\rangle_{\mathrm{GOE}} \approx 10.49$. Results are reported for \(L=15\) with \(5\) spins up and \(\omega=0\),  \(\epsilon_{\delta} = 0.5\), \(J_{xy} = 1\), \(J_{z} = 0.5\). The plot markers \(1,3,5,7\) correspond to the various choices of the defect site, with \(\delta = 1\) and \(\delta = 7\) corresponding to the integrable and chaotic limits, respectively. Figures (a) and (b) correspond to the two different bases, the site-basis and mean-field basis, respectively.}
\label{fig:goe-normalized-rel-coherence}
\end{figure*}

\textit{Coherence and its interplay with entanglement}.--- The study
of quantum coherence per se, makes no reference to the
\textit{locality} (or TPS) of a quantum system. However, many-body systems are
often endowed with a natural TPS and to study the interplay between
coherence and entanglement, it is often convenient to choose
incoherent states that are \emph{compatible} with the TPS, namely, the
incoherent states are also product
states~\cite{streltsovMeasuringQuantumCoherence2015,chitambar_relating_2016}. Consider,
for example, a two-qubit system, \(\mathcal{H} \cong \mathbb{C}^{2}
\otimes \mathbb{C}^{2}\), with an incoherent basis \(\mathbb{B} = \{ |
00 \rangle, | 01 \rangle, | 10 \rangle, | 11 \rangle \}\) that is also
separable~\footnote{An example of ``incompatible'' quantum coherence
  would be, for instance, if the incoherent basis for a $2$-qubit
  system is chosen to be the Bell-basis.}. Then, notice that any
entangled state is automatically coherent, since \(| \Psi_{AB}
\rangle\) is entangled if and only if \(| \Psi_{AB} \rangle \neq |
\phi \rangle_{A} \otimes | \phi \rangle_{B}\) for any \(| \phi
\rangle_{A(B)} \in \mathcal{H}_{A(B)}\). Therefore, when expressed as a
linear combination of the basis elements in \(\mathbb{B}\), we note that, for
every entangled state, \(| \Psi_{AB} \rangle =
\sum\nolimits_{j,k=0}^{1} c_{jk} | j\rangle_{A} |k \rangle_{B}\), we have \emph{at
  least} two non-zero coefficients \(c_{jk}\) --- that is, they are coherent
as well. Clearly, not every coherent state is entangled, for example, consider the state
\(| 0 \rangle \otimes | + \rangle\). This construction can be
generalized to the (simplest) multipartite~\footnote{In general, multipartite entanglement is much richer and less tractable than bipartite entanglement and that is why we consider the simplest scenario here~\cite{horodecki_quantum_2009}.} case as follows: Let \(\mathcal{H}
\cong \mathcal{H}_{1} \otimes \mathcal{H}_{2} \otimes \cdots
\mathcal{H}_{n}\) be a \(n\)-partite Hilbert space with \(\mathcal{F}_{e}\) being
the set of fully separable states (that is, they are convex combinations of
states that factorize over any tensor factor) and \(\mathcal{F}_{c}\) being
the set of incoherent states that are also fully separable. Then, it is easy
to see that \(\mathcal{F}_{c} \subset \mathcal{F}_{e}\) (since the
\(\mathcal{F}_{c}\) is \emph{compatible} with the TPS). As an
immediate consequence, note that, if \(\mathcal{R}(\cdot, \cdot):
\mathcal{S}(\mathcal{H}) \rightarrow \mathbb{R}_{0}^{+}\) is a
contractive distance (under the associated free operations, that leave the set of free states invariant), then, one can define a
``distance-based measure,'' \(\mathcal{R}_{\alpha}(\rho) \coloneqq
\min_{\sigma \in \mathcal{F}_{\alpha}} \mathcal{R}(\rho, \sigma) \),
where \(\alpha = \{ c,e \}\). Then, using the set inclusion of
\(\mathcal{F}_{c}, \mathcal{F}_{e}\), we have, \(~~\forall \rho, \mathcal{R}_{c}(\rho) \geq \mathcal{R}_{e}(\rho)\), that is, the amount of coherence is lower-bounded by the entanglement; or, the amount of coherence is an upper bound on the amount of entanglement~\footnote{This construction holds not only for contractive distances but the general class of functionals called gauge functions~\cite{regulaConvexGeometryQuantum2018}.}.

In light of the above observation, it is worth noting that there is a
semantical issue in calling these functionals \emph{delocalization}
measures since there is, per se, no \emph{locality} in their
definition. At this point, it is more appropriate to think of them as
quantifying the coherence of a state in some basis,
\(\mathbb{B}\); in fact, their definition reveals that this is precisely what they do. To connect quantum coherence with entanglement in a quantitative way (apart from the bounds realized from the
discussion above), as a first step, one needs to define a quantity that removes the basis-dependence of coherence (since entanglement is \textit{basis-independent}), which can be obtained by optimizing over various choices
of bases. Here, we prove one such result by minimizing the amount of coherence over all local bases: Given pure states in \(\mathcal{H} \cong \mathcal{H}_{a} \otimes
\mathcal{H}_{b}\), we have,\\

\begin{restatable}{thm}{coherenceentropy}
\label{prop:2coherence-entropy-connection}
\begin{align}
\min_{\mathbb{B}_{a},\mathbb{B}_{b}} \mathtt{c}^{(\text{2})}_{\mathbb{B}_{a} \otimes \mathbb{B}_{b}} (| \Psi \rangle \langle  \Psi | ) = 1 - \left\Vert \rho_{a} \right\Vert_{2}^{2} =: S_{\mathrm{lin}}(\rho_{a}),
\end{align}
where \(\rho_{a} = \operatorname{Tr}_{b}\left( | \Psi \rangle \langle
  \Psi |   \right)\) is the reduced density matrix and
\(S_{\mathrm{lin}}(\cdot)\) is the linear entropy, a quantifier of
entanglement.
\end{restatable}

That is, by minimizing the amount of coherence over all
local bases, we can (indirectly) compute a measure of entanglement. Another quantitative connection was obtained in Ref.~\cite{streltsovMaximalCoherenceResource2018}, where, by maximizing the amount of coherence over all bases, the amount of coherence in a state was connected with its purity. In summary, quantum coherence measures provide both upper bounds and in some cases precise connections with entanglement measures. Since entanglement measures have been widely used to detect quantum chaos, we now turn to studying quantum coherence in chaotic systems.

\subsection{Coherence of many-body eigenstates: XXZ spin-chain with defect}
\label{sec:xxz-numerics}

\begin{figure*}[!t]
   \raggedright
\begin{subfigure}{.37\textwidth}
  \includegraphics[width=1.3\linewidth]{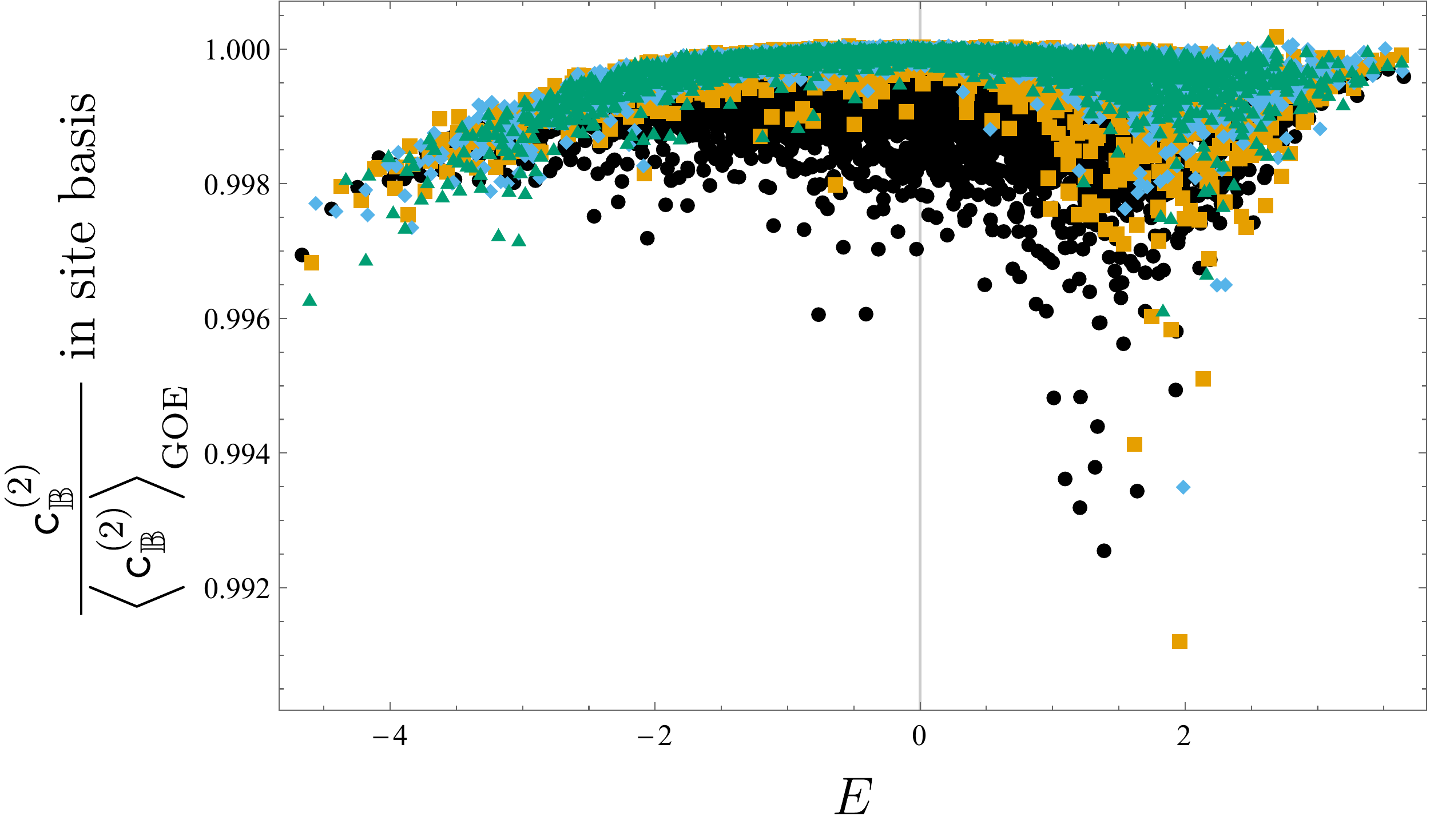}
  \caption{}
\end{subfigure}\hspace{65pt}
\begin{subfigure}{.39\textwidth}
  \includegraphics[width=1.3\linewidth]{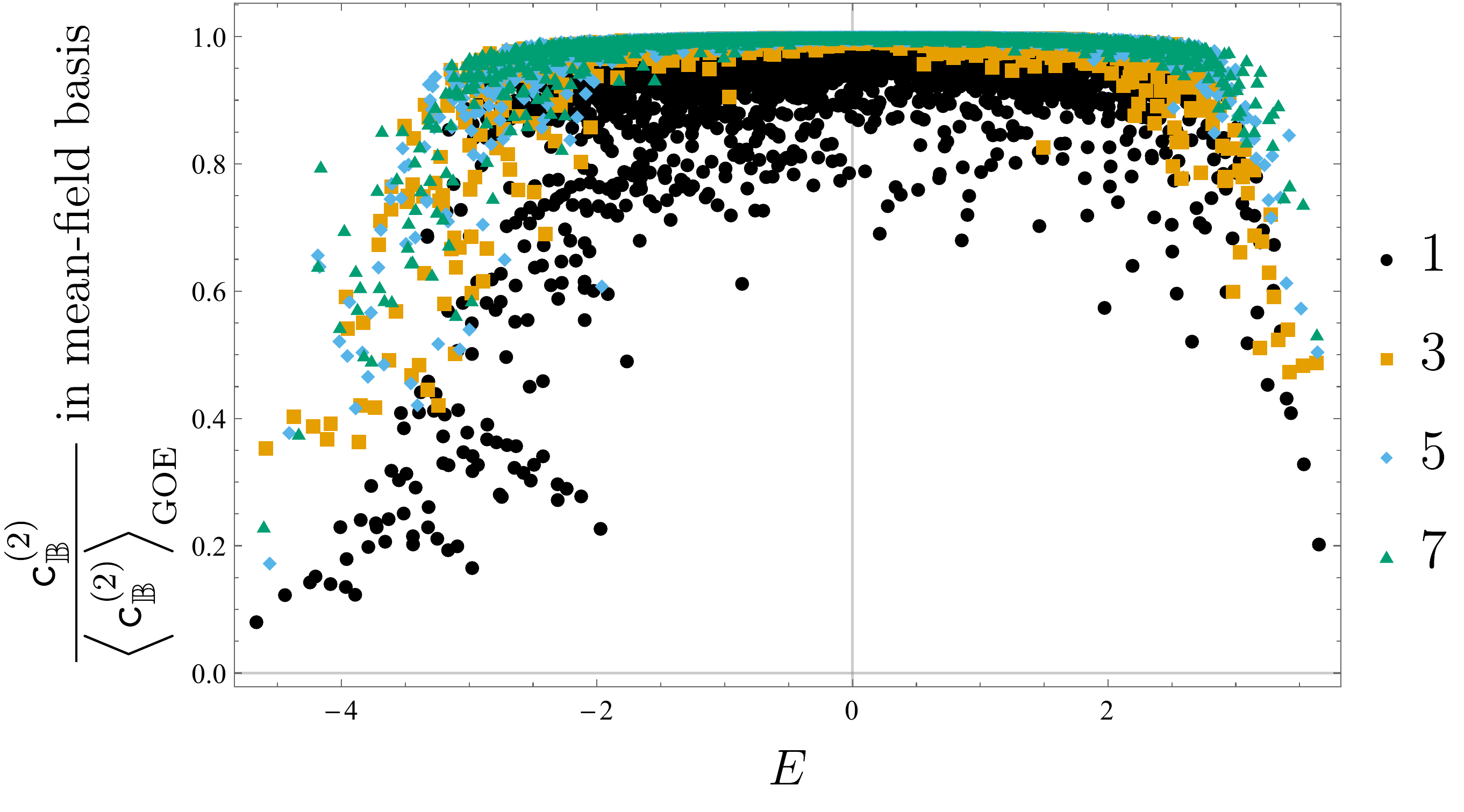}
    \caption{}
\end{subfigure}%
\caption{$2$-coherence for eigenstates of the Hamiltonian defined in \cref{eq:xxz-defect-defn} as a function of their energy, normalized with the GOE prediction obtained numerically, $\left\langle \mathtt{c}^{(2)}_{\mathbb{B}} \right\rangle_{\mathrm{GOE}} \approx 0.9991$. Results are reported for \(L=15\) with \(5\) spins up and \(\omega=0\),  \(\epsilon_{\delta} = 0.5\), \(J_{xy} = 1\), \(J_{z} = 0.5\). The plot markers \(1,3,5,7\) correspond to the various choices of the defect site, with \(\delta = 1\) and \(\delta = 7\) corresponding to the integrable and chaotic limits, respectively. Figures (a) and (b) correspond to the two different bases, the site-basis and mean-field basis, respectively.}
\label{fig:two-coherence-site-and-mf-basis}
\end{figure*}

The entanglement structure of excited states has been shown to be a successful diagnostic of quantum chaos~\cite{PhysRevX.8.021026,2010NJPh12g5021D,huang2019universal}. Here, we numerically study the coherence structure of Hamiltonian eigenstates, using an open XXZ spin-chain with an onsite defect~\footnote{See Ref.~\cite{2020arXiv200610779S} for other Hamiltonian systems that become quantum chaotic in the presence of defects.}, described via a Hamiltonian of the form~\cite{santos2004integrability,gubin_quantum_2012}
\begin{align}
  \label{eq:xxz-defect-defn}
  \begin{aligned}
  H &= \underbrace{\frac{1}{4} \sum\limits_{j=1}^{L-1} \left( J_{xy} \left( \sigma^{x}_{j} \sigma^{x}_{j+1} + \sigma^{y}_{j} \sigma^{y}_{j+1} \right) + J_{z} \sigma^{z}_{j} \sigma^{z}_{j+1} \right) }_{H_{\mathrm{XXZ}}} \\
  &+  \underbrace{ \frac{1}{2} \left( \sum\limits_{j=1}^{L} \omega \sigma^{z}_{j} + \epsilon_{\delta} \sigma^{z}_{\delta} \right)}_{H_{z}} ,
    \end{aligned}
\end{align}
where \(\delta \in \{ 1,2, \cdots, L \}\) is the label for the defect site. We set \(\hbar=1\) and all sites have the same energy splitting, except the site \(\delta\), which has a splitting of \(\omega+\epsilon_\delta\) (the defect corresponds to a different value of the Zeeman splitting). We assume open boundary conditions and  set the various parameters to the following values: \(\omega=0,
\epsilon_{\delta} = 0.5, J_{xy} = 1, J_{z} = 0.5\); for a detailed discussion
of the physics surrounding the choice of parameters and how this leads to the onset of chaos, see Sec. II of
Ref.~\cite{santos2004integrability,gubin_quantum_2012}. It is easy to see that the total spin in \(z\)-direction is conserved, that is, \(\left[ H, \sigma^{z}_{\mathrm{total}} \right]\), where \(\sigma_{\mathrm{total}}^{z} \equiv \sum\nolimits_{j=1}^{L} \sigma_{j}^{z}\).  The Hamiltonian in \cref{eq:xxz-defect-defn} is integrable when the defect is on the edges of the chain, that is, $\delta=1$ or $L$, while it is non-integrable for the defect in the middle of the chain \(\delta = \left \lfloor{L/2}\right \rfloor\)~\cite{santos2004integrability,gubin_quantum_2012}. One way to observe this transition to non-integrability is via the level-spacing distribution of the Hamiltonian, as studied in Ref.~\cite{santos2004integrability,gubin_quantum_2012} and reproduced independently in \cref{fig:level-spacing-defectsites}. The level-spacing distribution transitions from a Poisson to a (universal) Wigner-Dyson form, a common signature of quantum chaos. Note that, in general, to obtain a Wigner-Dyson level-spacing distribution for chaotic systems, one needs to make sure that all the symmetries have been removed, that is, we are working in a specific symmetry sector of the system. For the system in \cref{eq:xxz-defect-defn}, we consider the spin subspace corresponding to $\left\lfloor\frac{L}{3} \right\rfloor$ spins up; once we are in this subspace, there are no degeneracies in the Hamiltonian, see Refs.~\cite{santos2004integrability,gubin_quantum_2012} for more details. Moreover, for the XXZ model with defect, Refs. \cite{PhysRevE.89.062110,Torres_Herrera_2015,torres2017dynamical} verified other signatures of quantum chaos such as local observables satisfying diagonal ETH \cite{dalessio_quantum_2016,Borgonovi2016}, and the long-time dynamics developing spectral correlations. Furthermore, in recent years, this model has also been employed in the study of many-body chaos \cite{PhysRevX.10.041017}, thermalization \cite{PhysRevLett.125.070605,PhysRevE.102.042127}, and quantum transport \cite{PhysRevLett.125.180605}.

\begin{figure*}[!t]
   \raggedright
\begin{subfigure}{.45\textwidth}
  \includegraphics[width=1.3\linewidth]{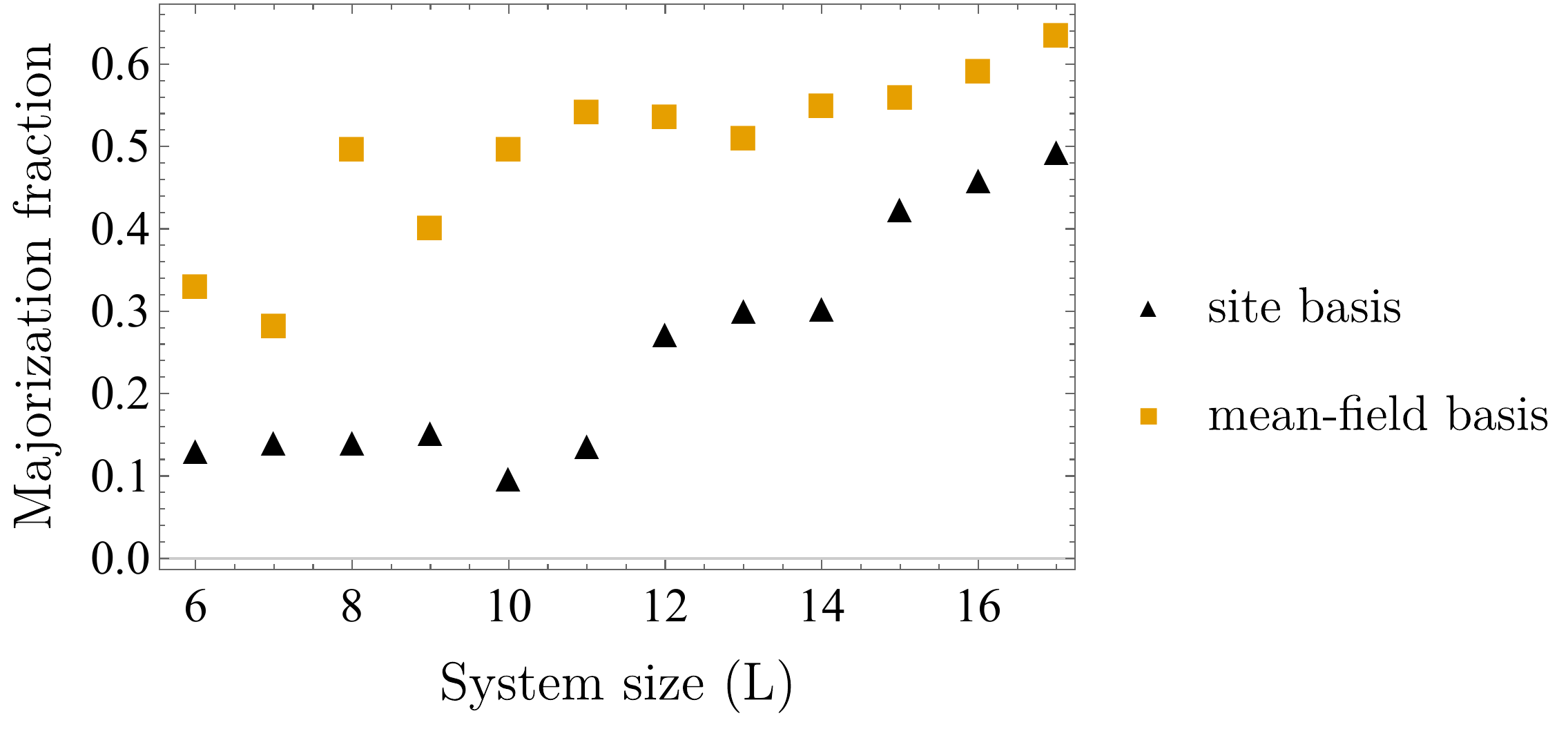}
  \caption{}
\end{subfigure}\hspace{65pt}
\begin{subfigure}{.30\textwidth}
  \includegraphics[width=1.3\linewidth]{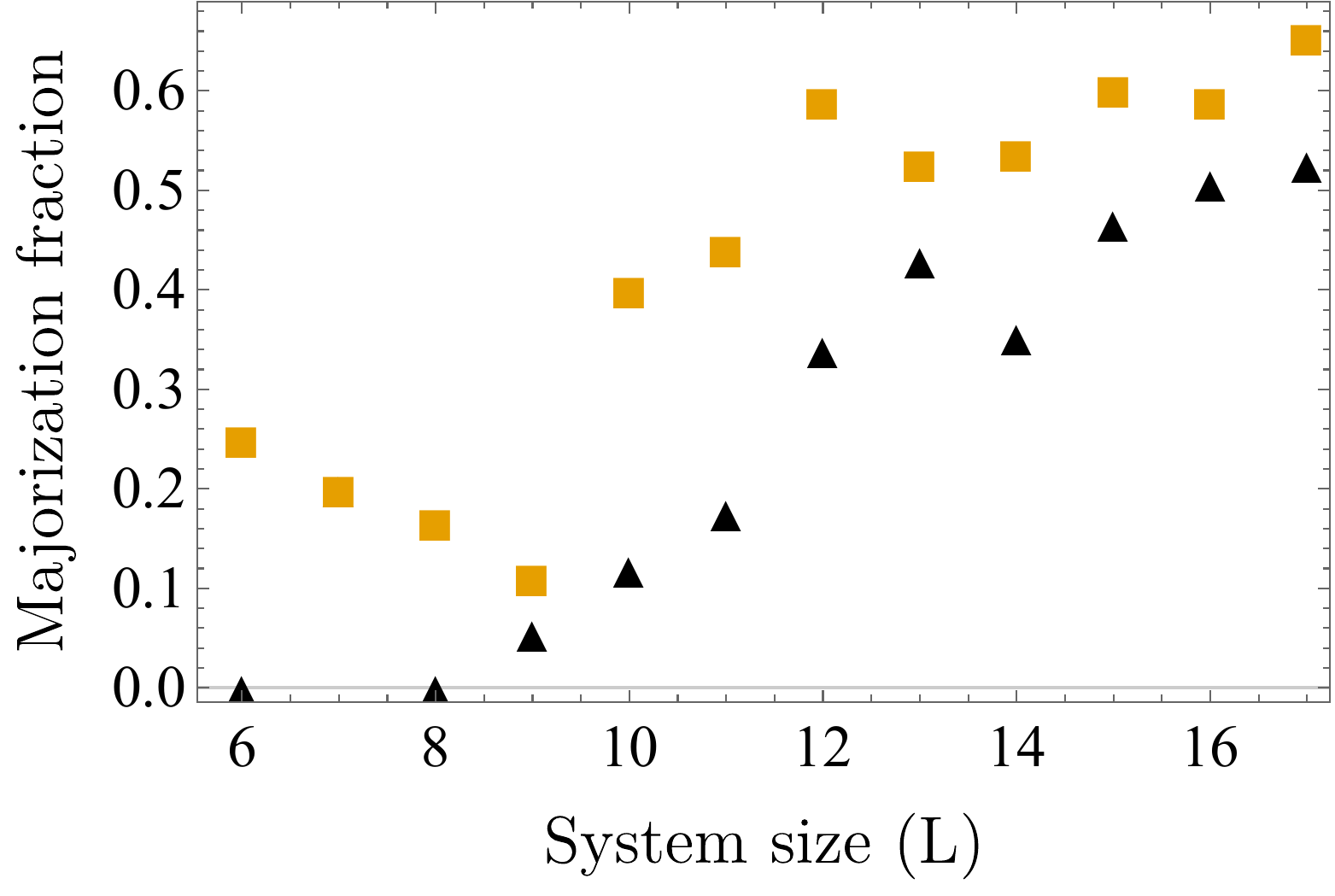}
\caption{}
\end{subfigure}%
\caption{Fraction of integrable eigenstates that majorize chaotic eigenstates for the
  Hamiltonian defined in \cref{eq:xxz-defect-defn} system size $L$. Here, $\delta=1$ for integrable eigenstates, and $\delta=\lfloor{L/2}\rfloor$ for chaotic eigenstates. The plot markers correspond to the two different bases, the
  site-basis and mean-field basis, respectively. Figures (a) and (b) correspond to the full spectrum and 20\% of eigenvectors in the middle of the spectrum, respectively.}
\label{fig:majorization-fraction}
\end{figure*}

In Ref.~\cite{santos_onset_2010}, the participation ratio as an
  indicator of chaos was studied and results similar to
  \cref{fig:ipr-site-and-mf-basis} were obtained. Using the relative
 entropy of coherence, \(2\)-coherence, inverse participation ratio (IPR), and \(1\)-norm coherence,
 we study the onset of chaos, as the defect site is moved to the
 middle of the chain. We study coherence in two different bases, the
 ``site basis'' and the ``mean-field basis''. The site basis is simply the local \(\sigma^{z}\) basis at each site and coherence in this basis is a measure of how uniformly spread is the eigenstate with respect to the local subsystems. To define the mean-field basis, we start by expressing the total Hamiltonian as \(H_{\mathrm{total}} = H_{0} + V\), where \(H_{0}\) is the Hamiltonian of noninteracting particles (or, more generally, degrees of freedom) and \(V\) the interaction between them \cite{santos_onset_2012,zelevinsky1996nuclear}. The mean-field basis is then the eigenbasis of the ``mean-field Hamiltonian,'' \(H_{0}\). This is, in fact, quite similar to the mean-field approach used in atomic and nuclear physics (and hence the terminology). It is immediately apparent that such a decomposition of the total Hamiltonian is \textit{not} unique, however, in many physical scenarios, there is a natural choice of the mean-field basis. The intuition here is that as the interaction strength increases, the eigenstates of the total Hamiltonian will become more uniformly spread when expressed in the mean-field basis. Following Refs.~\cite{santos2004integrability,gubin_quantum_2012}, we take the mean-field Hamiltonian to be \(J_{xy} \neq 0, \epsilon_{\delta} \neq 0, J_{z} = 0\). Notice that this is \textit{not} the same as the integrable limit above.

\textit{Random matrix theory}.--- Before going into the details of our numerical studies, let us briefly recall some key ideas from random matrix theory (RMT) and its predictions for quantum chaotic systems. First introduced by
Wigner~\cite{wigner_characteristic_1955, wigner_characteristics_1957,
  wigner_distribution_1958} and later developed by
Dyson~\cite{dyson_statistical_1962}, RMT has been widely used to study
complex systems and in particular, quantum chaotic systems (see Refs.~\cite{dalessio_quantum_2016,Borgonovi2016} for a pedagogical
review). Many of the originally introduced measures (like
level-spacing distribution) were purely spectral properties, but in
recent years, there has been more interest in going beyond the
spectral properties to understand the eigenstate structure of chaotic
systems~\cite{dalessio_quantum_2016,Borgonovi2016}. For instance, if quantum chaotic
systems can be well-described by RMT, then their eigenstate properties
are expected to resemble those of \emph{random} vectors in the Hilbert
space (namely, the eigenvectors of RMT Hamiltonians). However, this is not the complete picture. Many of the traditional Gaussian ensembles like the Gaussian
Orthogonal Ensemble (GOE), Gaussian Unitary Ensemble (GUE), etc. are ensembles of \textit{many-body} interactions and not \(2\)- and \(3\)-body interactions (reminiscent of physical Hamiltonians), and the properties of few-body Hamiltonians can be modelled more accurately by the use of the so-called \textit{embedded ensembles}~\cite{kota_embedded_2014}. Moreover, numerical studies have revealed that generically, only eigenstates in the middle of the spectrum correspond well to the (usual) RMT prediction (as will also be relevant for our numerical studies)~\cite{santos_onset_2010,santos_onset_2012,kota_embedded_2014}. 

We also note that using the connection between Shannon entropy and relative entropy of coherence as discussed in \cref{subsec:why-coherence}, we can infer analytically the ensemble averaged relative entropy of coherence for GOE eigenstates (see Sec. 2.3.2 of Ref.~\cite{kota_embedded_2014})
\begin{align}
  \label{eq:rel-entropy-goe-formula}
\left\langle \mathtt{c}^{(\text{rel})}_{\mathbb{B}} \right\rangle_{\mathrm{GOE}} = \ln \left( 0.48 d \right) + O \left( \frac{1}{d} \right),
\end{align}
where \(d\) is the Hilbert space dimension. Since GOE eigenvectors are (Haar) uniformly distributed, the basis $\mathbb{B}$ is a \textit{generic} basis, that is, the estimate for the ensemble average holds true for any basis~\cite{dalessio_quantum_2016,Borgonovi2016}. We use this analytical expression for normalizing the quantities studied in \cref{fig:two-coherence-site-and-mf-basis,fig:goe-normalized-rel-coherence}.

The Hamiltonian in
\cref{eq:xxz-defect-defn} is real and symmetric and belongs to the
Gaussian Orthogonal Ensemble (GOE) universality class. In  \cref{fig:ipr-site-and-mf-basis,fig:two-coherence-site-and-mf-basis,fig:goe-normalized-rel-coherence,fig:one-coherence-site-and-mf-basis}, we study the aforementioned coherence measures normalized by the GOE prediction and find that, in the middle of the spectrum, the chaotic model does reproduce the GOE prediction; which is consistent with previously known results (that the eigenstates of systems with few-body interactions delocalize in the middle of the spectrum)~\cite{kota_embedded_2014, santos2004integrability,gubin_quantum_2012,santos_onset_2010,santos_onset_2012}. Thus, this vindicates the various coherence measures as a signature of the transition to chaos.

\emph{What about other quantum coherence measures?}--- Apart from the
specific quantum coherence measures studied above, what, if anything, can be
said about an \textit{arbitrary} coherence measures' ability to probe
quantum chaos in a similar way? To answer this question, we turn to the
powerful mathematical formalism of majorization
theory~\cite{marshall_inequalities_2011} as discussed in \cref{subsec:why-coherence}. We numerically study the majorization condition in \autoref{thm:coherence-transformation-majorization} for the integrable and chaotic eigenstates of the XXZ spin-chain in \cref{eq:xxz-defect-defn} and analyze the extent to which the induced preorder order holds true. Specifically, for a given system size $L$, we consider the set of
integrable and chaotic eigenstates ordered respectively by the energies of the corresponding Hamiltonians. Then, we numerically check for the
majorization condition 
in \autoref{thm:coherence-transformation-majorization} between the $k$th
chaotic eigenstate and the $k$th integrable eigenstate (where the
index $k$ is ordered with respect to the energy). We find that the
majorization condition does not hold for all pairs of eigenstates (ordered by energy). For this reason, we introduce a weaker notion of ``majorization fraction,''
which is the fraction of eigenstates for which the majorization condition is true. Let \(\eta\) be the number of
chaotic eigenstates that are majorized by the corresponding integrable
eigenstates and \(d\) the total number of eigenstates, then, the
majorization fraction is simply the ratio \(\frac{\eta}{d}\). In
\cref{fig:majorization-fraction}, we plot the majorization fraction as
a function of the system size $L$, for both the site-basis and the
mean-field basis. We see that, for larger system sizes, a chaotic eigenstate picked at random (uniformly) is, with relatively high probability, majorized by its integrable counterpart and thus will have a larger value for any coherence measure; for example, as displayed by the relative entropy of coherence and the $2$-coherence in
\cref{fig:goe-normalized-rel-coherence,fig:two-coherence-site-and-mf-basis}. Since
physical eigenstates resemble random vectors in the middle of the
spectrum, we further consider the majorization fraction for \(20\%\)
of eigenstates in the middle of the
spectrum, and find a similar increase with system size (and a
non-monotonicity at small sizes).

\section{At the level of channels}
\label{sec:operations}
Having demonstrated the ability of quantum coherence measures to distinguish chaotic-vs-integrable eigenstates and a flurry of connections with delocalization measures, we now turn to chaos at the level of quantum dynamics (or more generally quantum channels~\footnote{We remark that quantum channels~\cite{nielsen_quantum_2010} provide a general framework that encapsulates the notions of unitary dynamics as well as open system effects, and therefore we refer to the connections henceforth as ``at the level of channels,'' for its generality.}). In particular, the ability of chaotic dynamics to generate quantum correlations has proven to be a rich framework~\cite{wang_entanglement_2004,scott_entangling_2003,madhok_signatures_2015} and here we establish rigorous connections with their ability to generate quantum coherence.

\subsection{The OTOC, quantum chaos, and its connection with CGP}
\label{sec:otoc-intro}
In recent years, the out-of-time-ordered correlator (OTOC) has emerged as a prominent diagnostic for quantum chaos at the level of dynamics~\cite{larkin_quasiclassical_1969,kitaev_simple_2015,MaldacenaChaos2016,PhysRevLett.115.131603,PolchinskiSYK2016,MezeiChaos2017,Roberts2017Chaos}. The precise role that the OTOC plays in characterizing quantum chaos via its short-time exponential growth is better understood in systems with either a semiclassical limit or systems with a large number of local degrees of freedom~\cite{kitaev_simple_2015,MaldacenaChaos2016}. On the other hand, the short-time growth does not seem to play any role for quantum chaos in finite systems such as spin-chains (without a semiclassical analog)~\cite{PhysRevB.98.134303,PhysRevLett.123.160401,luitz2017information,PhysRevE.101.010202,PhysRevLett.124.140602,hashimoto2020exponential}. However, the long-time limit of OTOCs may be expected to play a more clear role, see Refs.~\cite{PhysRevLett.121.210601,PhysRevE.100.042201}. Moreover, to further our understanding of the OTOC, several works have tried to establish a connection to well-studied signatures of chaos such as Loschmidt echo \cite{yan_information_2020} and entangling power \cite{PhysRevLett.126.030601}, which suggest that an information-theoretic investigation of the OTOC's properties might provide a fruitful direction.

The OTOC quantifies the rapid delocalization of quantum information initialized in local subsystems, which has been termed ``information scrambling''. One way to quantify this spread is to consider the growth of local operators under Heisenberg time evolution, captured by the following quantity (hereafter referred to as the ``squared commutator'' for brevity)
\begin{align}
\label{eq:otoc-definition}
  \begin{aligned}
   C_{V,W}^{(\beta)}(t) &\coloneqq \operatorname{Tr}\left( \left[V,  W(t) \right]^{\dagger} \left[ V, W(t) \right] \rho_{\beta}  \right)\\
   &= \left\Vert \left[ V, W(t) \right]  \right\Vert^2_{\beta},
     \end{aligned}
\end{align}
where \(W(t)  = \mathcal{U}^\dagger_{t} (W)\) is the Heisenberg-evolved operator, \(\rho_{\beta} \equiv e^{- \beta H}/
\operatorname{Tr}\left[ e^{-\beta H} \right]\) is the Gibbs state at
inverse temperature \(\beta\), and \(\left\Vert \cdot
\right\Vert_{\beta}\) be the norm induced from the inner product
\(\langle X,Y \rangle_{\beta} \coloneqq \operatorname{Tr}\left(
  X^{\dagger} Y \rho_{\beta} \right)\). Re-expressing
$C_{V,W}^{(\beta)}(t)$ in the commutator form resembles a
(state-dependent) variant of the Lieb-Robinson construction, which in
turn imposes fundamental limits on the speed of information
propagation in non-relativistic systems~\cite{liebFiniteGroupVelocity1972, hastingsSpectralGapExponential2006, lashkari_towards_2013,PhysRevLett.117.091602}. In this way,
$C_{V,W}^{(\beta)}(t)$ captures the spread of information through
nonlocal degrees of freedom of a system.

The connection between the squared commutator and the OTOC is revealed when we choose \(V,W\) to be unitary~\cite{larkin_quasiclassical_1969,kitaev_simple_2015}
\begin{align}
  \begin{aligned}
  \label{eq:squared-commutator-otoc}
&C_{V,W}^{(\beta)}(t) = 2 \left( 1  - \mathfrak{Re} \left\{ F_{V,W}^{(\beta)}(t) \right\}  \right),\\
& \text{ where, } F_{V,W}^{(\beta)}(t) \equiv \operatorname{Tr}\left( W(t)^{\dagger} V^{\dagger} W(t) V \rho_{\beta} \right),
  \end{aligned}
\end{align}
is a four-point function (with unusual time-ordering) called the
OTOC. Since the squared commutator above and the OTOC are related via
a simple affine function, we will focus here on the squared commutator
and refer to it interchangeably as the OTOC (the distinction should be
clear from the context). In this paper, we will focus on the
infinite-temperature ($\beta = 0$) case, that is, $\rho_{\beta} =
\frac{\mathbb{I}}{d}$. Hereafter, we define, $C_{V,W}^{(\beta=0)}(t)
\equiv C_{V,W}(t)$ and $F_{V,W}^{(\beta=0)}(t) \equiv F_{V,W}(t)$. In the following, we will connect the out-of-time-ordered correlator with the coherence-generating power, which we are now ready to introduce.

\prlsection{Coherence-generating power} How much coherence does an evolution
generate on average? Motivated from the resource theory of coherence, several meaningful quantifiers for this were obtained in
Refs.~\cite{zanardiCoherencegeneratingPowerQuantum2017,zanardiMeasuresCoherencegeneratingPower2017,zanardiQuantumCoherenceGenerating2018}. Here, we will consider the ``extremal CGP,'' defined as~\cite{styliaris_quantum_2019-1}
\begin{align} \label{eq:two-cgp-defn}
\mathfrak{C}_{\mathbb{B}} \left( \mathcal{U} \right)  = \frac{1}{d} \sum\limits_{j=1}^{d} \mathtt{c}_{\mathbb{B}}(\mathcal{U}(\Pi_{j})),
\end{align}
where \(\mathcal{U}(\cdot)= U (\cdot) U^{\dagger}\) is a unitary
channel, $\mathtt{c}_{\mathbb{B}}(\cdot)$ is a coherence measure, and $\mathbb{B} = \{ \Pi_{j}\}_{j=1}^d$ is an
orthonormal basis for the $d$-dimensional Hilbert space (see the
\cref{sec:Preliminaries} for more details). The CGP
measures the average coherence generated under time evolution $\mathcal{U}$ by its
action on the pure states in $\mathbb{B}$. For the rest of the paper
we choose \(\mathtt{c}^{(\mathrm{2})}_{\mathbb{B}}(\cdot)\) in the
above equation, that is, \(\mathfrak{C}_{\mathbb{B}} \left( \mathcal{U}
\right)  = \frac{1}{d} \sum\limits_{j=1}^{d}
\mathtt{c}^{(\mathrm{2})}_{\mathbb{B}}(\mathcal{U}(\Pi_{j}))\), which
has a closed form expression as~\cite{styliaris_quantum_2019-1}
\begin{align}
\label{eq:cgp-xmatrix-formula}
  \begin{aligned}
  \mathfrak{C}_{\mathbb{B}} \left( \mathcal{U} \right) = 1 -
  \frac{1}{d} \operatorname{Tr}\left( X_{\mathcal{U}}^{T}
  X_{\mathcal{U}} \right),\\
  \text{ where } \left[ X_{\mathcal{U}} \right]_{j,k} = \operatorname{Tr}\left( \Pi_{j} \mathcal{U}(\Pi_{k}) \right). 
  \end{aligned}
\end{align}
Hereafter, we will refer to the above quantity simply as CGP for
brevity. It is worth mentioning that
the formalism introduced in
Refs.~\cite{zanardiCoherencegeneratingPowerQuantum2017,zanardiMeasuresCoherencegeneratingPower2017,zanardiQuantumCoherenceGenerating2018,styliaris_quantum_2019-1}
is much more general than the definition \cref{eq:two-cgp-defn}. In
particular, one can consider various choices of coherence measures and
distributions over incoherent states.

The CGP defined above has many interesting properties, some of
which we review now. First, in the context of Anderson localization and many-body
localization, it was shown that the CGP acts as an ``order
parameter'' for the ergodic-to-localization transition~\cite{styliaris_quantum_2019-1}. Second, in the resource-theoretic
study of incompatibility of quantum measurements, the CGP arises
naturally as an incompatibility measure~\cite{styliaris_quantifying_2019-1}. And third, the CGP
lends itself to a power geometric connection: the \(\mathfrak{C}_{\mathbb{B}} \left( \mathcal{U} \right)\) is
proportional to the (square of the) Grasmmannian distance between two maximally
abelian subalgebras, the one generated by all bounded observables diagonal in
\(\mathbb{B}\) and those diagonal in \(\mathcal{U}(\mathbb{B})\)~\cite{zanardiQuantumCoherenceGenerating2018}. Using this connection, a closed form expression for CGP in a commutator form can be obtained as follows~\footnote{Note that this formula uses the extremal probability distribution over the incoherent states, instead of the uniform distribution, which accounts for the differing factors of \(d(d+1)\).} 

\begin{align}
\label{eq:extremal-cgp-defn}
\begin{aligned}
  \mathfrak{C} _{\mathbb{B}}(\mathcal{U}) &=\frac{1}{2d} \sum_{j,k=1}^{d}\left\|\left[\Pi_{j}, \mathcal{U}(\Pi_{k}) \right]\right\|_{2}^{2}\\
  &= \frac{1}{2d} \sum_{j,k=1}^{d} \operatorname{Tr}\left( \left[ \Pi_{j}, \mathcal{U}(\Pi_{k}) \right]^{\dagger} \left[ \Pi_{j}, \mathcal{U}(\Pi_{k}) \right]  \right).
  \end{aligned}
\end{align}

With the CGP expressed in the commutator form in
\cref{eq:extremal-cgp-defn},  we are now ready to introduce its connection
to the OTOC \(C_{V,W}(t)\). In anticipation of the theorem below, we define the following: Let
\(V,W\) be two nondegenerate unitaries with a spectral decomposition
\(V = \sum\limits_{j=1}^{d} v_{j} \Pi_{j}, W = \sum\limits_{j=1}^{d}
w_{j} \widetilde{\Pi}_{j}\). Let \(\mathbb{B}_{V} = \{ \Pi_{j}
\}, \mathbb{B}_{W} = \{ \widetilde{\Pi}_{j} \}\) be the corresponding
eigenbases, then, \(\mathcal{V}_{\mathbb{B}_{V} \rightarrow
  \mathbb{B}_{W}}\) is a unitary intertwiner connecting  \(\mathbb{B}_{V}\)
to \(\mathbb{B}_{W}\), whose action is \(\mathcal{V}_{\mathbb{B}_{V} \rightarrow \mathbb{B}_{W}} \left( \Pi_{j} \right) = \widetilde{\Pi}_{j} ~~\forall j \in \{ 1,2, \cdots,d \}\).\\ 

\begin{widetext}
\begin{restatable}{thm}{otoccgpconnection}
\label{thm:otoc-cgp-connection}
Given a unitary evolution operator \(\mathcal{U}_{t}\), and two
nondegenerate unitary operators $V$ and $W$, the infinite-temperature
out-of-time-ordered correlator ($C_{V,W}(t)$) and the CGP ($\mathfrak{C}_{\mathbb{B}} \left( \cdot \right)$) are related as
\begin{align}
\label{eq:otoc-cgp-connection}
C_{V,W}(t) = 2 \mathfrak{C}_{\mathbb{B}_{V}} \left( \mathcal{U}_{t} \circ \mathcal{V}_{\mathbb{B}_{V}  \rightarrow \mathbb{B}_{W}}  \right) - \frac{2}{d} \mathfrak{Re} \left\{ \sum\limits_{j \neq l, k \neq m}  v_{j}^{*} w_{k}^{*} v_{l} w_{m} \mathrm{Tr}\left(  \widetilde{\Pi}_{k}(t) \Pi_{j} \widetilde{\Pi}_{m}(t) \Pi_{l} \right)  \right\}.
\end{align}
\end{restatable}
\end{widetext}

\prlsection{Remarks} (a) While quantum coherence (and hence the CGP) is a basis-dependent quantity, the above theorem relates the OTOC to a CGP \textit{naturally}. Intuitively, the OTOC measures the growth of the noncommutativity between the operators $W(t)$ and $V$, and this intuition is made precise by the CGP $\mathfrak{C}_{\mathbb{B}_{V}} \left( \mathcal{U}_{t} \right)$, which measures the \textit{incompatibility}~\cite{styliaris_quantifying_2019-1} between the bases $\mathbb{B}_{V}$ and $\mathbb{B}_{\mathcal{U}_{t}}$.

(b) In \autoref{thm:otoc-cgp-connection} it is important
to emphasize that the CGP emerges as a \textit{subpart} of the
OTOC. By plugging in the spectral decomposition of the
operators \(V\) and \(W\), we obtain a summation over four indices and
by considering a subset of these terms, we obain the CGP. The
``extra'' term is of the form \(\mathrm{Tr}\left(  \widetilde{\Pi}_{k}(t) \Pi_{j}
  \widetilde{\Pi}_{m}(t) \Pi_{l} \right) \) (which is the second term
on the RHS of \cref{eq:otoc-cgp-connection}) and we refer to this as
the ``off-diagonal'' term. That is, the CGP is ``contained'' in the
OTOC. We refer the reader to the proof in the Appendix for more details.

(c) To help understand \autoref{thm:otoc-cgp-connection}, let us consider a
simple case: assume that the two operators commute at time $t=0$, that is, \(\left[ V,W \right] =0\). This is a common assumption when studying the OTOC's dynamical features, for example, by
choosing local operators on different sites (or, if they are on
the same site, by choosing them to be the same operator), then,
\(\mathcal{V}_{\mathbb{B}_{V}  \rightarrow \mathbb{B}_{W}} =
\mathcal{I}\), that is, the intertwiner can be chosen to be the
(trivial) identity superoperator. To fulfill the nondegeneracy
criteria (which we assumed initially), we can choose \(V\) and \(W\) to be quasilocal. Now, since \(\left[ V,W \right] =0\), the first term becomes equal to \(2\mathfrak{C}_{\mathbb{B}}(\mathcal{U}_{t})\),
with \(\mathbb{B}_{V} =  \mathbb{B}_{W} \equiv \mathbb{B}\). That is, simply (twice) the CGP of the time evolution unitary when measured in the basis of the operators $V$ and $W$. Using the forthcoming discussion, see \autoref{eq:otoc-avg-uniform-diagonal-unitaries}, let the eigenvalues of $V$ and $W$ be uniformly distributed over $[0,2\pi)$, then we have, $\left\langle C_{V,W}(t)\right\rangle_{V,W}= 2\mathfrak{C}_{\mathbb{B}}(\mathcal{U}_{t})$, where $\left\langle \cdots \right\rangle_{V,W}$ denotes averaging over $V,W$. That is, the ``extra term'' vanishes and the averaged OTOC is \textit{exactly} equal to twice the CGP.

\prlsection{Projection OTOCs} Here we establish another connection
between the OTOC and the CGP by choosing $V$ and $W$ to be
projection operators in the
OTOC. Similar constructions have been considered before, for example,
in Ref.~\cite{borgonovi_timescales_2019}, the authors used
``projection OTOCs'' to connect with the participation ratio. In
particular, similar a quantity known as
``fidelity OTOCs'' was proposed in Ref.~\cite{lewis-swan_unifying_2019} as an experimentally promising
approach to measure OTOCs and, in turn, to the study of scrambling and
thermalization. 
Let \(\mathbb{B}_{V} = \{ \Pi_{\alpha} \}\) and \(\mathbb{B}_{W} = \{ \widetilde{\Pi}_{\beta} \}\), we start by plugging in \(V =
\Pi_{\alpha}, W = \widetilde{\Pi}_{\beta}\) into the OTOC to obtain \(C_{\Pi_{\alpha},
  \widetilde{\Pi}_{\beta}}(t) = \frac{1}{d} \left\Vert \left[
    \Pi_{\alpha}, \widetilde{\Pi}_{\beta}(t) \right]
\right\Vert_{2}^{2}\). Then, by summing over \(\alpha\), we have,
\begin{align}
\sum\limits_{\alpha=1}^{d} C_{\Pi_{\alpha}, \widetilde{\Pi}_{\beta}}(t) = \frac{1}{d} \sum\limits_{\alpha=1}^{d} \left\Vert \left[ \Pi_{\alpha}, \widetilde{\Pi}_{\beta}(t) \right]  \right\Vert_{2}^{2} = \frac{2}{d} \mathtt{c}^{(2)}_{\mathbb{B}_{V}}(\widetilde{\Pi}_{\beta}(t)),
\end{align}
where \(\mathtt{c}^{(2)}_{\mathbb{B}_{V}}(\cdot)\) is the \(2\)-norm
coherence. Then, if we sum over \(\beta\),we have,
\begin{align}
\begin{aligned}
\sum\limits_{\alpha=1, \beta=1}^{d} C_{\Pi_{\alpha}, \widetilde{\Pi}_{\beta}}(t) &= \frac{2}{d} \sum\limits_{\beta=1}^{d}  \mathtt{c}^{(2)}_{\mathbb{B}_{V}}(\widetilde{\Pi}_{\beta}(t)) \\
&= 2 \mathfrak{C}_{\mathbb{B}_{V}} \left( \mathcal{U}_{t} \circ \mathcal{V}_{\mathbb{B}_{V} \rightarrow \mathbb{B}_{W}} \right) .
\end{aligned}
\end{align}
Therefore, given two bases, \(\mathbb{B}_{V}, \mathbb{B}_{W}\), we
have that the OTOC ``averaged'' over these bases is equal to (twice) the coherence-generating power of the unitary evolution (and the intertwiner connecting the bases). Moreover, if \(\mathbb{B}_{V}  = \mathbb{B}_{W}\), we have, 
\begin{align}
\sum\limits_{\alpha, \beta}^{d} C_{\Pi_{\alpha}, \Pi_{\beta}}(t) = 2 \mathfrak{C}_{\mathbb{B}_{V}} \left( \mathcal{U}_{t} \right) .
\end{align}
That is, the OTOC averaged over various projectors is equal to the CGP
of the time evolution unitary. Note that for a non-degenerate Hamiltonian, the CGP is equal to the average escape probability~\cite{styliaris_quantifying_2019-1}, which is intimately connected to quantities like the Loschmidt echo,
participation ratio, and others, as discussed in \cref{subsec:why-coherence}.

\prlsection{Average OTOC, coherence, and geometry} In the following we establish a connection between the average OTOC and the geometry of the set of maximally abelian subalgebras of the operator space (associated to the quantum system). For this, let us briefly introduce the geometric results obtained in Ref.~\cite{zanardiQuantumCoherenceGenerating2018} concerning CGP and \(2\)-coherence. Given a basis \(\mathbb{B}\), let \(\mathcal{A}_{\mathbb{B}}\) be the abelian algebra generated by its elements. Then, \(\mathcal{A}_{\mathbb{B}}\) is a subspace of the operator algebra \(\mathcal{B}(\mathcal{H})\) viewed as a Hilbert space \(\mathcal{H}_{\mathrm{HS}}\), endowed with the Hilbert-Schmidt inner product, \(\left\langle A,B \right\rangle_{\mathrm{HS}} \coloneqq \operatorname{Tr}\left( A^{\dagger} B \right)\), which induces the norm, \(\left\Vert A \right\Vert_{\mathrm{HS}} = \sqrt{\left\langle A,A \right\rangle}_{\mathrm{HS}} = \sqrt{\operatorname{Tr}\left( A^{\dagger}A \right)}\). If \(\mathbb{B}\) is obtained via a maximal orthogonal resolution of the identity in \(\mathcal{B}(\mathcal{H})\), then, \(\mathcal{A}_{\mathbb{B}}\) is a maximal abelian subalgebra (MASA)~\cite{griffiths1978principles,zanardiQuantumCoherenceGenerating2018}. The set of all MASAs is a topologically nontrivial subset of the Grassmannian of \(d\)-dimensional subspaces of \(\mathcal{H}_{\mathrm{HS}}\) and we can define a distance between two MASAs, \(\mathcal{A}_{\mathbb{B}}\) and \(\mathcal{A}_{\widetilde{\mathbb{B}}}\) as~\cite{zanardiQuantumCoherenceGenerating2018},
\begin{align}
  D\left(\mathcal{A}_{\mathbb{B}}, \mathcal{A}_{\mathbb{\widetilde{B}}}\right):=\left\|\mathcal{D}_{\mathbb{B}}-\mathcal{D}_{\mathbb{\widetilde{B}}}\right\|_{\mathrm{HS}},
\end{align}
where for superoperators, we have, \(\operatorname{Tr}_{\mathrm{HS}}\left( \mathcal{E} \right) := \sum_{j,k=1}^{d} \left\langle | j \rangle \langle  k | , \mathcal{E} \left( | j \rangle \langle  k |  \right) \right\rangle\). In fact, the CGP turns out to be proportional to the (squared) distance between the algebras \(\mathcal{A}_{\mathbb{B}}\) and \(\mathcal{U}(\mathcal{A}_{\mathbb{B}})\), that is~\cite{zanardiQuantumCoherenceGenerating2018},
\begin{align}
\mathfrak{C}^{}_{\mathbb{B}} \left( \mathcal{U} \right) =  \frac{1}{2 d} D^{2}\left(\mathcal{A}_{\mathbb{B}}, \mathcal{U}\left(\mathcal{A}_{\mathbb{B}}\right)\right).
\end{align}

We are now ready to introduce the main result of this section, the detailed proofs of which can be found in \cref{sec:proof-otoc-averages}. Let \(\mathbb{B}_{V} = \{ \Pi_{\alpha} \}\) and \(\mathbb{B}_{W} = \{ \widetilde{\Pi}_{\beta} \}\) be two bases. Consider unitaries diagonal in the respective bases, \(V = \sum\limits_{\alpha} e^{i \theta_{\alpha}} \Pi_{\alpha}\) and \(W = \sum\limits_{\beta} e^{i \widetilde{\theta}_{\beta}} \widetilde{\Pi}_{\beta}\), with \(\{ \theta_{\alpha} \}\) and \(\{ \widetilde{\theta}_{\beta} \}\) independent and identically distributed uniformly on the interval \(\left[ 0, 2 \pi \right) \). Then,
\begin{align}
\label{eq:otoc-avg-uniform-diagonal-unitaries}
\left\langle \left\Vert \left[ V,W(t) \right]  \right\Vert_{2}^{2}  \right\rangle_{\theta} = 2d \mathfrak{C}^{}_{\mathbb{B}_{V}} \left( \mathcal{U}_{t} \circ \mathcal{V}_{\mathbb{B}_{V} \rightarrow \mathbb{B}_{W}} \right).
\end{align}
That is, the OTOC averaged over diagonal unitaries with phases distributed uniformly reveals the CGP of the dynamics. Moreover, if \(\mathbb{B}_{V} = \mathbb{B}_{W}\), then, the relation simplifies to, \(\left\langle \left\Vert \left[ V,W(t) \right]  \right\Vert_{2}^{2}  \right\rangle_{\theta} = 2d\mathfrak{C}^{}_{\mathbb{B}_{V}} \left( \mathcal{U}_{t} \right).\) Using the connection with distance in the Grassmannian, we have,
\begin{align}
  \left\langle \left\Vert \left[ V,W(t) \right]  \right\Vert_{2}^{2}  \right\rangle_{\theta} = D^{2} \left( \mathcal{A}_{\mathbb{B}_{V}}, \mathcal{U}_{t} \left( \mathcal{A}_{\mathbb{B}_{V}} \right) \right).
\end{align}
Therefore, this average OTOC quantifies exactly the distance (squared) in the Grassmannian between MASAs \(\mathcal{A}_{\mathbb{B}_{V}}\) and \(\mathcal{U}_{t} \left( \mathcal{A}_{\mathbb{B}_{V}} \right)\). This is yet another way to understand the OTOC as measuring the incompatibility between the operators $V$ and $U_{t}$ and the bases associated to them. 

Furthermore, we can also use average OTOCs to estimate the coherence of a state. For this, we first prove the following result: given a state \(\rho\) and a unitary \(V\), we have
\begin{align}
\label{eq:otoc-avg-2-coherence}
  \left\langle \left\Vert \left[ \mathcal{D}_{\mathbb{B}}(V), \rho \right]  \right\Vert_{2}^{2}   \right\rangle_{V \in \mathrm{Haar}} = \frac{2}{d} \mathtt{c}^{(\mathrm{2})}_{\mathbb{B}}(\rho).
\end{align}

Then, as a corollary, we consider the following \textit{open system} OTOC, \(\left\langle \left\Vert \left[ \mathcal{E}_{t}(V), \rho \right]  \right\Vert_{2}^{2}  \right\rangle_{V \in \mathrm{Haar}}\), where $\{\mathcal{E}_{t}\}_{t}$ is a family of quantum channels~\cite{rivasOpenQuantumSystems2012}. If \(\{\mathcal{E}_{t}\}_{t}\) is such that \(\mathcal{E}_{t} \overset{t \rightarrow \infty} \longrightarrow \mathcal{D}_{\mathbb{B}}\), that is, in the long-time limit, \(\mathcal{E}_{t}\) converges to the dephasing channel in the basis \(\mathbb{B}\)~\cite{rivasOpenQuantumSystems2012}, then, the equilibration value of this averaged OTOC reveals the \(2\)-coherence of the state \(\rho\). That is, 
\begin{align}
  \left\langle \left\Vert \left[ \mathcal{E}_{t}(V), \rho \right]  \right\Vert_{2}^{2}  \right\rangle_{V \in \mathrm{Haar}}  \overset{t \rightarrow \infty} \longrightarrow   \frac{2}{d} \mathtt{c}^{(\mathrm{2})}_{\mathbb{B}}(\rho).
\end{align}
One can also consider instead of the quantum channel $\mathcal{E}_t$, unitary dynamics under a (time-independent) non-degenerate Hamiltonian. However, in this case, the limit \(\lim\limits_{t \rightarrow \infty}  \mathcal{U}_t\) does not exist (as opposed to \(\lim\limits_{t \rightarrow \infty} \mathcal{E}_{t}\), which does), and so we consider the infinite-time averaged value of the OTOC, which can be used to extract equilibration values of physical quantities for unitary dynamics that does equilibrate~\cite{reimann_foundation_2008,linden_quantum_2009}. 

The above result can also be generalized to the following scenario: consider two unitaries \(V, W\) and two bases \(\mathbb{B}, \widetilde{\mathbb{B}}\). Then, the following Haar-averaged OTOC is proportional to the (squared) distance in the Grassmannian between the MASAs associated to the bases \(\mathbb{B}, \mathbb{\widetilde{B}}\). That is,
\begin{align}
\label{eq:otoc-avg-masa}
  \left\langle \left\Vert \left[ \mathcal{D}_{\mathbb{B}}(V), \mathcal{D}_{\widetilde{\mathbb{B}}}(W) \right]  \right\Vert_{2}^{2}   \right\rangle_{V,W \in \mathrm{Haar}} = \frac{1}{d^{2}} D^{2}(\mathcal{A}_{\mathbb{B}}, \mathcal{A}_{\widetilde{\mathbb{B}}})
\end{align}
Following a similar corollary as above, consider two channels \(\mathcal{E}_{t}\) and \(\mathcal{N}_{\tau}\), whose long-time limit are the dephasing channels \(\mathcal{D}_{\mathbb{B}}\) and \(\mathcal{D}_{\widetilde{\mathbb{B}}}\), respectively. Then, the equilibration value of the following OTOC reveals the Grassmannian distance (squared) between the MASAs \(\mathcal{A}_{\mathbb{B}}\) and \(\mathcal{A}_{\widetilde{\mathbb{B}}}\),
\begin{align}
  \left\langle \left\Vert \left[ \mathcal{E}_{t}(V), \mathcal{N}_{\tau}(W) \right]  \right\Vert_{2}^{2}   \right\rangle_{V,W \in \mathrm{Haar}} \overset{t,\tau \rightarrow \infty} \longrightarrow  \frac{1}{d^{2}} D^{2}(\mathcal{A}_{\mathbb{B}}, \mathcal{A}_{\widetilde{\mathbb{B}}})
\end{align}
That is, average OTOCs of the above form can be used to probe geometrical distance in the Grassmannian between the two MASAs above.

Note that the Haar averages discussed above consist of a single adjoint action of $V$ (or $W$) and therefore, the same estimate can be obtained by simply averaging over elements of a $1$-design instead~\cite{divincenzo2002quantum, renes2004symmetric, scott2006tight, gross2007evenly}. For qubit systems, Pauli matrices form a $1$-design and so these averages can be accessed in a relatively simpler way. The same also holds true for the Haar-averaged $4$-point OTOCs as they do not probe the full Haar randomness either, which would (generally) require considering even higher-point functions~\cite{cotler_chaos_2017}. In summary, suitably averaged OTOCs can probe $2$-coherence of a state, the CGP of the dynamics, and the Grassmannian distance (squared) between MASAs; and in this sense quantitatively connect with the notion of coherence and incompatibility. And finally, it is worth emphasizing that although the CGP is related to quantities such as the Loschmidt echo (or survival probability) and effective dimension; see the discussion in Ref. \cite{styliaris_quantum_2019-1}, it remains unclear if the OTOC-CGP connection has any direct implications for characterizing quantum chaos.

\subsection{CGP, random matrices, and short-time growth}
\label{sec:rmt-and-short-time}
The unusual effectiveness of RMT in predicting the physics of quantum chaotic systems is quite astonishing, especially
since physical Hamiltonians (and their eigenstates) are far from random. In \cref{subsec:why-coherence} we saw that the coherence of eigenstates in the middle of the spectrum is close to the ensemble averages obtained from RMT. We now turn to dynamical features which are relevant for experimental systems such as cold atoms and ion traps~\cite{chaudhury2009nature,PhysRevLett.79.4790} which focus on time evolution; as opposed to spectral features, useful in other setups such as nuclear scattering experiments~\cite{wigner_characteristic_1955,wigner_characteristics_1957}. Here, we provide an analytical upper bound on the CGP averaged over GUE Hamiltonians and unravel a connection with the Spectral Form Factor (SFF)~\cite{berry1985semiclassical,kota_embedded_2014,haake_quantum_2010,PhysRevD.98.086026,cotler_chaos_2017}, a prominent measure of spectral correlations for quantum chaos. We begin by recalling that the GUE is defined via the following probability distribution over $d \times d$ Hermitian matrices,
\begin{align}
    P(H) \propto \exp \left(-\frac{d}{2} \operatorname{Tr}\left(H^{2}\right)\right).
\end{align}
It is easy to see that transformations of the form $H \mapsto U H U^\dagger$ leave the ensemble invariant (that is, it is unitarily invariant). The probability measure can also be written in terms of the eigenvalues $\{ \lambda_{j} \}_{j=1}^{d}$ as the following joint probability distribution
\begin{align}
    P\left(\lambda_{1}, \lambda_{2} \ldots, \lambda_{d}\right)=\exp \left(-\frac{d}{2} \sum_{i} \lambda_{i}^{2}\right) \prod_{i<j}\left(\lambda_{i}-\lambda_{j}\right)^{2}.
\end{align}
Then, defining the joint probability distribution of $n$ eigenvalues, that is, the spectral $n$-point correlation function (for $n<d$) as
\begin{align}
    \rho^{(n)}\left(\lambda_{1}, \ldots, \lambda_{n}\right)=\int d \lambda_{n+1} \ldots d \lambda_{d} P\left(\lambda_{1}, \ldots, \lambda_{d}\right),
\end{align}
where we integrate all eigenvalues from $n+1$ to $d$. We are now ready to define the SFF, which is the Fourier transform of the \(n\)-point correlation function,~\cite{kota_embedded_2014,haake_quantum_2010,PhysRevD.98.086026,cotler_chaos_2017}
\begin{align}
\label{eq:spectralff-defn}
\begin{aligned}
\mathcal{R}_{2 k}(t)=\sum\limits_{\substack{i_1, i_2, \cdots, i_k \\ j_1,j_2,\cdots,j_k}} &\int d \lambda  \rho^{(2k)}\left(\lambda_{1}, \ldots, \lambda_{2k}\right) \\
&e^{i\left(\lambda_{i_{1}}+\cdots+\lambda_{i_{k}}-\lambda_{j_{1}}-\cdots-\lambda_{j_{k}}\right)t},
\end{aligned}
\end{align}
where \(k\) is any positive integer. In particular, the four-point SFF is
\begin{align}
\begin{aligned}
  \label{eq:R2-and-R4-defn}
\mathcal{R}_{4}(t)=\sum_{k,l,m,n} &\int d \lambda \rho^{(4)}\left(\lambda_{k}, \lambda_{l} , \lambda_{m}, \lambda_{n} \right) \\
&e^{-i\left(\lambda_{k}+\lambda_{l} - \lambda_{m} - \lambda_{n} \right)t}.
\end{aligned}
\end{align}

By considering the Hamiltonian in the CGP $\mathfrak{C}_{\mathbb{B}} \left( e^{-iHt} \right)$ as a random variable over the GUE, we provide an analytical upper bound on its average value in terms of the four-point SFF.\\

\begin{widetext}
\begin{restatable}{thm}{cgpspectralff}
\label{thm:cgp-spectralff}
The coherence-generating power averaged over the Gaussian Unitary Ensemble (GUE) is upper bounded by the four-point spectral form factor as
\begin{align}
  \label{eq:avg-cgp-spectralff}
\left\langle \mathfrak{C}_{\mathbb{B}} \left( e^{-iHt} \right) \right\rangle_{\mathrm{GUE}} \leq 1 - \frac{1}{d (d+1) (d+2) (d+3)} \underbrace{ \sum\limits_{k,l,m,n} \int d \lambda \rho^{(4)}\left(\lambda_{k}, \lambda_{l} , \lambda_{m}, \lambda_{n} \right) e^{-i \left( \lambda_{k} + \lambda_{l} - \lambda_{m} - \lambda_{n} \right) t}}_{\mathcal{R}_{4}}. 
\end{align}

Moreover, the bound is tight for short times.
\end{restatable}
\end{widetext}

\autoref{thm:otoc-cgp-connection} and \autoref{thm:cgp-spectralff} establish a three-way connection between CGP, OTOCs, and SFF; with the CGP a subpart of the OTOC and its GUE average upper bounded by the SFF. The SFF as a function of time has a characteristic qualitative features for quantum chaotic systems resembling a slope, dip, ramp, and plateau~\cite{cotler_chaos_2017,Cotler2017bhrmt}. As a future work, it would be interesting to see whether the CGP --- which is connected to the SFF via \autoref{thm:cgp-spectralff} --- can capture similar features, and in turn be used to detect associated quantum signatures of chaos.

In a similar spirit to the RMT average above, one can treat the time evolution unitary $U$ itself as a random variable. This allows us to address an important question: How well can chaotic dynamics be approximated by random unitaries? The pursuit of this question has revealed many physical insights into the
nature of strongly-interacting systems, from condensed matter systems to black holes and has inspired a multitude of quantitative connections
between chaos and random unitaries; see for example Refs.~\cite{hosur_chaos_2016,cotler_chaos_2017,Cotler2017bhrmt}. To establish similar connections, we now compute the Haar average of the OTOC-CGP relation using \autoref{thm:otoc-cgp-connection}.\\

\begin{widetext}
\begin{restatable}{prop}{haaravgotoc}
\label{thm:haar-avg-otoc}
The Haar-averaged OTOC is given by
\begin{align}
  \label{eq:haar-avg-otoc}
\left\langle C_{V,W} \right\rangle_{U \sim \mathrm{Haar}} = \frac{2(d-1)}{(d+1)} +\frac{2}{d^{2}(d^{2}-1)} \mathfrak{Re} \left\{ \sum\limits_{j \neq l \text{ and } k \neq m} v_{j}^{*} w_{k}^{*} v_{l} w_{m}  \right\} - \frac{2}{d(d+1)} \mathfrak{Re} \left\{ \sum\limits_{j \neq l}  v_{j}^{*} v_{l} + \sum\limits_{k \neq m}^{}  w_{k}^{*} w_{m} \right\},
\end{align}
where \(U \sim \mathrm{Haar}\) represents Haar-averaging over the time-evolution unitary.
\end{restatable}
\end{widetext}

The first term in this expression is obtained from the Haar-average of
the \(2\)-CGP, while the other two terms originate from the
off-diagonal contribution. We briefly remark that since the function
\(C_{V,W}(t)\) is Lipschitz continuous,
using tools from measure concentration and Levy's lemma~\cite{ledouxConcentrationMeasurePhenomenon2005}, we have that
the probability of a random instance of \(C_{V,W}(t)\) deviating
from its Haar average \(\left\langle C_{V,W}
\right\rangle_{U \sim \mathrm{Haar}}\) is exponentially suppressed. That is,
the Haar-average is representative of \textit{almost all} instances of
the OTOC. Furthermore, similar to the discussion following \autoref{thm:otoc-cgp-connection}, that is, using the result of \autoref{eq:otoc-avg-uniform-diagonal-unitaries}, we note that averaging over commuting unitaries with their phases distributed uniformly on \([0,2 \pi)\), the extra terms vanish for the averaged OTOCs. Moreover, for generic operators $V$ and $W$, the main
contribution comes from the Haar-average of the CGP, which gets
exponentially close to \(2\) in the dimension (if $d$ scales as $2^n$ for $n$ qubits). Therefore, the typical OTOC for Haar-random evolutions is exponentially well-approximated by the CGP value.

\prlsection{Short-time growth} 
To further establish dynamical features of the CGP, we focus on its short-time behavior. While the OTOC's short-time growth has been used as a diagnostic of chaos for systems with a semiclassical or large-$N$ limits, its behavior for general many-body systems with local interactions and finite degrees of freedom can simply be understood via Lieb-Robinson bounds \cite{lieb1972finite,Vershyninaliebrobinson,luitz2017information,PhysRevB.96.060301} and does not necessarily characterize chaos~\cite{PhysRevX.8.021013,PhysRevX.8.021014,PhysRevX.8.031057,PhysRevX.8.031058,PhysRevB.98.134303,PhysRevLett.123.160401,luitz2017information,PhysRevE.101.010202,PhysRevLett.124.140602,hashimoto2020exponential,wang2020quantum}. To provide information-theoretic meaning to a subpart of the OTOC (that is, the CGP), we connect it to the notion of quantum fluctuations and incompatibility. Incompatibility of observables in quantum theory is perhaps most commonly understood in terms of a non-vanishing commutator (for example, the canonical $[\hat{x},\hat{p}]$ commutator) and the related Heisenberg uncertainty relations. In recent years, however, entropic uncertainty relations have emerged as a generalized and more robust way to quantify the incompatibility of observables~\cite{colesEntropicUncertaintyRelations2017}. In Ref.~\cite{styliaris_quantifying_2019-1}, the authors introduced a formalism that encompasses both and quantified the notion of incompatibility between bases $\mathbb{B}_{0}$ and $\mathbb{B}_{1}$ (and not just observables). Among many interesting connections, it was shown how this incompatibility manifests itself as the coherence of states $|\psi\rangle \in \mathbb{B}_{0}$ when expressed as a linear combination of elements from $\mathbb{B}_{1}$. Moreover, using tools from matrix majorization, a partial order on bases was unveiled, with the order quantifying incompatibility. In particular, the CGP was established as a measure of incompatibility between different bases and its connection to entropic uncertainty relations was discussed. In the theorem below, we find that the short-time growth of the CGP captures incompatibility between the basis $\mathbb{B}$ in which we measure coherence and the basis of the Hamiltonian $\mathbb{B}_{H}$.\\

\begin{restatable}{prop}{shorttimecgp}
\label{thm:short-time-cgp}
The short-time growth of the CGP is connected to the variance of the
Hamiltonian as
\begin{align}
\frac{1}{2} \frac{d^2 \mathfrak{C}_{\mathbb{B}}(\mathcal{U}_{t})}{d {t^2} } \bigg\rvert_{t=0} = \frac{1}{d} \sum\limits_{j=1}^{d} \mathrm{var}_{j} \left( H \right),
\end{align}
where \(\text{var}_{j}(H) \equiv \langle H^{2} \rangle_{\Pi_{j}} -
\langle H \rangle_{\Pi_{j}}^{2}\) is the variance of the Hamiltonian
in the basis state \(\Pi_{j}\). Moreover, the following bounds hold:
\begin{align}
\begin{aligned}
 \frac{1}{d} \sum\limits_{j=1}^{d} \mathrm{var}_{j} \left( H \right) &\leq \frac{\left\Vert H \right\Vert_{2}^{2}}{d} \left\Vert 1 - X^{T}(\mathbb{B},\mathbb{B}_{H}) X(\mathbb{B},\mathbb{B}_{H}) \right\Vert_{\infty} \\
 &\leq \left\Vert H \right\Vert^2_{\infty} q(\mathbb{B},\mathbb{B}_{H})
 \end{aligned}
\end{align}

where, \( \left[ X(\mathbb{B},\mathbb{B}_{H}) \right]_{j,k} \equiv
\operatorname{Tr}\left( \Pi_{j} P_{k} \right) \) and  \(q (\mathbb{B}_{H}, \mathbb{B}_{0}) \equiv \left\Vert 1 - X^{T}(\mathbb{B},\mathbb{B}_{H}) X(\mathbb{B},\mathbb{B}_{H}) \right\Vert_{\infty} \).
\end{restatable}

To understand the upper bound, $\frac{1}{2} \frac{d^2 \mathfrak{C}_{\mathbb{B}}(\mathcal{U}_{t})}{d {t^2} } \bigg\rvert_{t=0} \leq \left\Vert H \right\Vert^2_{\infty} q(\mathbb{B},\mathbb{B}_{H})$, first note that the matrix $X(\mathbb{B},\mathbb{B}_{H})$ is bistochastic. And, so is $X^{T}(\mathbb{B},\mathbb{B}_{H}) X(\mathbb{B},\mathbb{B}_{H})$, using the fact that the set of bistochastic matrices is closed under transposition and multiplication~\cite{marshall_inequalities_2011}. Using this, it is easy to see that $q(\mathbb{B},\mathbb{B}_{H}) \leq 1$, therefore, we have the following bound $\frac{1}{\left\Vert H_a \right\Vert^2_{\infty}} \frac{1}{2}
\frac{d^2 \mathfrak{C}_{\mathbb{B}}(\mathcal{U}_{t})}{d {t^2} }
\bigg\rvert_{t=0} \leq 1$. Note that this quantity also provides a physically meaningful normalization on the short-time growth of the CGP: when comparing the timescales generated by Hamiltonian dynamics, $U_t = e^{-iHt}$, one can increase/decrease the associated timescales by scaling the Hamiltonian $H \mapsto \alpha H$. To fix this arbitrariness, when comparing two different dynamics, it makes sense to normalize the norm of various Hamiltonians, which, in this case happens naturally via the operator norm.

To further elucidate the theorem above and the associated bounds, we introduce a family of commuting \(k\)-local Hamiltonians of the form,
\begin{align}
H^{(k)}:= \sum\limits_{j=1}^{L-(k-1)} \left( \sigma^{x}_{j} \otimes \sigma^{x}_{j+1} \otimes \cdots \otimes \sigma^{x}_{j+(k-1)} \right).
\end{align}
Recall that a Hamiltonian is called \(k\)-local ($k\leq L$) if it can be written as a sum over terms which act on \textit{at most} \(k\) subsystems \cite{kitaev2002classical}. Let \(\mathbb{B}\) be the computational basis (that is, the local \(\sigma^{z}\) basis), then, we prove the following,
\begin{align}
\label{eq:shorttime-k-local}
\left.\frac{1}{\left\|H^{(k)}\right\|_{\infty}^{2}} \frac{1}{2} \frac{d^{2} \mathfrak{C}_{\mathbb{B}}\left(\mathcal{U}_{t}\right)}{d t^{2}}\right|_{t=0}=\frac{1}{L-(k-1)}.
\end{align}
We provide a brief sketch of the proof in \cref{sec:proof-shorttime-examples}. For \(k=1\) this generates a \(1\)-local Hamiltonian, that is, composed of purely local interactions, \(H^{(1)} =
\sum\limits_{j=1}^{L} \sigma^{x}_{j}\), which does not generate entanglement or correlations. And, for $k=L$, we have, a highly nonlocal Hamiltonian, \(H^{(L)} = \otimes^{L}_{j=1} \sigma^{x}_{j} \), which can generate an $L$-qubit Greenberger–Horne–Zeilinger (GHZ) state starting from product states\footnote{This follows immediately by expanding \(\exp \left[ -i H_{b}t \right] = \cos (t) \mathbb{I} - i \sin (t) H_{b}\), letting \(t = \pi/4\), and choosing the initial state to be \(| 0 \rangle^{\otimes n}\), using which, we get, \(| \psi(t=\pi/4) \rangle = \frac{1}{\sqrt{2}} \left( | 0 \rangle^{\otimes n} -i | 1 \rangle^{\otimes n} \right)\).} ~\cite{nielsen_quantum_2010,greenberger2007going}. Using the general result above, we notice that if $k=O(1)$, then, the normalized short-time growth, $\left.\frac{1}{\left\|H^{(k=O(1))}\right\|_{\infty}^{2}} \frac{1}{2} \frac{d^{2} \mathfrak{C}_{\mathbb{B}}\left(\mathcal{U}_{t}\right)}{d t^{2}}\right|_{t=0} \sim O(1/L)$ to leading order, while it is saturated for the nonlocal Hamiltonian $\left.\frac{1}{\left\|H^{(k=L)}\right\|_{\infty}^{2}} \frac{1}{2} \frac{d^{2} \mathfrak{C}_{\mathbb{B}}\left(\mathcal{U}_{t}\right)}{d t^{2}}\right|_{t=0} = 1$. Finally, we note that the variance of the Hamiltonian that shows up in the theorem above is intimately related to (i) quantum speed limits and the resource theory of asymmetry (see~\cite{marvianQuantumSpeedLimits2016,marvianHowQuantifyCoherence2016} and the references therein) and (ii) the ``strength function,'' widely used in quantum chaos literature (see Sec. 3 of Ref. \cite{Borgonovi2016} for more details). It would be an interesting future direction to quantitatively establish these connections further.

\subsection{Quantifying chaos with recurrences: numerical simulations}
\label{sec:recurrences}

 \begin{figure*}[!t]
   \raggedright
\begin{subfigure}{.39\textwidth}
  \includegraphics[width=1.3\linewidth]{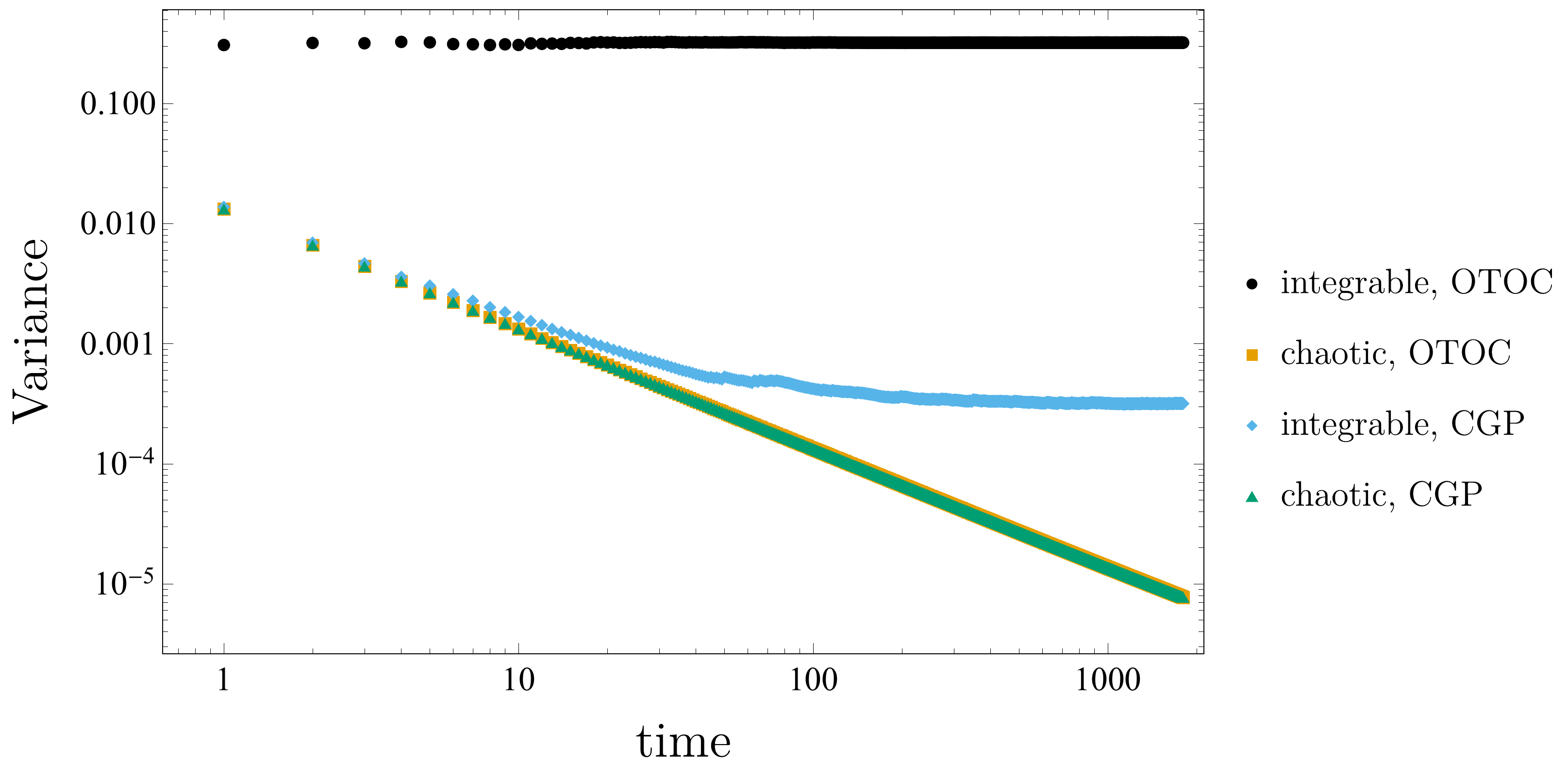}
  \caption{}
\end{subfigure}\hspace{65pt}
\begin{subfigure}{.31\textwidth}
  \includegraphics[width=1.3\linewidth]{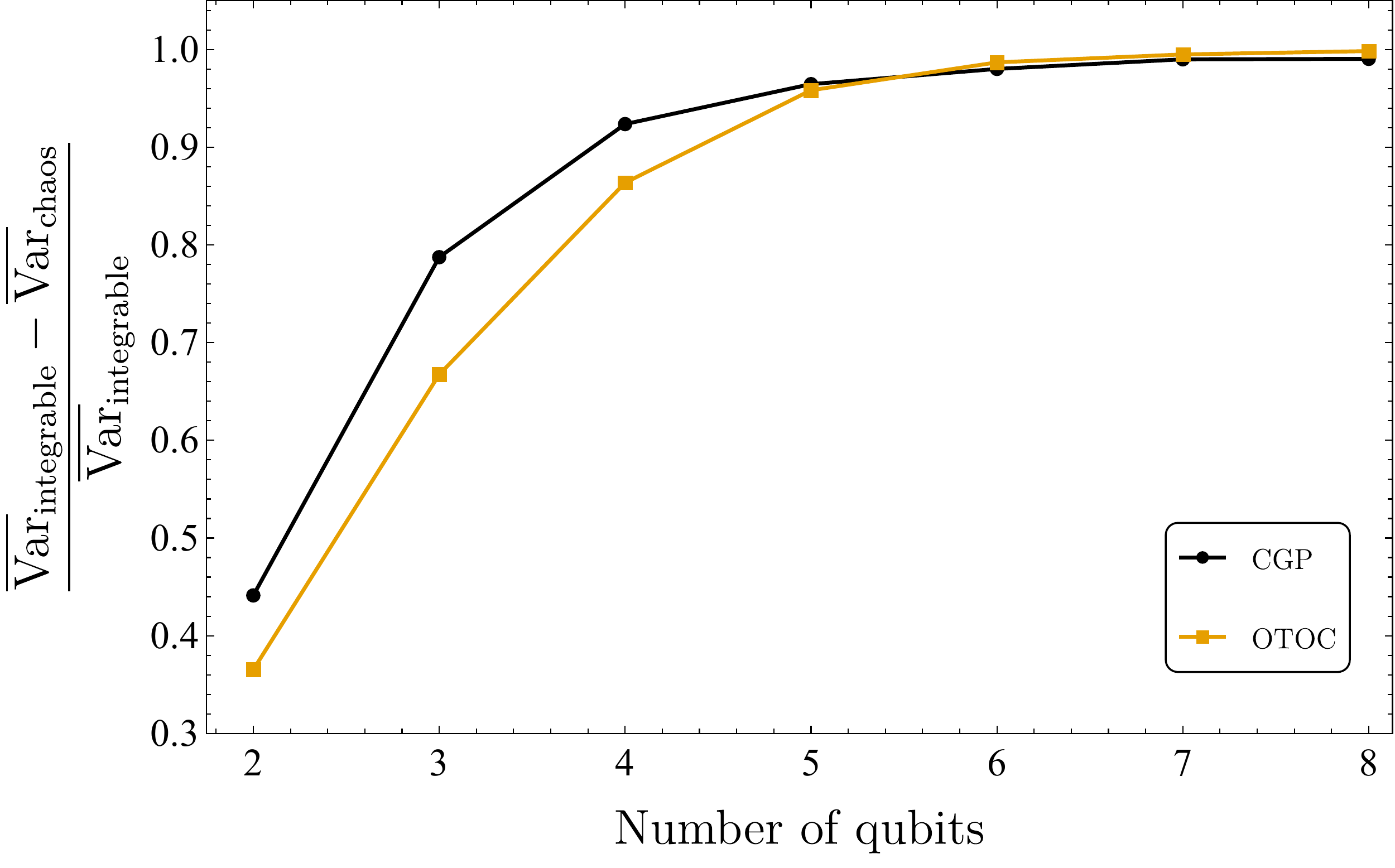}
    \caption{}
\end{subfigure}%
\caption{(a) Log-Log plot of the variance of CGP and OTOC for \(n=9\)
      qubits. We study the dynamics of the Hamiltonian given by \cref{eq:tfim-defn}
      with \(g=1, h=0\) as the integrable limit and \(g=-1.05, h=0.5\)
      as the chaotic one. We set, \(V =
      \sigma_{1}^{z}, W = \sigma^{z}_{9}\) for the OTOC and CGP in
      \cref{eq:otoc-cgp-connection}. (b) Fraction of the long-time average of the variance of
      chaotic and integrable OTOC, CGP, that is,
      \(\frac{\overline{\mathrm{Var}}_{\mathrm{integrable}} - \overline{\mathrm{Var}}_{\mathrm{chaos}}}{\overline{\mathrm{Var}}_{\mathrm{integrable}}}\) for the Hamiltonian given by \cref{eq:tfim-defn}
      with \(g=1, h=0\) as the integrable limit and \(g=-1.05, h=0.5\)
      as the chaotic one. We set, \(V =
      \sigma_{1}^{z}, W = \sigma^{z}_{9}\) for the OTOC and CGP in
      \cref{eq:otoc-cgp-connection}.}
\label{fig:CGP-vs-OTOC-variance}
\end{figure*}

OTOCs capture the scrambling of quantum information. As localized
information spreads through the nonlocal degrees of freedom of a
system, it becomes inaccessible to local observables and their
expectation values reveal an equilibration of the subsystem state. This apparent
irreversible loss of information under unitary dynamics (which is
reversible) has been termed scrambling. Signatures of scrambling can
be observed in the long-time
averages of both simple physical quantities like local expectation
values and in ``complex'' quantities such as the OTOC and
CGP. However, in finite systems, such long-time averages do not
converge in the limit \(t \rightarrow \infty\), instead they typically
oscillate around some equilibrium value. This equilibrium value can be
obtained from the infinite-time average, \(\overline{A}
\coloneqq \lim_{T \rightarrow \infty} \frac{1}{T}
\int\limits_{0}^{T} A(t) dt\). In Ref.~\cite{2020arXiv200708570S}, the infinite-time average of the averaged OTOC (with a bipartition in the system Hilbert space) was studied for both integrable and chaotic models and its equilibration value was used to successfully distinguish the two phases; see also related work studying the long-time limit of OTOCs for the integrability-to-chaos transition~\cite{PhysRevLett.121.210601,PhysRevE.100.042201}. Along the way, connections with entropy production, operator entanglement, and channel distinguishability were also discussed.

It was previously shown that in the long-time limit, the strength of recurrences can distinguish chaotic and integrable systems~\cite{2015arXiv150904352C,PhysRevLett.49.509}. Let \(n\) be the number of qubits (or more generally,
the system size), then integrable systems typically have a quantum recurrence
time that is a polynomial in \(n\), while chaotic systems typically
have recurrence times that are doubly exponential in \(n\), that is,
\(O(e^{e^{n}})\). Therefore, when studying recurrences in the
expectation values of observables for a finite (but large) time, one expects integrable systems to
show larger recurrences than chaotic systems. Building on the work of
Refs.~\cite{hosur_chaos_2016, cotler_chaos_2017}, we show that by considering the
OTOC and the CGP as ``complex observables'' and quantifying their recurrences via their temporal variance, one can distinguish integrable and chaotic regimes. We also argue that
for the purposes of distinguishing these two phases via the strength
of their recurrences, the OTOC and CGP capture effectively the same
behavior, vindicating our \autoref{thm:otoc-cgp-connection}.

The physical system we use to study this temporal variance is the paradigmatic transverse-field Ising model with open boundary conditions,
\begin{align}
  \label{eq:tfim-defn}
H_{\mathrm{TFIM}} = - \left( \sum\limits_{j=1}^{L-1} \sigma^{z}_{j} \sigma^{z}_{j+1} + \sum\limits_{j=1}^{L} g \sigma^{x}_{j} + h \sigma^{z}_{j} \right).
\end{align}
The system has an integrable limit for \(h=0\), where the Hamiltonian can be mapped onto free fermions; we
set \(g=1,h=0\) as the integrable point. The system is quantum chaotic
for the parameter choices \(g=-1.05, h=0.5\) which can be seen, for
example, by studying the level spacing distribution. In
Ref.~\cite{hosur_chaos_2016}, the OTOC averaged over local observables
was used to distinguish the two phases and it was observed that in the chaotic limit, the system
quickly asymptotes to just below the Haar-averaged value, while in the integrable regime, the systems displays large recurrences and does not show any features of scrambling. A
similar behavior was observed for the mutual information between different subsystems. Here, we compare the dynamical behavior of the OTOC and the CGP for \(V =
\sigma^{z}_{1}, W = \sigma_{L}^{z}\) for an \(L\)-site system. Notice that our numerical simulations use exact dynamics but are limited to timescales far below the expected recurrence time for the chaotic limit. However, we are able to observe and quantify recurrences for the integrable case in a way that is sufficient to distinguish the two phases.

For systems satisfying the ETH Ansatz~\cite{srednickiChaosQuantumThermalization1994, deutsch_quantum_1991, rigol_thermalization_2008}, fluctuations around the long-time averages of expectation values of observables
  will be exponentially small in the system size~\cite{dalessio_quantum_2016,Borgonovi2016}. While the CGP and
  OTOC are ``complex'' quantities, their behavior can be expected to
  resemble that of simpler observables, especially for finite systems and simple local operators such as Pauli matrices. Since quantum chaotic systems typically obey the ETH Ansatz (after
removing trivial symmetries), the fluctuations in the OTOC and CGP around their long-time average may be expected to become exponentially small in the system size. Our numerical findings
summarized in \cref{fig:CGP-vs-OTOC-variance} vindicate this intuition: we consider the long-time average of the OTOC and the CGP in the integrable and chaotic regimes. In the chaotic
regime the variance of the CGP and the OTOC are equal up to numerical
error (\(\approx 10^{-10}\) in dimensionless units), while in the
integrable regime the variance seems to asymptote to different values
for the CGP and the OTOC -- which is simply a consequence of the
different timescales of recurrences in these two quantities. A more
meaningful comparison can be obtained by computing the
\textit{relative} fluctuations in the integrable and chaotic regimes,
for which we compute the ratio 
$$\frac{\overline{\mathrm{Var}}_{\mathrm{integrable}} -  \overline{\mathrm{Var}}_{\mathrm{chaos}}}{\overline{\mathrm{Var}}_{\mathrm{integrable}}},$$
where $\overline{\mathrm{Var}}_{\mathrm{integrable}}$ is the long-time average of the temporal variance of the CGP/OTOC in the integrable regime, performed numerically. We find that for both the OTOC and CGP, this quantity becomes exponentially close to one as a function of the system
size. Therefore, the fluctuations around the average in the chaotic regime are
exponentially smaller than that in the integrable case, as expected, and both the OTOC
and its \textit{subpart}, the CGP can diagnose chaoticity in this way.

\section{Discussion}
\label{sec:discussion}
While the role of quantum entanglement in characterizing quantum chaos has been widely explored, it remained unclear what precise role quantum coherence plays, if any, in diagnosing quantum chaos. Our work affirmatively answers this question by establishing rigorous connections between measures of quantum coherence and signatures of quantum chaos. Coherence of Hamiltonian eigenstates is shown to be an ``order parameter'' for the integrable-to-chaotic transition and we numerically demonstrate this by studying quantum chaos in an XXZ spin-chain with defect and find excellent agreement with random matrix theory (RMT) in the bulk of the spectrum, as expected. Furthermore, using the mathematical formalism of majorization theory and fundamental results from the resource theory of coherence, we argue why \textit{every}  quantum coherence measure is a ``delocalization'' measure --- a class of signatures of quantum chaos that quantify spread, in say, the position eigenbasis, energy eigenbasis, and others. Moreover, our \autoref{prop:2coherence-entropy-connection} shows that for pure states in a bipartite system, the \(2\)-coherence minimized over product bases is equal to the linear entropy of the reduced state. That is, quantum coherence measures can be used to detect the entanglement in a quantum state, as has also been demonstrated previously~\cite{streltsovMeasuringQuantumCoherence2015}.

For dynamical signatures of chaos, our \autoref{thm:otoc-cgp-connection} establishes the coherence-generating power (CGP) as a \textit{subpart}  of the OTOC, a prominent measure of information scrambling in quantum systems. In particular, the (associated) squared-commutator’s growth signals the increasing incompatibility of the operators under time-evolution. Our theorem paves a way to make this intuition precise as the CGP quantifies incompatibility between the bases associated to the time-evolving operator in the OTOC and the fixed one. Moreover, we analytically show, in many different ways, how the OTOC, suitably averaged, connects with \(2\)-coherence of a state, the CGP of dynamics, and the geometric distance between the MASAs associated to the bases of the operators in the OTOC. Among a plethora of other reasons, the CGP is particularly well-suited to quantify this incompatibility since it also happens to be a formal measure in the resource theory of measurement incompatibility~\cite{styliaris_quantifying_2019-1}. 

Furthermore, using RMT we provide an upper bound on the average CGP for GUE Hamiltonians in terms of the Spectral Form Factor, a well-established measure of quantum chaos. We also find an analytical expression for the Haar-averaged OTOC-CGP relation, which allows us to argue that under certain assumptions, the OTOC is approximated exponentially-well (in the system size) by the CGP.

The short-time behavior of the OTOC has received considerable attention in recent years and so we analyze the short-time growth of the CGP (a subpart of the OTOC) which, to leading order, is characterized by the variance of the Hamiltonian with respect to a basis; for the OTOC this basis is inherited from the choice of the OTOC operators. We remark that this variance of the Hamiltonian (for pure states) is related to quantum speed limits and the resource theory of asymmetry~\cite{marvianQuantumSpeedLimits2016}. And finally, we numerically study the long-time behavior of the OTOC and CGP in a transverse-field Ising model and find that their temporal variances quantify chaos in effectively the same way.

In closing, our results establish quantum coherence as a signature of quantum chaos, both at the level of states and dynamics. As a future work, it would be interesting to see how well suited measures of quantum coherence are to the study few-body chaos, in particular, using paradigmatic systems like the quantum kicked top~\cite{wang_entanglement_2004}. Few-body systems provide a powerful experimental testbed for studying signatures of thermalization and scrambling, which are intimately linked with quantum coherence measures. Quantitatively establishing these connections will also be a promising future direction.

\section{Acknowledgments}
N.A. would like to thank Todd Brun, Bibek Pokharel, and Evangelos Vlachos for many insightful discussions about quantum chaos. Research was funded by the Deutsche Forschungsgemeinschaft (DFG, German Research Foundation) under Germany's Excellence Strategy -- EXC-2111 -- 390814868. This research was supported in part by Perimeter Institute for Theoretical Physics. Research at Perimeter Institute is supported in part by the Government of Canada through the Department of Innovation, Science and Economic Development Canada and by the Province of Ontario through the Ministry of Colleges and Universities. P.Z. acknowledges partial support from the NSF award PHY-1819189. This research was (partially) sponsored by the Army Research Office and was accomplished under Grant Number W911NF-20-1-0075. The views and conclusions contained in this document are those of the authors and should not be interpreted as representing the official policies, either expressed or implied, of the Army Research Office or the U.S. Government. The U.S. Government is authorized to reproduce and distribute reprints for Government purposes notwithstanding any copyright notation herein. 

\bibliography{my_library}

\onecolumngrid
\widetext
\newpage

\appendix

\renewcommand{\thepage}{A\arabic{page}}
\setcounter{page}{1}
\renewcommand{\thesection}{A\arabic{section}}
\setcounter{section}{0}
\renewcommand{\thetable}{A\arabic{table}}
\setcounter{table}{0}
\renewcommand{\theequation}{A\arabic{equation}}
\setcounter{equation}{0}

\section*{\large Appendices}
\label{sec:appendix}

\section{Level spacing distribution}
\label{sec:level-spacing}
\begin{figure}[!ht]
\centering
\begin{subfigure}{.25\textwidth}
  \centering
  \includegraphics[width=1.3\linewidth]{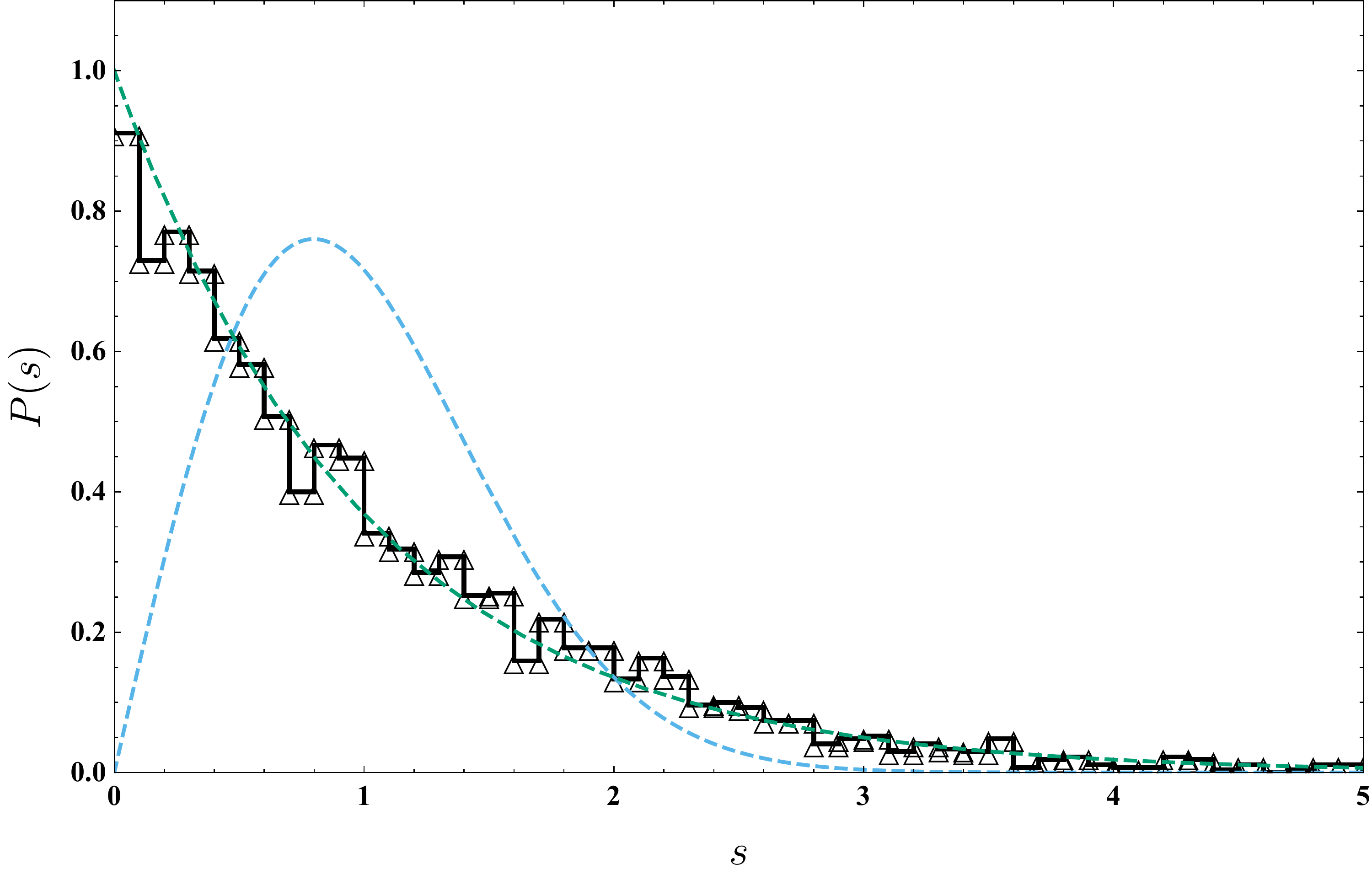}
  \caption{}
\end{subfigure}\hspace{50pt}
\begin{subfigure}{.25\textwidth}
  \centering
  \includegraphics[width=1.3\linewidth]{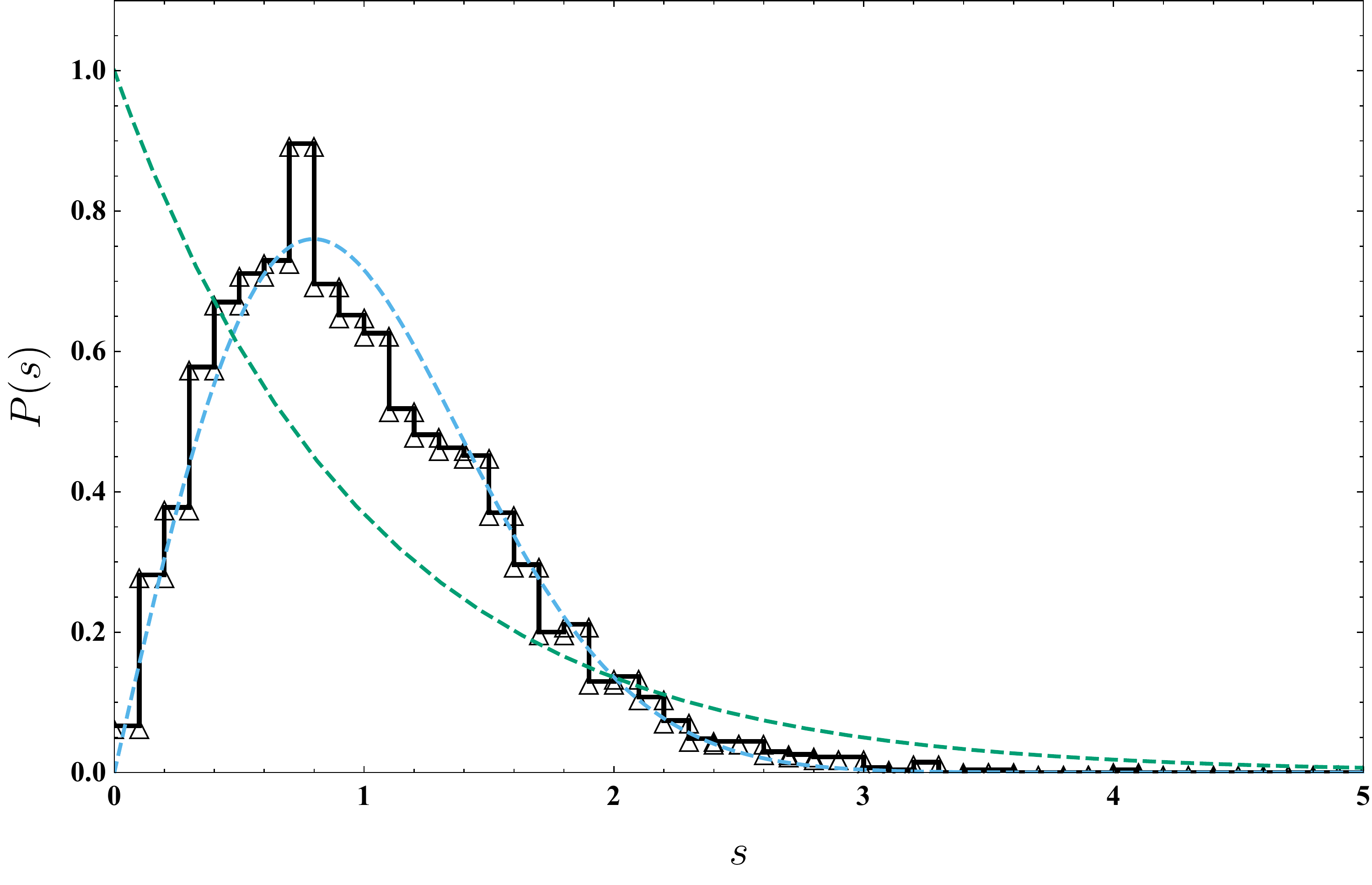}
  \caption{}
\end{subfigure}\\[1ex]
\begin{subfigure}{.25\textwidth}
  \centering
  \includegraphics[width=1.3\linewidth]{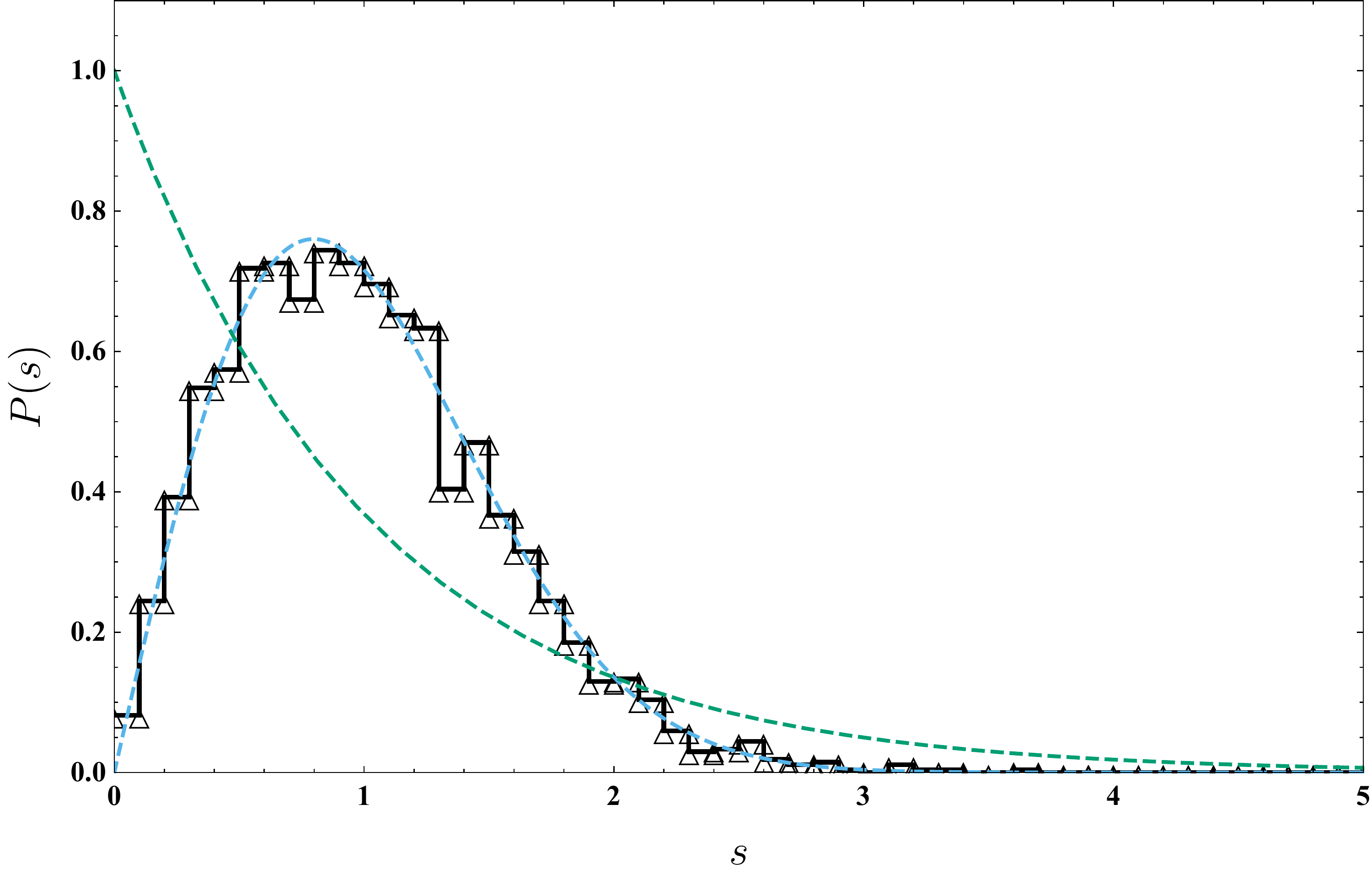}
  \caption{}
\end{subfigure}\hspace{50pt}
\begin{subfigure}{.25\textwidth}
  \centering
  \includegraphics[width=1.7\linewidth]{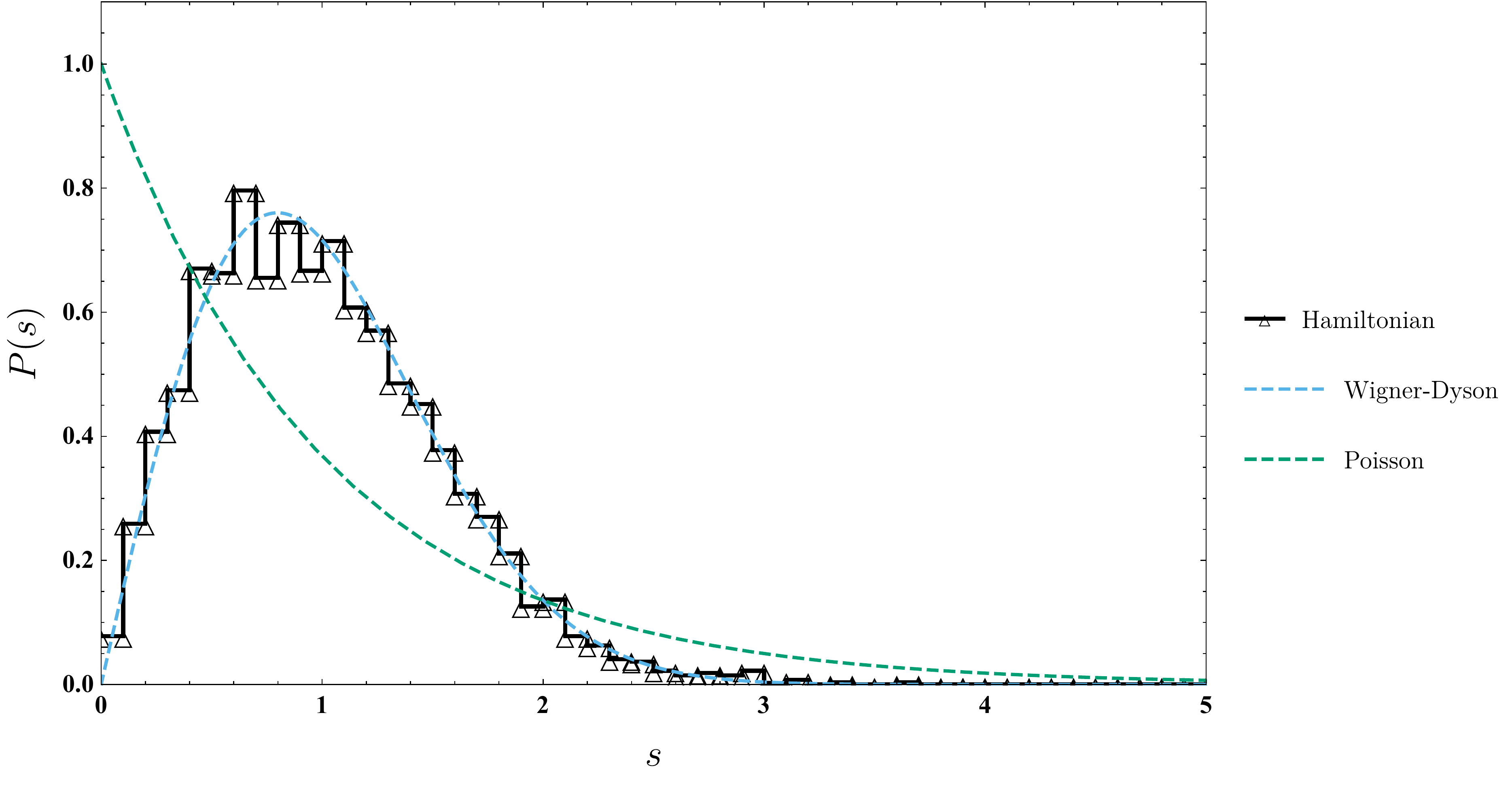}
    \caption{}
\end{subfigure}%
\caption{The transition in level-spacing distribution from Poisson to the (universal) Wigner-Dyson distribution for the Hamiltonian described in \cref{eq:xxz-defect-defn} as we move the defectsite to the middle of the chain. Figures (a), (b), (c), and (d) correspond to the defect at sites $\delta = 1, \delta =3, \delta =5,$ and $\delta = 7$, respectively. Results are reported for \(L=15\) with \(5\) spins up and \(\omega=0,  \epsilon_{\delta} = 0.5, J_{xy} = 1, J_{z} =  0.5\). Similar results were obtained for $L=15$ and $\delta = 1,7$ in Ref.~\cite{gubin_quantum_2012} (but not for intermediate positions of the defect site).}
\label{fig:level-spacing-defectsites}
\end{figure}

\section{Coherence quantifiers for integrable and chaotic eigenstates}
\label{sec:coherence-measures-appendix}

\begin{figure*}[!ht]
   \raggedright
\begin{subfigure}{.37\textwidth}
  \includegraphics[width=1.3\linewidth]{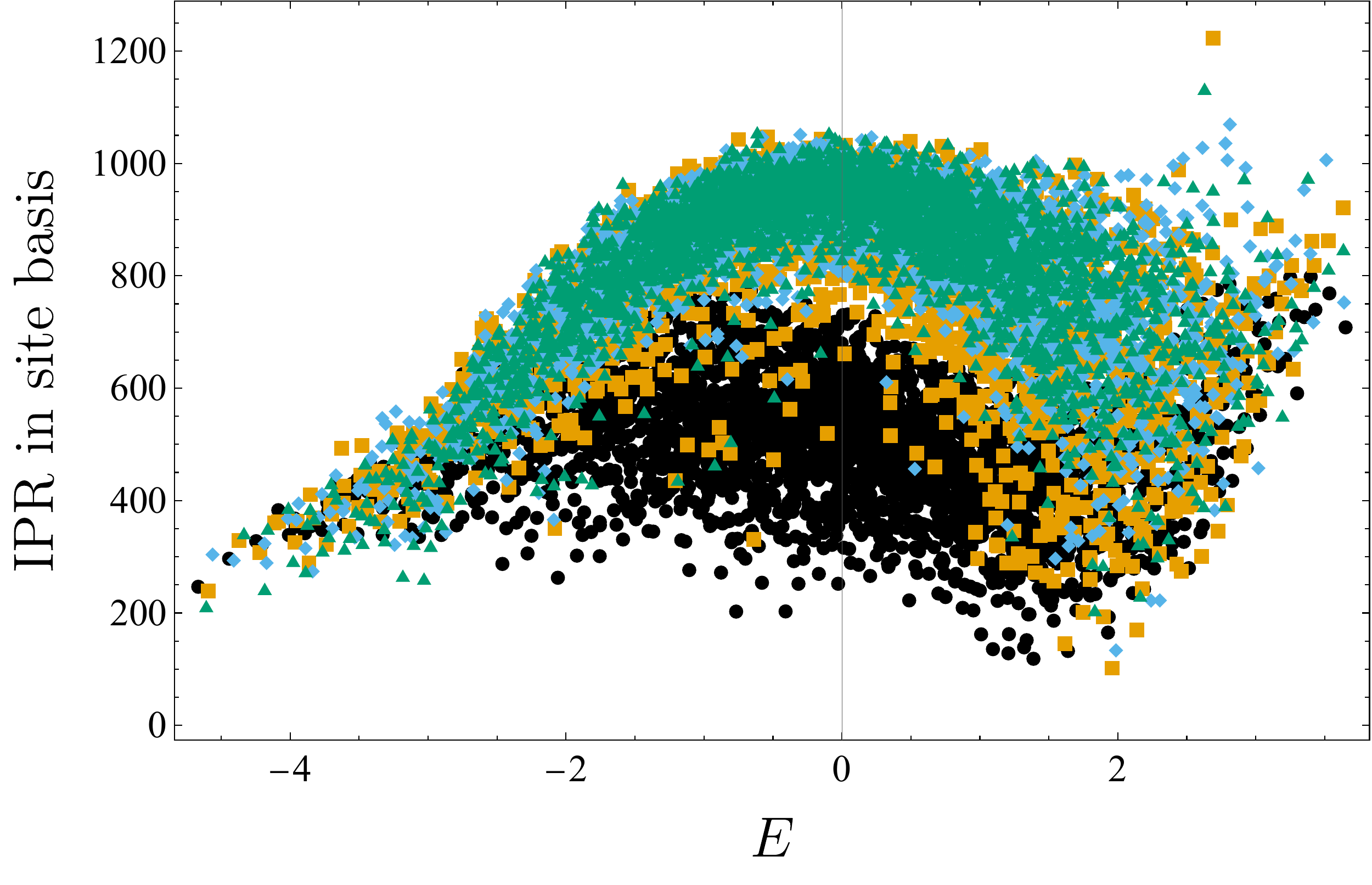}
    \caption{}
\end{subfigure}\hspace{65pt}
\begin{subfigure}{.39\textwidth}
  \includegraphics[width=1.3\linewidth]{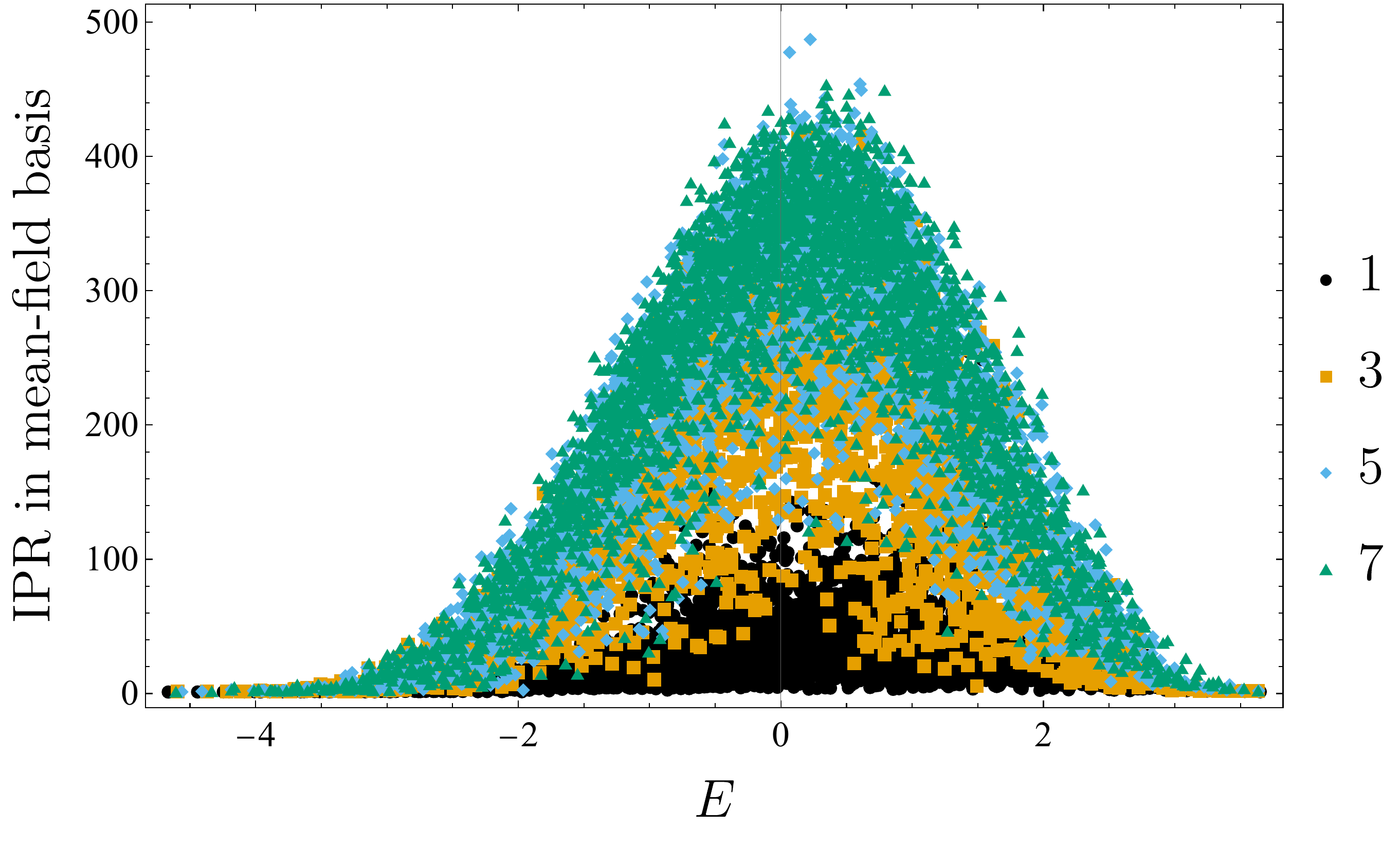}
  \caption{}
\end{subfigure}%
\caption{Inverse participation ratio for eigenstates of the
  Hamiltonian defined in \cref{eq:xxz-defect-defn} as a function of their energy. Results are reported for \(L=15\) with \(5\) spins
  up and \(\omega=0,  \epsilon_{\delta} = 0.5, J_{xy} = 1, J_{z} =
  0.5\). The plot markers \(1,3,5,7\) correspond to the various choices of
  the defect site, with \(\delta = 1,7\) corresponding to the integrable and chaotic
  limits, respectively. Figures (a) and (b) correspond to the two different bases, the
  site-basis and the mean-field basis, respectively. Similar results were obtained for $L=18$ and $\delta = 1,9$ in Ref.~\cite{gubin_quantum_2012} (but not for intermediate positions of the defect site).}
\label{fig:ipr-site-and-mf-basis}
\end{figure*}

\begin{figure*}[!ht]
   \raggedright
\begin{subfigure}{.37\textwidth}
  \includegraphics[width=1.3\linewidth]{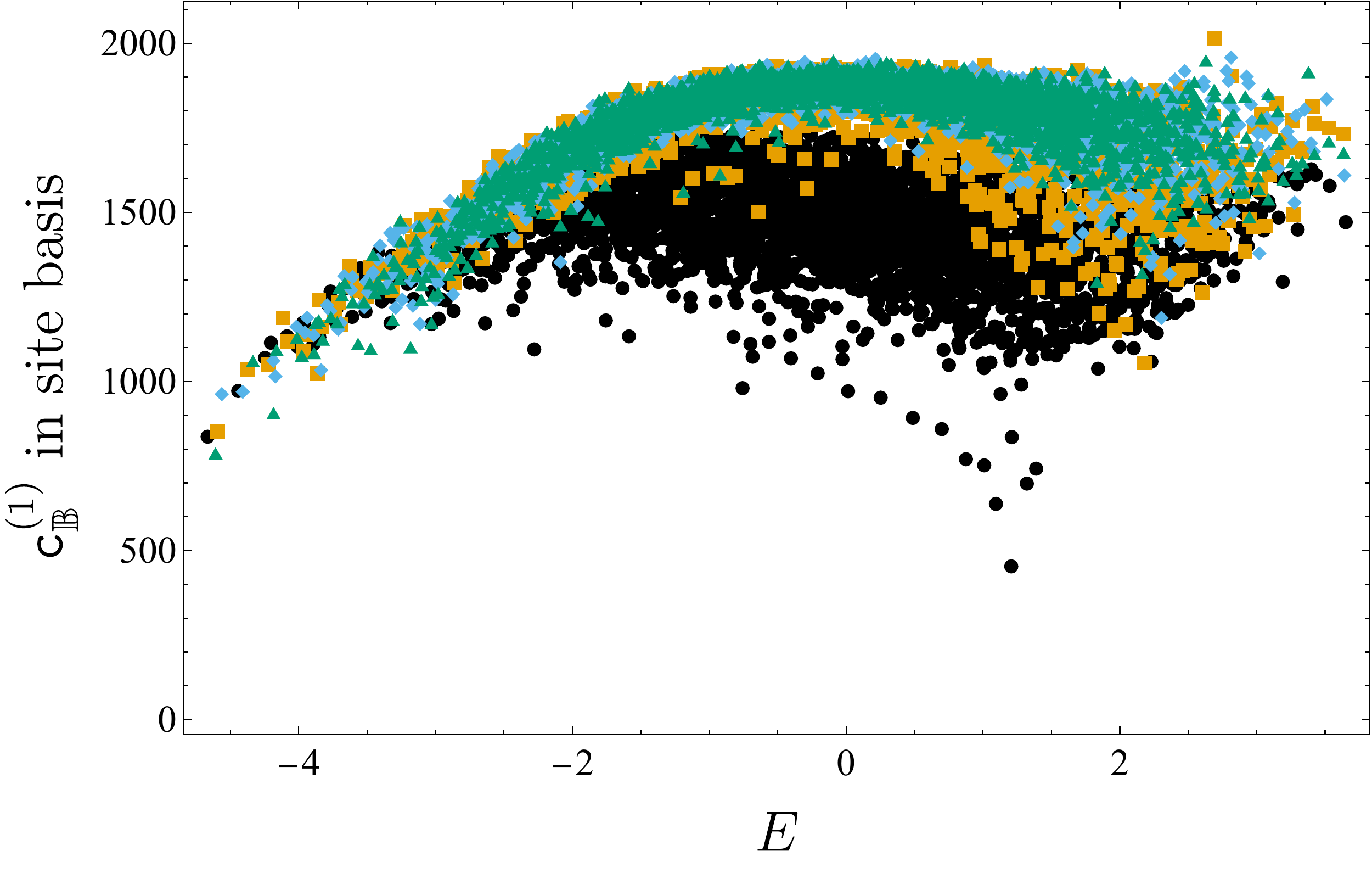}
  \caption{}
\end{subfigure}\hspace{65pt}
\begin{subfigure}{.39\textwidth}
  \includegraphics[width=1.3\linewidth]{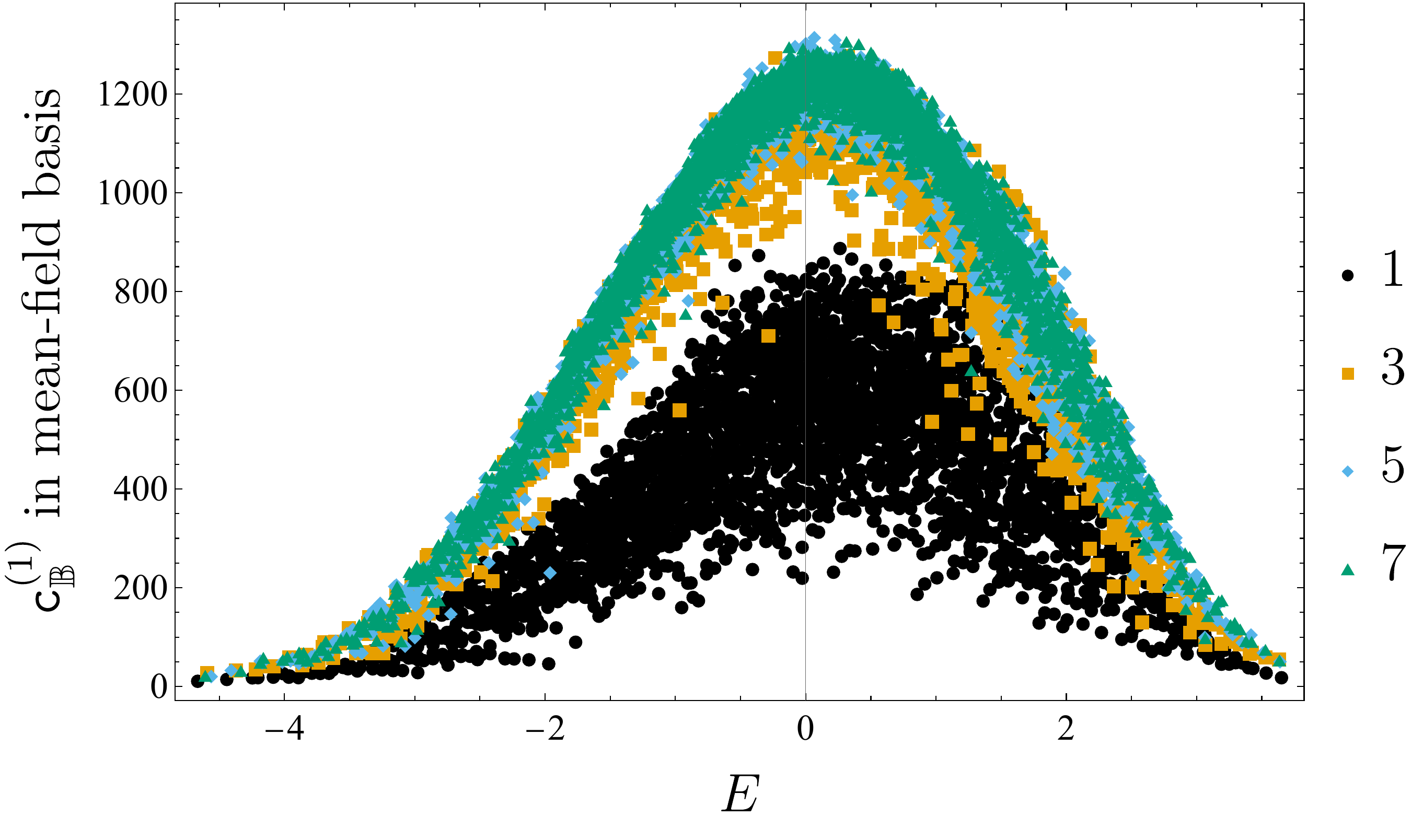}
  \caption{}
\end{subfigure}%
\caption{\(1\)-coherence for eigenstates of the
  Hamiltonian defined in \cref{eq:xxz-defect-defn} as a function of their energy. Results are reported for \(L=15\) with \(5\) spins
  up and \(\omega=0,  \epsilon_{\delta} = 0.5, J_{xy} = 1, J_{z} =
  0.5\). The plot markers \(1,3,5,7\) correspond to the various choices of
  the defect site, with \(\delta = 1,7\) corresponding to the integrable and chaotic
  limits, respectively. Figures (a) and (b) correspond to the two different bases, the
  site-basis and the mean-field basis, respectively.}
\label{fig:one-coherence-site-and-mf-basis}
\end{figure*}

\newpage

\section{Proofs} 
\label{sec:app:proofs}

Here we restate the Propositions, Theorems, as well as other mathematical claims appearing in the main text, and give their proofs.

\subsection*{Proof of \autoref{prop:2coherence-entropy-connection}}
\coherenceentropy*
\begin{proof}
We start by collecting a few simple results. First, recall that the \(2\)-coherence is
\begin{align}
  \mathtt{c}^{(\mathrm{2})}_{\mathbb{B}}(\rho) &= \left\Vert \rho - \mathcal{D}_{\mathbb{B}}(\rho) \right\Vert_{2}^{2} = \left\langle \rho - \mathcal{D}_{\mathbb{B}}(\rho), \rho - \mathcal{D}_{\mathbb{B}}(\rho) \right\rangle = \left\langle \rho, \rho \right\rangle - \left\langle \rho, \mathcal{D}_{\mathbb{B}}(\rho) \right\rangle - \left\langle \mathcal{D}_{\mathbb{B}}(\rho), \rho \right\rangle + \left\langle \mathcal{D}_{\mathbb{B}}(\rho), \mathcal{D}_{\mathbb{B}}(\rho) \right\rangle\\
  &= \left\langle \rho, \rho \right\rangle - \left\langle \rho, \mathcal{D}_{\mathbb{B}}(\rho) \right\rangle,
\end{align}
where in the second line, we have used \(\left\langle \rho,
  \mathcal{D}_{\mathbb{B}}(\rho) \right\rangle = \left\langle
  \mathcal{D}_{\mathbb{B}}(\rho), \rho \right\rangle \) since
\(\mathcal{D}_{\mathbb{B}}\) is a self-adjoint superoperator and
\(\left\langle \mathcal{D}_{\mathbb{B}}(\rho),
  \mathcal{D}_{\mathbb{B}}(\rho) \right\rangle = \left\langle \rho,
  \mathcal{D}_{\mathbb{B}}(\rho) \right\rangle\) since
\(\mathcal{D}_{\mathbb{B}}\) is a projection superoperator, that is,
\((\mathcal{D}_{\mathbb{B}})^{2} = \mathcal{D}_{\mathbb{B}}\). 

For pure states, we have, \(\left\langle \rho, \rho \right\rangle =
1\) and therefore, the \(2\)-coherence for pure states is equal to
\(\mathtt{c}^{(\mathrm{2})}_{\mathbb{B}}(\rho) = 1 - \left\langle \rho,
  \mathcal{D}_{\mathbb{B}}(\rho)\right\rangle\).\\

Second, a pure bipartite state, \(| \Psi \rangle_{AB} \in \mathcal{H}
\cong \mathcal{H}_{A} \otimes \mathcal{H}_{B}\) can be written in the
Schmidt form (that is, using Schmidt decomposition theorem)~\cite{nielsen_quantum_2010},
\begin{align}
  | \Psi \rangle_{AB} = \sum\limits_{j=1}^{\mathrm{min} \left\{ d_{A},
  d_{B} \right\}} \lambda_{j} | j \rangle_{A} \otimes | \widetilde{j}
  \rangle_{B},
\end{align}
 where $\{ | j \rangle_{A} \}, \{ | \widetilde{j} \rangle_{B} \}$ is
 an orthonormal basis for subsystems A, B, respectively, and \(\{
 \lambda_{j} \}\) are non-negative coefficients satisfying
 \(\sum\limits_{j}^{} \lambda_{j}^{2} = 1\). The coefficients \(\lambda_{j}^{2}\)
 are the eigenvalues of the reduced density matrix \(\rho_{A}\);
 recall also that \(\rho_{A}\) and \(\rho_{B}\) are isospectral. Then,
 re-expressing the state in this form, we have (dropping the
 subscripts for the subsytems \(A,B\)),
 \begin{align}
   | \Psi \rangle \langle  \Psi |  = \sum\limits_{j,k}^{} \lambda_{j} \lambda_{k} | j \rangle \langle  k | \otimes | \widetilde{j} \rangle \langle  \widetilde{k} | .
 \end{align}

 And third, the dephasing superoperator factorizes, that is,
 \begin{align}
 \mathcal{D}_{\mathbb{B}_{a} \otimes \mathbb{B}_{b}}  =  \mathcal{D}_{\mathbb{B}_{a}} \otimes \mathcal{D}_{\mathbb{B}_{b}}.
\end{align}
To see this, let \(\mathbb{B}_{a} = \{ \Pi_{j}^{(a)} \}_{j=1}^{d}\) and  \(\mathbb{B}_{b} = \{ \Pi_{k}^{(b)} \}_{k=1}^{d}\), then, the action of \(\mathcal{D}_{\mathbb{B}_{a} \otimes \mathbb{B}_{b}}\) is 
\begin{align}
  \mathcal{D}_{\mathbb{B}_{a} \otimes \mathbb{B}_{b}} (X) = \sum\limits_{j,k=1}^{d} \left( \Pi_{j}^{(a)} \otimes \Pi_{k}^{(b)} \right) X \left( \Pi_{j}^{(a)} \otimes \Pi_{k}^{(b)} \right)
\end{align}
and the action of \(\mathcal{D}_{\mathbb{B}_{a}} \otimes \mathcal{D}_{\mathbb{B}_{b}}\) is,
\begin{align}
  \mathcal{D}_{\mathbb{B}_{a}} \otimes \mathcal{D}_{\mathbb{B}_{b}} (X) &= \mathcal{D}_{\mathbb{B}_{a}} \left( \sum\limits_{k=1}^{d} \left( \mathbb{I} \otimes \Pi_{k}^{(b)} \right) X \left( \mathbb{I} \otimes \Pi_{k}^{(b)} \right) \right) \\
&= \sum\limits_{j,k=1}^{d} \left( \Pi_{j}^{(a)} \otimes \mathbb{I} \right) \left( \mathbb{I} \otimes \Pi_{k}^{(b)} \right) X \left( \mathbb{I} \otimes \Pi_{k}^{(b)} \right)  \left( \Pi_{j}^{(a)} \otimes \mathbb{I} \right)\\
&= \sum\limits_{j,k=1}^{d} \left( \Pi_{j}^{(a)} \otimes \Pi_{k}^{(b)} \right) X \left( \Pi_{j}^{(a)} \otimes \Pi_{k}^{(b)} \right) =  \mathcal{D}_{\mathbb{B}_{a} \otimes \mathbb{B}_{b}} (X).
\end{align}

We are now ready to prove the main result.
\begin{align}
  \min_{\mathbb{B}_{a},\mathbb{B}_{b}} \mathtt{c}^{(\text{2})}_{\mathbb{B}_{a} \otimes \mathbb{B}_{b}} (| \Psi \rangle \langle  \Psi | ) =  \min_{\mathbb{B}_{a},\mathbb{B}_{b}} \left\{ 1 - \left\langle | \Psi \rangle \langle  \Psi | , \mathcal{D}_{\mathbb{B}_{a} \otimes \mathbb{B}_{b}} \left( | \Psi \rangle \langle  \Psi |  \right) \right\rangle \right\} = 1 - \max_{\mathbb{B}_{a},\mathbb{B}_{b}} \left\{ \left\langle | \Psi \rangle \langle  \Psi | , \mathcal{D}_{\mathbb{B}_{a} \otimes \mathbb{B}_{b}} \left( | \Psi \rangle \langle  \Psi |  \right) \right\rangle \right\}. 
\end{align}

Let us consider the term inside the maximization, \(\left\langle | \Psi \rangle \langle  \Psi | , \mathcal{D}_{\mathbb{B}_{a} \otimes \mathbb{B}_{b}} \left( | \Psi \rangle \langle  \Psi |  \right) \right\rangle\). We use the Schmidt form of \(| \Psi \rangle\) and substitute \(\mathcal{D}_{\mathbb{B}_{a} \otimes \mathbb{B}_{b}}\) by \(\mathcal{D}_{\mathbb{B}_{a}} \otimes \mathcal{D}_{\mathbb{B}_{b}}\) to get
\begin{align}
  \left\langle | \Psi \rangle \langle  \Psi | , \mathcal{D}_{\mathbb{B}_{a} \otimes \mathbb{B}_{b}} \left( | \Psi \rangle \langle  \Psi |  \right) \right\rangle &= \sum\limits_{j,k,l,m}^{d} \lambda_{j} \lambda_{k} \lambda_{l} \lambda_{m} \operatorname{Tr}\left( | l \rangle \langle  m | \otimes | \widetilde{l} \rangle \langle  \widetilde{m} | \mathcal{D}_{\mathbb{B}_{a}} \left( | j \rangle \langle  k |  \right) \otimes \mathcal{D}_{\mathbb{B}_{b}} \left( | \widetilde{j} \rangle \langle  \widetilde{k} |  \right) \right)\\
&= \sum\limits_{j,k,l,m}^{d} \lambda_{j} \lambda_{k} \lambda_{l} \lambda_{m}  \operatorname{Tr}\left( | l \rangle \langle  m |   \mathcal{D}_{\mathbb{B}_{a}} \left( | j \rangle \langle  k |  \right) \right) \operatorname{Tr}\left( | \widetilde{l} \rangle \langle  \widetilde{m} | \mathcal{D}_{\mathbb{B}_{b}} \left( | \widetilde{j} \rangle \langle  \widetilde{k} |  \right)   \right).
\end{align}
It is easy to see that to maximize these inner products, we need to choose the dephasing basis to be the same as the local basis \(\{ | j \rangle \}, \{ | \widetilde{j} \rangle \}\), respectively. To see this, let \(\mathbb{B}_{a} = \{ | \phi_{j} \rangle \langle  \phi_{j} |  \}\), then, the term \(\operatorname{Tr}\left( | l \rangle \langle  m |   \mathcal{D}_{\mathbb{B}_{a}} \left( | j \rangle \langle  k |  \right) \right)\) becomes,
\begin{align}
  \sum\limits_{j=1}^{d} \langle \phi_{j} | l  \rangle \langle m | \phi_{j}  \rangle \langle \phi_{j} | j  \rangle \langle k | \phi_{j}  \rangle,
\end{align}
an upper bound on which can be obtained using Cauchy-Schwarz inequality repeatedly to see that it is maximized when \(| \phi_{j} \rangle = | j \rangle ~~\forall j\). That is, the local basis in the Schmidt decomposition of the state and the dephasing basis are the same. Therefore, \( \mathcal{D}_{\mathbb{B}_{a}} \left( | j \rangle \langle  k |  \right) = | j \rangle \langle  k | \delta_{j,k} \) and \( \mathcal{D}_{\mathbb{B}_{b}} \left( | \widetilde{j} \rangle \langle  \widetilde{k} |  \right) = | \widetilde{j} \rangle \langle  \widetilde{k} | \delta_{j,k} \). Plugging it back, we have, 
\begin{align}
\max_{\mathbb{B}_{a},\mathbb{B}_{b}} \left\{ \left\langle | \Psi \rangle \langle  \Psi | , \mathcal{D}_{\mathbb{B}_{a} \otimes \mathbb{B}_{b}} \left( | \Psi \rangle \langle  \Psi |  \right) \right\rangle \right\} = \sum\limits_{j,k,l,m}^{d} \lambda_{j} \lambda_{k} \lambda_{l} \lambda_{m}  \delta_{jk} \delta_{m,j} \delta_{l,j} = \sum\limits_{j=1}^{d} \lambda_{j}^{4} = \left\Vert \rho_{a} \right\Vert_{2}^{2}.
\end{align}

Therefore, putting everything together, we have,

\begin{align}
\min_{\mathbb{B}_{a},\mathbb{B}_{b}} \mathtt{c}^{(\text{2})}_{\mathbb{B}_{a} \otimes \mathbb{B}_{b}} (| \Psi \rangle \langle  \Psi | ) = 1 - \left\Vert \rho_{a} \right\Vert_{2}^{2} =: S_{\text{lin}}(\rho_{a}).
\end{align}
\end{proof}

\subsection*{Proof of \autoref{thm:otoc-cgp-connection}}
\otoccgpconnection*
\begin{proof}
Consider the infinite-temperature OTOC, \(C^{(\beta=0)}_{V,W}(t) = \frac{1}{d} \operatorname{Tr}\left( \left[ V,W(t) \right]^{\dagger} \left[ V,W(t) \right]   \right)\). Then, plugging in the spectral decomposition of \(V,W\), that is, \(V = \sum\limits_{j} v_{j} \Pi_{j},  W(t) = \sum\limits_{j} w_{j} \widetilde{\Pi}_{j}(t)\), we have
\begin{align*}
C_{V,W}(t)  =  \frac{1}{d} \sum\limits_{j,k,l,m} v_{j}^{*} w_{k}^{*} v_{l} w_{m} \operatorname{Tr}\left( \left[ \Pi_{j}, \widetilde{\Pi}_{k}(t) \right]^{\dagger} \left[ \Pi_{l}, \widetilde{\Pi}_{m}(t) \right] \right).
\end{align*}
Then, extracting the \(j = l\) and \(k = m\) terms, we have,
\begin{align}
C_{V,W}(t) =  \frac{1}{d} \sum\limits_{j,k} \left| v_{j} \right|^{2} \left| w_{k} \right|^{2} \operatorname{Tr}\left( \left[ \Pi_{j}, \widetilde{\Pi}_{k}(t) \right]^{\dagger} \left[ \Pi_{j}, \widetilde{\Pi}_{k}(t) \right] \right) + \frac{1}{d} \sum\limits_{j \neq l, k \neq m} v_{j}^{*} w_{k}^{*} v_{l} w_{m} \operatorname{Tr}\left( \left[ \Pi_{j}, \widetilde{\Pi}_{k}(t) \right]^{\dagger} \left[ \Pi_{l}, \widetilde{\Pi}_{m}(t) \right] \right).
\end{align}
Since, \(V,W\) are unitary, we have, \(\left| v_{j} \right|^2 = 1 = \left| w_{j} \right|^2 ~~\forall j \in \{ 1,2, \cdots, d \}\). Therefore,
\begin{align}
C_{V,W}(t) =  \frac{1}{d} \sum\limits_{j,k} \operatorname{Tr}\left( \left[ \Pi_{j}, \widetilde{\Pi}_{k}(t) \right]^{\dagger} \left[ \Pi_{j}, \widetilde{\Pi}_{k}(t) \right] \right) + \frac{1}{d} \sum\limits_{j \neq l, k \neq m} v_{j}^{*} w_{k}^{*} v_{l} w_{m} \operatorname{Tr}\left( \left[ \Pi_{j}, \widetilde{\Pi}_{k}(t) \right]^{\dagger} \left[ \Pi_{l}, \widetilde{\Pi}_{m}(t) \right] \right).
\end{align}
Then, recalling \cref{eq:extremal-cgp-defn}, we have,
\begin{align}
C_{V,W}(t) =  2 \mathfrak{C}_{\mathbb{B}_{V}} \left( \mathcal{U}_{t} \circ \mathcal{V}_{\mathbb{B}_{V} \rightarrow \mathbb{B}_{W}}  \right) + \frac{1}{d} \sum\limits_{j \neq l, k \neq m} v_{j}^{*} w_{k}^{*} v_{l} w_{m} \operatorname{Tr}\left( \left[ \Pi_{j}, \widetilde{\Pi}_{k}(t) \right]^{\dagger} \left[ \Pi_{l}, \widetilde{\Pi}_{m}(t) \right] \right),
\end{align}
where \(\mathcal{V}_{\mathbb{B}_{V} \rightarrow \mathbb{B}_{W}}\) is the intertwiner connecting the bases \(\mathbb{B}_{V}\) to \(\mathbb{B}_{W}\) as \(\mathcal{V}_{\mathbb{B}_{V} \rightarrow \mathbb{B}_{W}} \left( \Pi_{j} \right) = \widetilde{\Pi}_{j} ~~\forall j \in \{ 1,2, \cdots,d \}\). 

Next, we would like to simplify the second term of the summation. For this, note that 
\begin{align}
&\operatorname{Tr}\left( \left[ \Pi_{j}, \widetilde{\Pi}_{k}(t) \right]^{\dagger} \left[ \Pi_{l}, \widetilde{\Pi}_{m}(t) \right] \right) =  \operatorname{Tr}\left( \left\{ \widetilde{\Pi}_{k} \Pi_{j} - \Pi_{j} \widetilde{\Pi}_{k} \right\} \left\{  \Pi_{l} \widetilde{\Pi}_{m} - \widetilde{\Pi}_{m} \Pi_{l} \right\} \right) \\
&= \operatorname{Tr}\left( \widetilde{\Pi}_{k} \Pi_{j}  \Pi_{l} \widetilde{\Pi}_{m} \right) - \operatorname{Tr}\left( \widetilde{\Pi}_{k} \Pi_{j} \widetilde{\Pi}_{m} \Pi_{l}  \right) -  \operatorname{Tr}\left( \Pi_{j} \widetilde{\Pi}_{k} \Pi_{l} \widetilde{\Pi}_{m} \right) +  \operatorname{Tr}\left( \Pi_{j} \widetilde{\Pi}_{k}  \widetilde{\Pi}_{m} \Pi_{l}  \right) \\
&= \delta_{km} \delta_{jl} \operatorname{Tr}\left( \Pi_{j} \widetilde{\Pi}_{k} \right) - \operatorname{Tr}\left( \widetilde{\Pi}_{k} \Pi_{j} \widetilde{\Pi}_{m} \Pi_{l}  \right) -  \operatorname{Tr}\left( \Pi_{j} \widetilde{\Pi}_{k} \Pi_{l} \widetilde{\Pi}_{m} \right) + \delta_{jl} \delta_{km} \operatorname{Tr}\left( \Pi_{j} \widetilde{\Pi}_{k} \right) \\
&= 2  \delta_{km} \delta_{jl} \operatorname{Tr}\left( \Pi_{j} \widetilde{\Pi}_{k} \right) - 2 \mathfrak{Re} \left\{ \operatorname{Tr}\left( \widetilde{\Pi}_{k} \Pi_{j} \widetilde{\Pi}_{m} \Pi_{l}  \right) \right\}.  
\end{align}

The summation indices for the second term are \(j \neq l\) OR \(k \neq m\), which has three possibilities: \(j \neq l \text{ and } k \neq m\), \(j \neq l \text{ but } k =m\), and finally, \(j=l \text{ but } k \neq m\). In each case, the product of delta functions, \(\delta_{km} \delta_{jl}\) vanishes and we are left with the second term only. Therefore,
\begin{align}
  C_{V,W}(t) = 2 \mathfrak{C}_{\mathbb{B}_{V}} \left( \mathcal{U}_{t} \circ \mathcal{V}_{\mathbb{B}_{V} \rightarrow \mathbb{B}_{W}}  \right) - \frac{2}{d} \mathfrak{Re} \left\{ \sum\limits_{j \neq l, k \neq m}  v_{j}^{*} w_{k}^{*} v_{l} w_{m} \operatorname{Tr}\left(  \widetilde{\Pi}_{k}(t) \Pi_{j} \widetilde{\Pi}_{m}(t) \Pi_{l} \right)  \right\},
\end{align}
where we emphasize that the indices of the summation have the three possibilities listed above.

The relation between \(F_{V,W}(t)\) and \(\mathfrak{C}_{\mathbb{B}_{V}}\) is obtained simply by using \(C_{V,W}(t) = 2 \left( 1 - \mathfrak{Re} \left\{ F_{V,W}(t)  \right\} \right)\). This completes the proof.
\end{proof}

\subsection*{Proof of \autoref{eq:otoc-avg-uniform-diagonal-unitaries}, \autoref{eq:otoc-avg-2-coherence}, and \autoref{eq:otoc-avg-masa}}
\label{sec:proof-otoc-averages}
Let's start with \autoref{eq:otoc-avg-masa}. We have two unitaries \(V, W\) and two bases \(\mathbb{B}, \widetilde{\mathbb{B}}\). Then, the following Haar-averaged squared commutator is proportional to the (squared) distance in the Grassmannian between the MASAs associated to the bases \(\mathbb{B}, \mathbb{\widetilde{B}}\). That is,
\begin{align}
  \left\langle \left\Vert \left[ \mathcal{D}_{\mathbb{B}}(V), \mathcal{D}_{\widetilde{\mathbb{B}}}(W) \right]  \right\Vert_{2}^{2}   \right\rangle_{V,W \in \mathrm{Haar}} = \frac{1}{d^{2}} D^{2}(\mathcal{A}_{\mathbb{B}}, \mathcal{A}_{\widetilde{\mathbb{B}}})
\end{align}

We start by expanding \(\left\Vert \left[ \mathcal{D}_{\mathbb{B}}(V), \mathcal{D}_{\widetilde{\mathbb{B}}}(W) \right]  \right\Vert_{2}^{2}\) with \(v_{\alpha} \coloneqq \operatorname{Tr}\left( \Pi_{\alpha} V \right), w_{\beta} \coloneqq \operatorname{Tr}\left( \widetilde{\Pi}_{\beta} W \right)\). Then,
\begin{align}
\left\Vert \left[ \mathcal{D}_{\mathbb{B}}(V), \mathcal{D}_{\widetilde{\mathbb{B}}}(W) \right]  \right\Vert_{2}^{2} = \left\Vert \sum\limits_{\alpha, \beta}^{d} v_{\alpha} w_{\beta} \left[ \Pi_{\alpha}, \widetilde{\Pi}_{\beta} \right]  \right\Vert_{2}^{2} = \sum\limits_{\alpha, \beta, \gamma, \eta}^{d} v_{\alpha}^{*} w_{\beta}^{*} v_{\gamma} w_{\eta} \operatorname{Tr}\left( \left[ \Pi_{\alpha}, \widetilde{\Pi}_{\beta} \right]^{\dagger} \left[ \Pi_{\gamma}, \widetilde{\Pi}_{\eta} \right]   \right). 
\end{align}

Now, \(v_{\alpha}^{*} v_{\gamma} = \operatorname{Tr}\left( \Pi_{\alpha} V^{\dagger} \right) \operatorname{Tr}\left( \Pi_{\gamma} V \right) = \operatorname{Tr}\left( \left( \Pi_{\alpha} \otimes \Pi_{\gamma} \right) \left( V^{\dagger} \otimes V \right) \right)\) and using the lemma,
\begin{align}
  \int_{\mathrm{Haar}} dA A^{\dagger} \otimes A = \frac{S}{d}, \text{ where S is the SWAP operator,}
\end{align}
we have
\begin{align}
  \left\langle v_{\alpha}^{*} v_{\gamma} \right\rangle_{V \in \mathrm{Haar}} = \int_{\mathrm{Haar}} dV \operatorname{Tr}\left( \left( \Pi_{\alpha} \otimes \Pi_{\gamma} \right) \left( V^{\dagger} \otimes V \right) \right) = \frac{1}{d}\operatorname{Tr}\left( \left( \Pi_{\alpha} \otimes \Pi_{\gamma} \right) S \right) = \frac{1}{d}\operatorname{Tr}\left( \Pi_{\alpha} \Pi_{\gamma} \right) = \frac{1}{d}\delta_{\alpha, \gamma}.
\end{align}
Similarly, for \(\left\langle w_{\beta}^{*} w_{\eta} \right\rangle_{W \in \mathrm{Haar}} = \frac{1}{d}\delta_{\beta, \eta}\). Putting everything together, we have,
\begin{align}
   \left\langle \left\Vert \left[ \mathcal{D}_{\mathbb{B}}(V), \mathcal{D}_{\widetilde{\mathbb{B}}}(W) \right]  \right\Vert_{2}^{2}   \right\rangle_{V,W \in \mathrm{Haar}} = \frac{1}{d^{2}} \sum\limits_{\alpha, \beta, \gamma, \eta}^{d} \delta_{\alpha, \gamma} \delta_{\beta, \eta} \operatorname{Tr}\left( \left[ \Pi_{\alpha}, \widetilde{\Pi}_{\beta} \right]^{\dagger} \left[ \Pi_{\gamma}, \widetilde{\Pi}_{\eta} \right]   \right) = \frac{1}{d^{2}} \sum\limits_{\alpha, \beta}^{d} \left\Vert \left[ \Pi_{\alpha}, \widetilde{\Pi}_{\beta} \right]  \right\Vert_{2}^{2} = \frac{1}{d^{2}} D^{2}(\mathcal{A}_{\mathbb{B}}, \mathcal{A}_{\widetilde{\mathbb{B}}}).
\end{align}

To prove \autoref{eq:otoc-avg-2-coherence}, \(\left\langle\left\|\left[\mathcal{D}_{\mathbb{B}}(V), \rho\right]\right\|_{2}^{2}\right\rangle_{V \in \mathrm{Haar}}=\frac{2}{d} \mathrm{c}_{\mathbb{B}}^{(2)}(\rho)\), we follow a similar sequence of arguments as above. We prove a slightly general version of the result here, where \(V\) is a unitary and \(X\) an arbitrary operator (and not necessarily a quantum state)
\begin{align}
  \left\Vert \left[ \mathcal{D}_{\mathbb{B}}(V), X \right]  \right\Vert_{2}^{2} = \sum\limits_{\alpha, \beta}^{d} v_{\alpha}^{*} v_{\beta} \operatorname{Tr}\left( \left[ \Pi_{\alpha}, X \right]^{\dagger} \left[ \Pi_{\beta}, X \right]   \right).
\end{align}
As above, \( \left\langle v_{\alpha}^{*} v_{\beta} \right\rangle_{V \in \mathrm{Haar}} = \frac{1}{d} \delta_{\alpha, \beta}\). Therefore,
\begin{align}
  \left\langle\left\|\left[\mathcal{D}_{\mathbb{B}}(V), X\right]\right\|_{2}^{2}\right\rangle_{V \in \mathrm{Haar}}= \frac{1}{d} \sum\limits_{\alpha}^{d} \left\Vert \left[ \Pi_{\alpha}, X \right]  \right\Vert_{2}^{2} = \frac{2}{d} \mathrm{c}_{\mathbb{B}}^{(2)}(X).
\end{align}

And finally, to prove \autoref{eq:otoc-avg-uniform-diagonal-unitaries}, \(\left\langle \left\Vert \left[ V,W(t) \right]  \right\Vert_{2}^{2}  \right\rangle_{\theta} = 2d \mathfrak{C}^{}_{\mathbb{B}_{V}} \left( \mathcal{U}_{t} \circ \mathcal{V}_{\mathbb{B}_{V} \rightarrow \mathbb{B}_{W}} \right)\), we proceed as above and note that the key step is \(\left\langle v_{\alpha}^{*} v_{\gamma} \right\rangle_{\theta} = \left\langle e^{i(\theta_{\gamma} - \theta_{\alpha})} \right\rangle_{\theta} = \frac{1}{2 \pi} \int\limits_{0}^{2 \pi} e^{i \left( \theta_{\gamma} -\theta_{\alpha} \right)} d\theta = \delta_{\alpha, \gamma}\). And similarly for \(\left\langle w_{\beta}^{*} w_{\eta} \right\rangle_{\theta} = \delta_{\beta, \eta}\). Putting everything together, we then have the desired result.

\subsection*{Proof of \autoref{thm:cgp-spectralff}}
\cgpspectralff*
\begin{proof}
Let $\hat{S}$ be the SWAP operator defined in \cref{eq:swap-defn}. Then,

\begin{align}
\mathfrak{C}_{\mathbb{B}} \left( e^{-i H t} \right) & \leq 1 - \frac{1}{d} \sum\limits_{j} \left[\operatorname{Tr}{\left( P_{j} U P_{j} U^{\dagger} \right)}\right]^{2}\\
&= 1 - \frac{1}{d} \sum\limits_{j} \operatorname{Tr}\left( {P_{j}}^{\otimes 2} U^{\otimes 2} {P_{j}}^{\otimes 2} {U^{\dagger}}^{\otimes 2} \right)\\
&= 1 - \frac{1}{d} \sum\limits_{j} \operatorname{Tr}\left( \hat{S} P_{j}^{\otimes 4} \left( U^{\otimes 2} \otimes  {U^{\dagger}}^{\otimes 2} \right) \right)\\
&= 1 - \frac{1}{d} \sum\limits_{j} \sum\limits_{k,l,m,n}^{} \operatorname{Tr}\left( P_{j}^{\otimes 4} V^{\otimes 4} \left( P_{k} \otimes P_{l} \otimes P_{m} \otimes P_{n} \right) \right) \times e^{-i \left( E_{k} + E_{l} - E_{m} - E_{n} \right) t},
\end{align}
where in the first inequality, we have dropped the off-diagonal terms in the CGP, that is, using $\mathfrak{C}_{\mathbb{B}} \left( U\right) = 1 - 1/d \sum_{j,k} \left[\operatorname{Tr}\left( \Pi_{j} U \Pi_{k} U^{\dagger}  \right)\right]^{2}$ (see \cref{eq:cgp-xmatrix-formula}) and only keeping the terms with $j=k$. In the second line, we have simply re-expressed the trace by using $\left[\operatorname{Tr}(A)\right]^2 = \operatorname{Tr}(A \otimes A)$. And, in the last line we have plugged in \(U = \sum\limits_{k} e^{-i E_{k} t} V P_{k} V^{\dagger}\), where $\mathcal{V}(\cdot) = V(\cdot) V^{\dagger}$ is the unitary intertwiner connecting the Hamiltonian eigenbasis with the basis $\mathbb{B}$.

In the following we will make a simple change of notation for both convenience and consistency with other works: \(E_{j} \mapsto \lambda_{j}\). Now, recall that for GUE, we have, \(P(H) \propto \exp \left( - \frac{d}{2} \operatorname{Tr}\left( H^{2} \right) \right)\), therefore,
\begin{align}
  \int dH P(H) = \int d \lambda P(\lambda) \times \int dV,
\end{align}
where the average decomposes into the eigenvalues and the eigenvectors. Recall that \(V\) is Haar-distributed. Then,
\begin{align}
  P(\lambda) = c \left| \Delta(\lambda) \right|^2 e^{- \frac{d}{2} \sum\limits_{j}^{} \lambda_{j}^2}, \text{ where } \Delta(\lambda) \equiv \prod_{1 \leq j<k \leq d}\left(\lambda_{j}-\lambda_{k}\right) \text{ is the Vandermonde matrix.}
\end{align}

Then,
\begin{align}
  \left\langle \mathfrak{C}_{\mathbb{B}}(e^{-iHt}) \right\rangle_{\mathrm{GUE}} &\leq 1 - \frac{1}{d} \sum\limits_{j}^{} \sum\limits_{k,l,m,n}^{} \left( \int d \lambda P(\lambda) e^{-i \left( \lambda_{k} + \lambda_{l} - \lambda_{m} - \lambda_{l}\right) t} \times \int dV \operatorname{Tr}\left( P_{j}^{\otimes 4} \mathcal{V}^{\otimes 4} \left( P_{k} \otimes P_{l} \otimes P_{m} \otimes P_{n} \right) \right) \right).
\end{align}
Notice that if \(\lambda_{j} = \lambda_{k}\) for any \(j,k\) then \(\Delta(\lambda) = 0\). Therefore, in the summation \(\sum\limits_{k,l,m,n}^{}\), we only need to consider \(\lambda_{k} \neq \lambda_{l} \neq \lambda_{m} \neq \lambda_{n}\). Then, one can show that, \(\int dV \operatorname{Tr}\left( P_{j}^{\otimes 4} \mathcal{V}^{\otimes 4} \left( P_{k} \otimes P_{l} \otimes P_{m} \otimes P_{n} \right) \right) = \frac{1}{d (d+1) (d+2) (d+3)}\); see Ref.~\cite{2011arXiv1109.4244P} for integrals of this form. Therefore,
\begin{align}
  \left\langle \mathfrak{C}_{\mathbb{B}} \left( e^{-iHt} \right) \right\rangle_{\mathrm{GUE}} \leq 1 - \frac{1}{d (d+1) (d+2) (d+3)} \underbrace{ \sum\limits_{k,l,m,n} \int e^{-i \left( \lambda_{k} + \lambda_{l} - \lambda_{m} - \lambda_{n} \right) t} P(\lambda) d\lambda}_{\mathcal{R}_{4}}.
\end{align}

To see that the bound is tight for short times, notice that in order to establish the connection to the spectral form factor, the first step in the proof is the inequality, \(\mathfrak{C}_{\mathbb{B}}\left(e^{-i H t}\right) \leq 1-\frac{1}{d} \sum_{j}\left[\operatorname{Tr}\left(P_{j} U P_{j} U^{\dagger}\right)\right]^{2}\) which is obtained by ignoring the off-diagonal terms in the CGP, \(\mathfrak{C}_{\mathbb{B}}(U)=1-1 / d \sum_{j, k}\left[\operatorname{Tr}\left(\Pi_{j} U \Pi_{k} U^{\dagger}\right)\right]^{2}\) (we drop the \(j \neq k\) terms). Now, for short times, let $t = O(\epsilon)$, then, the contribution from the off-diagonal terms scales as $O(\epsilon^4) \ll 1$, making the bound tight.
\end{proof}

\subsection*{Proof of \autoref{thm:haar-avg-otoc}}
\haaravgotoc*
\begin{proof}
First, note that using \autoref{thm:otoc-cgp-connection}, we can Haar-average the CGP and the ``off-diagonal'' terms independently. Following Refs.~\cite{zanardiCoherencegeneratingPowerQuantum2017, styliaris_quantum_2019-1}, we have that 
\begin{align}
\left\langle \mathfrak{C} _{\mathbb{B}} \left( \mathcal{U} \right) \right\rangle_{\mathrm{Haar}} = \frac{\left( d-1 \right)}{\left( d+1 \right)}.
\end{align}

Now, for the ``off-diagonal'' term, let us look at terms of the form \(\operatorname{Tr}\left( \widetilde{\Pi}_{k}(t) \Pi_{j} \widetilde{\Pi}_{m}(t) \Pi_{l} \right)\). Let $\mathcal{H} = \mathcal{H}_{A} \otimes \mathcal{H}_{A'}$, where $\mathcal{H}_{A} \cong \mathcal{H}_{A'}$, that is, we take two copies of the Hilbert space. The SWAP operator acting on this doubled space is defined as,
\begin{align}
\label{eq:swap-defn}
    \hat{S} = \sum_{i, j}|i\rangle_{A}\langle j|\otimes| j\rangle_{A'}\langle i|.
\end{align}
It is easy to show that $\operatorname{Tr}\left( X Y \right) = \operatorname{Tr}\left( \hat{S} X \otimes Y \right)$, which we use in the following (and variants thereof). Then,
\begin{align}
  & \operatorname{Tr}\left( \widetilde{\Pi}_{k}(t) \Pi_{j} \widetilde{\Pi}_{m}(t) \Pi_{l} \right) = \operatorname{Tr}\left( \Pi_{l} \widetilde{\Pi}_{k}(t) \otimes  \Pi_{j} \widetilde{\Pi}_{m}(t) \hat{S} \right) \\
  & = \operatorname{Tr}\left( \left( \Pi_{l} \otimes  \Pi_{j} \right) \left( \widetilde{\Pi}_{k}(t) \otimes  \widetilde{\Pi}_{m}(t) \right) \hat{S} \right)  = \operatorname{Tr}\left( \left( \Pi_{l} \otimes  \Pi_{j} \right) \mathcal{U}_{t}^{\otimes 2} \left( \widetilde{\Pi}_{k} \otimes  \widetilde{\Pi}_{m} \right) \hat{S} \right) .
\end{align}

Then, to Haar-average the above term, we collect a few results,
\begin{align}
  \left\langle \mathcal{U}^{\otimes 2} \left( X \right) \right\rangle_{\mathrm{Haar}} = \frac{1}{2} \frac{\left( \mathbb{I} + \hat{S} \right)}{d(d+1)} \operatorname{Tr}\left( \left( \mathbb{I} + \hat{S} \right) X \right) + \frac{1}{2} \frac{\left( \mathbb{I} - \hat{S} \right)}{d(d-1)} \operatorname{Tr}\left( \left( \mathbb{I} - \hat{S} \right) X \right) .
\end{align}

Now, taking \(X  = \widetilde{\Pi}_{k} \otimes \widetilde{\Pi}_{m}\), we have, 
\begin{align}
  \operatorname{Tr}\left( \left( \mathbb{I} \pm \hat{S} \right)  \widetilde{\Pi}_{k} \otimes \widetilde{\Pi}_{m}  \right) = 1 \pm \delta_{km}. 
\end{align}

Then,
\begin{align}
&  \left\langle \operatorname{Tr}\left( \widetilde{\Pi}_{k}(t) \Pi_{j} \widetilde{\Pi}_{m}(t) \Pi_{l} \right) \right\rangle_{\mathrm{Haar}} \\
&  = \operatorname{Tr}\left( \left( \Pi_{l} \otimes  \Pi_{j} \right) \left\langle \mathcal{U}_{t}^{\otimes 2} \left( \widetilde{\Pi}_{k} \otimes  \widetilde{\Pi}_{m} \right) \right\rangle_{\mathrm{Haar}} \hat{S} \right).
\end{align}
Using, \(\left( \mathbb{I} \pm \hat{S} \right) \hat{S} = \left( \hat{S} \pm I \right)\), we have, \(\operatorname{Tr}\left( \left( \Pi_{l} \otimes \Pi_{j} \right) \left( \hat{S} \pm \mathbb{I} \right) \right) = \delta_{lj} \pm 1\).

Putting everything together, and recalling that the ``off-diagonal'' term has the form \(\sum\limits_{j \neq l, k \neq m}  v_{j}^{*} w_{k}^{*} v_{l} w_{m} \operatorname{Tr}\left(  \widetilde{\Pi}_{k}(t) \Pi_{j} \widetilde{\Pi}_{m}(t) \Pi_{l} \right)\), where, we recall that the indices have the form \(j \neq l\) OR \(k \neq m\).

Then, for different choices of indices, we have,
\begin{align}
&\text{For } j \neq l \text{ and } k \neq m: \quad \frac{1}{2d(d+1)} - \frac{1}{2d(d-1)} = -\frac{1}{d(d^{2}-1)}, \\
&\text{For } j \neq l \text{ and } k = m: \quad \frac{1}{d(d+1)}, \\
&\text{For } j = l \text{ and } k \neq m: \quad \frac{1}{d(d+1)}.
\end{align}

Combining with the phases, we have,
\begin{align}
  & \frac{2}{d(d^{2}-1)} \mathfrak{Re} \left\{ \sum\limits_{j \neq l \text{ and } k \neq m} v_{j}^{*} w_{k}^{*} v_{l} w_{m}  \right\} - \frac{2}{d(d+1)} \mathfrak{Re} \left\{ \sum\limits_{j \neq l, k = m} \left| w_{k} \right|^2  v_{j}^{*} v_{l} \right\} - \frac{2}{d(d+1)} \mathfrak{Re} \left\{ \sum\limits_{j = l, k \neq m} \left| v_{k} \right|^2  w_{k}^{*} w_{m} \right\}
\end{align}

Then, since \(V,W\) are unitaries, \(\left| w_{k} \right|^2 = 1 = \left| v_{k} \right|^2 ~~\forall k \in \{ 1,2, \cdots,d \}\). Therefore, the above term becomes,
\begin{align}
  & \frac{2}{d(d^{2}-1)} \mathfrak{Re} \left\{ \sum\limits_{j \neq l \text{ and } k \neq m} v_{j}^{*} w_{k}^{*} v_{l} w_{m}  \right\} - \frac{2}{d(d+1)} \mathfrak{Re} \left\{ \sum\limits_{j \neq l}  v_{j}^{*} v_{l} + \sum\limits_{k \neq m}^{}  w_{k}^{*} w_{m} \right\}
\end{align}

Collecting the CGP and ``off-diagonal'' terms together, we have the desired result.

\end{proof}

\subsection*{Proof of \autoref{thm:short-time-cgp}}
\shorttimecgp*

\begin{proof}
Using Proposition 1 of Ref.~\cite{styliaris_quantum_2019-1}, we have,
\begin{align}
  \mathfrak{C}_{\mathbb{B}}(\mathcal{U}_{t}) &= 1 - \frac{1}{d} \sum\limits_{j,k}^{} \operatorname{Tr}\left( \Pi_{j} \Pi_{k} (t) \Pi_{j} \Pi_{k} (t) \right), \text{ where } \Pi_{j}(t) \equiv \mathcal{U}_{t}(\Pi_{j}), \\
&= 1 - \frac{1}{d} \sum\limits_{j,k}^{} \operatorname{Tr}\left( \left( \Pi_{j} \otimes \Pi_{j} \right) \mathcal{U}_{t}^{\otimes 2} \left( \Pi_{k} \otimes  \Pi_{k} \right) \hat{S}\right),
\end{align}
where \(\hat{S}\) is the SWAP operator on the doubled Hilbert space defined in \cref{eq:swap-defn}.

Now, recall that the time evolution superoperator can be expanded at short times as,
\begin{align}
   \mathcal{U}_{t} \approx \mathcal{I} - i \mathcal{H}t - \frac{1}{2} \mathcal{H}^{2} t^{2} + \cdots 
\end{align}
where \(\mathcal{I}\) is the Identity superoperator and \(\mathcal{H}(X) \equiv \left[ H,X \right]\). Therefore, 
\begin{align}
   \mathcal{U}_{t}^{\otimes 2} \approx \mathcal{I} \otimes \mathcal{I} - it \left( \mathcal{H} \otimes \mathcal{I} + \mathcal{I} \otimes \mathcal{H} \right)  - \frac{t^{2}}{2} \left( \mathcal{H} \otimes \mathcal{I} + \mathcal{I} \otimes \mathcal{H} \right)^{2}  + \cdots 
\end{align}

Let us consider the various terms in the short-time expansion of the doubled evolution.\\
\uline{Zeroth order:} 
\begin{align}
 & \operatorname{Tr}\left( \left( \Pi_{j} \otimes \Pi_{j} \right) \mathcal{I}^{\otimes 2} \left( \Pi_{k} \otimes  \Pi_{k} \right) \hat{S}\right) = \operatorname{Tr} \left( \Pi_{j} \Pi_{k} \otimes \Pi_{j} \Pi_{k} \hat{S} \right)  =\delta_{jk} \delta_{jk}.\\
&\implies \frac{1}{d} \sum\limits_{j,k}^{} \operatorname{Tr}\left( \cdots \right) = 1.
\end{align}
Therefore, the zeroth order term is one.

\uline{First order:}
\begin{align}
&  \operatorname{Tr}\left( \Pi_{j}^{\otimes 2} \left( \mathcal{H} \otimes \mathcal{I} + \mathcal{I} \otimes \mathcal{H} \right) \Pi_{k}^{\otimes 2} \hat{S} \right) \\
&= \operatorname{Tr}\left( \Pi_{j}^{\otimes 2} \left( \mathcal{H} \otimes \mathcal{I}\right) \Pi_{k}^{\otimes 2} \hat{S} \right) + \operatorname{Tr}\left( \Pi_{j}^{\otimes 2} \left( \mathcal{I} \otimes \mathcal{H}\right) \Pi_{k}^{\otimes 2} \hat{S} \right) \\
& \text{Let us consider the first term in the summation:} \quad = \operatorname{Tr}\left( \Pi_{j} \mathcal{H}(\Pi_{k}) \otimes  \Pi_{j} \Pi_{k} \hat{S} \right) \\
&= \operatorname{Tr}\left( \Pi_{j} \mathcal{H}(\Pi_{k}) \Pi_{j} \Pi_{k} \right) = \delta_{jk} \operatorname{Tr}\left( \Pi_{j} \mathcal{H}(\Pi_{k}) \right) = \operatorname{Tr}\left( \Pi_{j} \mathcal{H}(\Pi_{j}) \right) = 0.
\end{align}
The same holds for the second term in the summation above. Therefore, the linear term is zero.

\uline{Second order:}
\begin{align}
 & \operatorname{Tr}\left( \Pi_{j}^{\otimes 2} \left( \mathcal{H} \otimes \mathcal{I} + \mathcal{I} \otimes \mathcal{H} \right)^{\otimes 2} \Pi_{k}^{\otimes 2} \hat{S} \right) = \operatorname{Tr}\left( \Pi_{j}^{\otimes 2} \left( \mathcal{H}^{2} \otimes \mathcal{I} + \mathcal{I} \otimes \mathcal{H}^{2} + 2 \mathcal{H} \otimes \mathcal{H} \right) \Pi_{k}^{\otimes 2} \hat{S} \right) \\
& = 2 \operatorname{Tr}\left( \Pi_{j}^{\otimes 2} \left( \mathcal{H} \otimes \mathcal{I} + \mathcal{H} \otimes \mathcal{H} \right) \Pi_{k}^{\otimes 2} \hat{S} \right),
\end{align}
where, the last equality follows from a simple symmetry argument.

Note that \(\mathcal{H}^{2} \otimes \mathcal{I} \left( X \otimes Y \right) = [H,[H, X]] \otimes Y\) and \(\mathcal{H} \otimes \mathcal{H} \left( X \otimes Y \right) = [H, X] \otimes[H, Y]\). Therefore,
\begin{align}
&  \mathcal{H}^2 \otimes \mathcal{I} (\Pi_{k} \otimes \Pi_{k}) = \left[ H, \left[ H, \Pi_{k} \right]  \right] \otimes \Pi_{k} = \left\{ H \left( H \Pi_{k} - \Pi_{k} H \right) - \left( H \Pi_{k} - \Pi_{k} H \right) H \right\} \otimes \Pi_{k} \\
&= \left(H^{2} \Pi_{k} -2 H \Pi_{k} H + \Pi_{k} H^{2}\right) \otimes \Pi_{k} .
\end{align}
Plugging this back into the trace, we have,
\begin{align}
  &\operatorname{Tr}\left( \Pi_{j}^{\otimes 2} \left( \mathcal{H}^{2} \otimes \mathcal{I} \right) \Pi_{k}^{\otimes 2} \hat{S} \right) \\
  & = \operatorname{Tr}\left( \Pi_{j}^{\otimes 2} \left(H^{2} \Pi_{k} -2 H \Pi_{k} H + \Pi_{k} H^{2}\right) \otimes \Pi_{k}  \hat{S} \right) \\
  & = \delta_{jk} \operatorname{Tr}\left( \Pi_{j} \left(H^{2} \Pi_{k} -2 H \Pi_{k} H + \Pi_{k} H^{2}\right)  \right) \\
  & = \operatorname{Tr}\left( \Pi_{j} \left(H^{2} \Pi_{k} -2 H \Pi_{k} H + \Pi_{k} H^{2}\right)  \right) \\
  & = 2 \left( \operatorname{Tr}\left( \Pi_{j} H^{2} \Pi_{j} \right) - \left( \operatorname{Tr}\left( \Pi_{j} H \right) \right)^2 \right) \\
  & = 2 \mathrm{var}_{j}(H), \text{ where } \mathrm{var}_{j}(H) \equiv \left\langle H^{2} \right\rangle_{\Pi_{j}} - \left\langle H \right\rangle_{\Pi_{j}}^{2}. 
\end{align}

Now, we need to look at the \(\mathcal{H} \otimes \mathcal{H}\) term. 
\begin{align}
  & \operatorname{Tr}\left( \Pi_{j}^{\otimes 2} \left[ H, \Pi_{k} \right] \otimes \left[ H, \Pi_{k} \right] \hat{S}  \right) = \operatorname{Tr}\left( \Pi_{j}  \left[ H, \Pi_{k} \right]  \Pi_{j}  \left[ H, \Pi_{k} \right]  \right) \\
& = \operatorname{Tr}\left( \left( \Pi_{j} H \Pi_{k} - \Pi_{j} \Pi_{k} H \right) \left( \Pi_{j} H \Pi_{k} - \Pi_{j} \Pi_{k} H  \right) \right) = 0.
\end{align}

That is, the \(\mathcal{H} \otimes \mathcal{H}\) term is zero.

Therefore, putting everything together, we have,
\begin{align}
\frac{1}{2} \frac{d^2 \mathfrak{C}_{\mathbb{B}}(\mathcal{U}_{t})}{d {t^2} } \bigg\rvert_{t=0} = \frac{1}{d} \sum\limits_{j=1}^{d} \mathrm{var}_{j} \left( H \right).
\end{align}

\end{proof}

\subsection*{Proof of short-time growth of $k$-local commuting Hamiltonians, \autoref{eq:shorttime-k-local}}
\label{sec:proof-shorttime-examples}

To prove this we need three ingredients. First, notice that since each term in the Hamiltonian \(H^{(k)}\) commutes, we have, \(\left\Vert H^{(k)} \right\Vert_{\infty}^{} = \sum\limits_{j=1}^{L-(k-1)} \left\Vert \sigma^{x}_{j} \otimes \sigma^{x}_{j+1} \otimes \cdots \otimes \sigma^{x}_{j+(k-1)}  \right\Vert_{\infty}^{} = L-(k-1)\) where in the second equality, we have used the fact that \(\sigma^{x}_{j} \otimes \sigma^{x}_{j+1} \otimes \cdots \otimes \sigma^{x}_{j+(k-1)}\) is a unitary for each \(j\) and \(\left\Vert U \right\Vert_{\infty}^{} =1\) for all unitaries. Second, to compute \(\frac{1}{2}
\frac{d^2 \mathfrak{C}_{\mathbb{B}}(\mathcal{U}_{t})}{d {t^2} }
\bigg\rvert_{t=0}\), we can use its equality with \(\frac{1}{d} \sum\limits_{j=1}^{d} \mathrm{var}_{j} \left( H \right)\). 

Then, we note that,
\begin{align}
\frac{1}{d} \sum\limits_{j=1}^{d} \mathrm{var}_{j}(H) = \frac{1}{d} \left( \sum\limits_{j=1}^{d} \operatorname{Tr}\left( H^{2} \Pi_{j} \right) - \sum\limits_{j=1}^{d} \left( \operatorname{Tr}\left( H \Pi_{j} \right) \right)^{2} \right) = \frac{1}{d} \left( \operatorname{Tr}\left( H^{2} \right) - \sum\limits_{j=1}^{d} H_{jj}^{2} \right),
\end{align}
where \(H_{jj} = \langle j | H |  j \rangle\). 

Third, for each \(H^{(k)}\) notice that, 
\begin{align}
\operatorname{Tr}\left[ \left( H^{(k)} \right)^{2} \right] = \sum\limits_{\alpha,\beta=1}^{L-(k-1)} \operatorname{Tr}\left[ \left( \sigma^{x}_{\alpha} \otimes \sigma^{x}_{\alpha+1} \otimes \cdots \otimes \sigma^{x}_{\alpha+(k-1)} \right)  \left( \sigma^{x}_{\beta} \otimes \sigma^{x}_{\beta+1} \otimes \cdots \otimes \sigma^{x}_{\beta+(k-1)} \right)\right].
\end{align}
It is easy to see that since \(\operatorname{Tr}\left[ \sigma^{x} \right] = 0\) the above trace is only nonzero if \(\alpha = \beta\) and therefore, we have, using \(\operatorname{Tr}\left[ \mathbb{I} \right] = d\),
\begin{align}
\operatorname{Tr}\left[ \left( H^{(k)} \right)^{2} \right] = d \left( L-(k-1) \right)
\end{align}

Moreover, notice that, \(H_{jj} = 0 ~~\forall j\) (since \(\sigma^{x}\)'s flip the spins). Therefore,
\begin{align}
\frac{1}{d} \sum\limits_{j=1}^{d} \mathrm{var}_{j}(H) = L-(k-1)
\end{align}

And, finally, normalizing this with the (squared) operator norm of the Hamiltonian, we have the desired result.
\end{document}